\documentclass{article}

\usepackage{fullpage}
\usepackage{amsthm}
\usepackage{amsmath}
\usepackage{amsfonts}
\usepackage{amssymb}
\usepackage{stmaryrd}
\usepackage{verbatim}
\usepackage{epic}

\usepackage{amscd}
\usepackage[dvips]{epsfig}
\usepackage{wrapfig}
\usepackage{enumitem}
\usepackage{pst-tree}
\usepackage{bussproofs}

\usepackage{cmll}

\newcommand{\lone}{1}
	\newcommand{\ltens}{\otimes}

	\newcommand{\lpar}{\parr}
	
\newcommand{\target}[1]{t_{#1}}

\theoremstyle{plain}

\newtheorem*{notations}{Notations}
\newtheorem{definition}{Definition}
\newtheorem{theorem}{Theorem}
\newtheorem{example}{Example}
\newtheorem{prop}[theorem]{Proposition}
\newtheorem{lem}[theorem]{Lemma}
\newtheorem{fact}[theorem]{Fact}
\newtheorem{rem}{Remark}

\newcommand{\Nat}{\ensuremath{\mathbb{N}}}
\newcommand{\atoms}[1]{\textit{At}(#1)}
\newcommand{\nontrivialconnected}[3]{\mathcal{S}_{#1}^{#3}(#2)}
\newcommand{\connectedcomponents}[2]{\mathcal{C}^{#2}(#1)}
\newcommand{\size}[1]{\textit{size}(#1)}
\newcommand{\leftwires}[1]{\mathcal{L}(#1)}
\newcommand{\labelofcell}[1]{l_{#1}}
\newcommand{\taylor}[2]{\mathcal{T}(#1)[#2]}
\newcommand{\criticalports}[3]{\mathcal{K}_{#2, #3}(#1)}
\newcommand{\groundof}[1]{\mathcal{G}(#1)}
\newcommand{\finitemultisets}[1]{\mathcal{M}_\textit{fin}(#1)}
\newcommand{\finitesequences}[1]{{#1}^{< \infty}}
\newcommand{\cosize}[1]{\textit{cosize}(#1)}
\newcommand{\depthof}[1]{\textit{depth}(#1)}
\psset{treemode=U,levelsep=5mm}

\def\restriction#1#2{\mathchoice
              {\setbox1\hbox{${\displaystyle #1}_{\scriptstyle #2}$}
              \restrictionaux{#1}{#2}}
              {\setbox1\hbox{${\textstyle #1}_{\scriptstyle #2}$}
              \restrictionaux{#1}{#2}}
              {\setbox1\hbox{${\scriptstyle #1}_{\scriptscriptstyle #2}$}
              \restrictionaux{#1}{#2}}
              {\setbox1\hbox{${\scriptscriptstyle #1}_{\scriptscriptstyle #2}$}
              \restrictionaux{#1}{#2}}}
\def\restrictionaux#1#2{{#1\,\smash{\vrule height .8\ht1 depth .85\dp1}}_{\,#2}} 

\newcommand{\dom}[1]{\textsf{dom}(#1)}

\newcommand{\im}[1]{\textsf{im}(#1)}
\newcommand{\emptysequence}{\varepsilon}
\newcommand{\typesoflinks}{\mathfrak{T}}
\newcommand{\tens}{\otimes}
\newcommand{\one}{1}
\newcommand{\bottom}{\bot}
\newcommand{\cod}{\oc}
\newcommand{\contr}{\wn}

\newcommand{\portsatzero}[1]{\mathcal{P}_0(#1)}
\newcommand{\wiresatzero}[1]{\mathcal{W}_0(#1)}

\newcommand{\arity}[1]{{\textit{a}}_{#1}}
\newcommand{\ports}[1]{\mathcal{P}(#1)}
\newcommand{\conclusions}[1]{\mathcal{P}^{\textsf{f}}(#1)}
\newcommand{\axioms}[1]{\mathcal{A}(#1)}
\newcommand{\multiplicativeports}[1]{\mathcal{P}^\textit{m}(#1)}
\newcommand{\multiplicativeportsatzero}[1]{\mathcal{P}_0^\textit{m}(#1)}

\newcommand{\sm}[1]{\llbracket #1 \rrbracket}
\newcommand{\Card}[1]{\textsf{Card}\left( #1 \right)}
\newcommand{\portsoftype}[2]{\mathcal{P}^{#1}(#2)}
\newcommand{\portsatzerooftype}[2]{\mathcal{P}_0^{#1}(#2)}
\newcommand{\exponentialportsatzero}[1]{\mathcal{P}_0^\textit{e}(#1)}

\newcommand{\conclusionscirc}[1]{\mathcal{P}_\circ^{\textsf{f}}(#1)}
\newcommand{\conclusionsnotcirc}[1]{\mathcal{P}_\bullet^{\textsf{f}}(#1)}

\newcommand{\boxes}[1]{\mathcal{B}(#1)}
\newcommand{\boxesatzero}[1]{\mathcal{B}_{0}(#1)}
\newcommand{\exactboxes}[2]{\mathcal{B}^{=#2}(#1)}
\newcommand{\exactboxesatzero}[2]{\mathcal{B}_{0}^{=#2}(#1)}
\newcommand{\boxesgeq}[2]{\mathcal{B}^{\geq #2}(#1)}
\newcommand{\boxesatzerogeq}[2]{\mathcal{B}_0^{\geq #2}(#1)}
\newcommand{\boxesatzerosmaller}[2]{\mathcal{B}_0^{< #2}(#1)}
\newcommand{\exponentialports}[1]{\mathcal{P}^{\textit{e}}(#1)}
\newcommand{\wires}[1]{\mathcal{W}(#1)}
\newcommand{\supp}[1]{\textit{Supp}(#1)}

\newcommand{\scalefactbis}{0.4}

\newcommand{\scalefactter}{0.33}

\newcommand{\scalefactnine}{0.3}

\newcommand{\scalefactten}{0.45}

\newcommand{\scalefacteight}{0.30}

\newcommand{\scalefactseven}{0.23}

\newcommand{\scalefactfour}{0.7}

\newcommand{\scalefactfive}{0.4}

  {\gdef\scalefactor{#1}\begin{center}\proofSkipAmount \leavevmode}%
  {\scalebox{\scalefactor}{\DisplayProof}\proofSkipAmount \end{center} }

\title{The relational model is injective for Multiplicative Exponential Linear Logic}

\author{Daniel de Carvalho}

\begin{document}

\maketitle

\begin{abstract}
We prove a completeness result for Multiplicative Exponential Linear Logic (MELL): we show that the relational model is injective for MELL proof-nets, i.e. the equality between MELL proof-nets in the relational model is exactly axiomatized by cut-elimination.
\end{abstract}

In the seminal paper by Harvey Friedman \cite{Friedman}, it has been shown that equality between simply-typed lambda terms in the full typed structure $\mathcal{M}_X$ over an infnite set $X$ is completely axiomatized by $\beta$ and $\eta$: we have $\mathcal{M}_X \vDash v = u \Leftrightarrow v \simeq_{\beta \eta} u$. A natural problem is to know whether a similar result could be obtained for Linear Logic. 

Such a result can be seen as a ``separation'' theorem. To obtain such separation theorems, it is a prerequesite to have a ``canonical'' syntax. 
When Jean-Yves Girard introduced Linear Logic (LL) \cite{ll}, he not only introduced a sequent calculus system but also ``proof-nets''. Indeed, as for LJ and LK (sequent calculus systems for intuitionnistic and classical logic, respectively), different proofs in LL sequent calculus can represent ``morally'' the same proof: proof-nets were introduced to find a unique representative for these proofs. 

The technology of proof-nets was completely satisfactory for the multiplicative fragment without units.\footnote{For the multiplicative fragment with units, it has been recently shown \cite{MLLwithunits} that, in some sense, no satisfactory notion of proof-net can exist. Our proof-nets have no jump, so they identify too many sequent calculus proofs, but not more than the relational semantics.} For proof-nets having additives, contractions or weakenings, it was easy to exhibit different proof-nets that should be identified. Despite some flaws, the discovery of proof-nets was striking. In particular, Vincent Danos proved by syntactical means in \cite{phddanos} the confluence of these proof-nets for the Multiplicative Exponential Linear Logic fragment (MELL). 
For additives, the problem to have a satisfactory notion of proof-net has been addressed in \cite{mallpn}. For MELL, a ``new syntax'' was introduced in \cite{hilbert}. In the original syntax, the following properties of the weakening and of the contraction did not hold:
\begin{itemize}
\item the associativity of the contraction;
\item the neutrality of the weakening for the contraction;
\item the contraction and the weakening as morphisms of coalgebras.
\end{itemize}
But they hold in the new syntax; at least for MELL, we got a syntax that was a good candidate to deserve to be considered as being ``canonical''. Then trying to prove that any two ($\eta$-expanded) MELL proof-nets that are equal in some denotational semantics are $\beta$-joinable has become sensible and had at least the two following motivations:
\begin{itemize}
\item to prove the canonicity of the ``new syntax'' (if we quotient more normal proof-nets, then we would identify proof-nets having different semantics);
\item to prove by semantics means the confluence (if a proof-net reduces to two cut-free proof-nets, then they have the same semantics, so they would be $\beta$-joinable, hence equal).
\end{itemize}
The problem of \emph{injectivity}\footnote{The tradition of the lambda-calculus community rather suggests the word ``completeness'' and the terminology of category theory rather suggets the word ``faithfulness'', but we follow here the tradition of the Linear Logic community.} of the denotational semantics for MELL, which is the question whether equality in the denotational semantics between ($\eta$-expanded) MELL proof-nets is exactly axiomatized by cut-elimination or not, can be seen as a study of the separation property with a semantic approach. The first work on the study of this property in the framework of proof-nets is \cite{Marco} where the
authors deal with the translation into LL of the pure $\lambda$-calculus; it has been studied more recently for the intuitionistic multiplicative
fragment of LL \cite{typedbohm} and for differential nets \cite{separationdiff}. For Parigot's $\lambda \mu$-calculus, see \cite{lmbohm} and \cite{separationsaurin}.

Finally the precise problem of  injectivity for MELL has been adressed by Lorenzo Tortora de Falco in his PhD thesis \cite{phdtortora} and in \cite{injectcoh} for the (multiset based) coherence semantics and the multiset based relational semantics. 
He gave partial results and counter-examples for the coherence semantics: the (multiset based) coherence semantics is not injective for MELL.  
Also, it was conjectured that the relational model is
injective for MELL. We prove the conjecture in the present paper.

In \cite{injectcoh}, a proof of the injectivity of the relational model is given for a weak fragment. 
But despite many efforts (\cite{phdtortora}, \cite{injectcoh}, \cite{boudesunifying}, \cite{pagani06a}, \cite{separationdiff}, \cite{taylorexpansioninverse}...), all
the attempts to prove the conjecture failed up to now. 
New progress was made in \cite{LPSinjectivity}, where it has been proved that the relational semantics is injective for ``connected'' MELL proof-nets. Still, there, ``connected'' is understood as a very strong assumption, the set of ``connected'' MELL proof-nets contains the fragment of MELL defined by removing weakenings and units. Actually \cite{LPSinjectivity} proved a much stronger result: in the full MELL fragment two proof-nets $R$ and $R'$ with the same
interpretation are the same ``up to the connections between the doors of exponential
boxes'' (we say that they have the same LPS\footnote{The LPS of a proof-net is the graph obtained by forgetting the outline of the boxes but keeping the auxiliary doors.} - see Figures~\ref{fig: R_1}, \ref{fig: R_2} and \ref{fig: R_3} for an example of three different proof-nets having the same LPS). We wrote: ``This result can be expressed in terms of differential nets: two cut-free proof-nets with different LPS have different Taylor expansions. We also believe this work is an essential step
towards the proof of the full conjecture.'' Despite the fact we obtained a very interesting result about \emph{all} the proof-nets (i.e. also for non-``connected'' proof-nets\footnote{and even adding the MIX rule}), the last sentence was a bit too optimistic, since, in the present paper, which presents a proof of the full conjecture, we could not use any previous result nor any previous technic/idea.

The result of the present paper can be seen as
\begin{itemize}
\item a semantic separation property in the sense of \cite{Friedman};
\item a semantic proof of the confluence property;
\item a proof of the ``canonicity'' of the new syntax of MELL proof-nets;
\item a proof of the fact that if the Taylor expansions of two cut-free MELL proof-nets into differential nets \cite{EhrhardRegnier:DiffNets} coincide, then the two proof-nets coincide.
\end{itemize}

Let us give one more interpretation of its signifance. First, notice that a proof of this result should consist in showing that, given two non $\beta$-equivalent proof-nets $R$ and $R'$, their respective semantics $\sm{R}$ and $\sm{R'}$ are not equal, i.e. $\sm{R} \setminus \sm{R'} \not= \emptyset$ or $\sm{R'} \setminus \sm{R} \not= \emptyset$.\footnote{The converse, i.e. two $\beta$-equivalent proof-nets have the same semantics, holds by definition of soundness.} But, actually, we prove something much stronger: we prove that, given a proof-net $R$, there exist two points $\alpha$ and $\beta$ such that, for any proof-net $R'$, we have $\{ \alpha, \beta \} \subseteq \sm{R'} \Leftrightarrow R \simeq_\beta R'$. 

Now, the points of the relational model can be seen as non-idempotent intersection types\footnote{Idempotency of intersection ($\alpha \cap \alpha = \alpha$) does not hold.}  (see \cite{phddecarvalho} and \cite{Carvalhoexecution} for a correspondance between points of the relational model and System R - System R has also been studied recently in \cite{inhabitation}). And the proof given in the present paper uses the types only to derive the normalization property; actually we prove the injectivity for cut-free proof-nets in an untyped framework:\footnote{For cut-free proof-nets, types guarantee that they are not cyclic as graphs - instead of typing, it is enough to assume this property. Our proof even works for ``non-correct'' proof-structures (correctness is the property characterizing nets corresponding in a typed framework with proofs in sequent calculus): we could expect that if the injectivity of the relational semantics holds for proof-nets corresponding with MELL sequent calculus, then it still holds for proof-nets corresponding with MELL+MIX sequent calculus, since the category \textbf{Rel} of sets and relations is a compact closed category. \cite{k=2} assuming correctness substituted in the proof the ``bridges'' of \cite{LPSinjectivity}  by ``empires''.} substituting the assumption that proof-nets are typed by the assumption that proof-nets are normalizable does not change anything to the proof.\footnote{Except that we have to consider the \emph{atomic} subset of the interpretation instead of the full interpretation (see Remark~\ref{remark: untyped framework}).} 
In \cite{CarvPagTdF10}, we gave a semantic characterization of normalizable untyped proof-nets and we characterized ``head-normalizable'' proof-nets as proof-nets having a non-empty interpretation in the relational semantics. Principal typings in untyped $\lambda$-calculus are intersection types which allow to recover all the intersection types of some term. If, for instance, we consider the System $R$ of \cite{phddecarvalho} and \cite{Carvalhoexecution}, it is enough to consider some \emph{injective $1$-point}\footnote{An \emph{injective $k$-point} is a point in which all the positive multisets have cardinality $k$ and in which each atom occurring in it occurs exactly twice.} to obtain the principal typing of an untyped $\lambda$-term. But, generally, for normalizable MELL proof-nets, \emph{injective $k$-points}, for any $k$, are not principal typings; indeed, two cut-free MELL proof-nets having the same LPS have the same injective $k$-points for any $k \in \Nat$. In the current paper we show that a $1$-point and a \emph{$k$-injective point}\footnote{The reader should not confuse \emph{$k$-injective points} with \emph{injective $k$-points}. \emph{$k$-injective points} are points in which every positive multiset has cardinality $k^j$ for some $j> 0$ and, for any $j > 0$, there is at most one occurrence of a positive multiset having cardinality $k^j$ - they are obtained by \emph{$k$-injective experiments} (see our Definition~\ref{defin: k-injective}).} together allow to recover the interpretation of any normalizable MELL proof-net. So, the result of the current paper can be seen as a first attempt to find a right notion of ``principal typing'' of intersection types in Linear Logic. As a consequence, normalization by evaluation, as in \cite{Rocca88} for $\lambda$-calculus, finally becomes possible in Linear Logic too.

Section~\ref{section: Syntax} formalizes PS's (our cut-free proof-nets). Section~\ref{section: Experiments} gives a sketch of our algorithm leading from $\sm{R}$ to the rebuilding of $R$. Section~\ref{section: one step} describes more precisely one step of the algorithm and states our theorem (Theorem~\ref{thm: injectivity}): 
$\llbracket R \rrbracket = \llbracket R' \rrbracket \Leftrightarrow R \simeq_\beta R'$, 
where $\simeq_\beta$ is the reflexive symmetric transitive closure of the cut-elimination relation.

\begin{notations} 
We denote by $\emptysequence$ any empty sequence. If $a$ is a sequence $(\alpha_1, \ldots, \alpha_n)$, then $\alpha_0:a$ denotes the sequence $(\alpha_0, \ldots, \alpha_n)$; otherwise, it denotes the sequence $(\alpha, a)$ of length $2$. The set of finite sequences of elements of some set $\mathcal{E}$ is denoted by $\finitesequences{\mathcal{E}}$.

A multiset $f$ of elements of some set $\mathcal{E}$ is a function $\mathcal{E} \to \Nat$; we denote by $\supp{f}$ \emph{the support of $f$} i.e. the set $\{ e \in \mathcal{E} ; f(e) \not= 0 \}$. A multiset $f$ is said to be \emph{finite} if $\supp{f}$ is finite. The set of finite multisets of elements of some set $\mathcal{E}$ is denoted by $\finitemultisets{\mathcal{E}}$.

If $f$ is a function $\mathcal{E} \to \mathcal{E'}$, $x_0 \in \mathcal{E}$ and $y \in \mathcal{E'}$, then we denote by $f[x_0 \mapsto y]$ the function $\mathcal{E} \to \mathcal{E'}$ defined by $f[x_0 \mapsto y](x) = \left\lbrace \begin{array}{ll} f(x) & \textit{if $x \not= x_0$;}\\ y & \textit{if $ x = x_0$.} \end{array} \right.$
\end{notations}

\section{Syntax}\label{section: Syntax}

We introduce the syntactical objects we are interested in. As recalled in the introduction, simple types guarantee normalization, so we can limit ourselves to nets without any cut. Correctness does not play any role, that is why we do not restrict our nets to be correct and we rather consider proof-structures (\emph{PS's}). Since our proof is easily extended to MELL with axioms, we remove them for simplicity. Moreover, since it is convenient to represent formally our proof using differential nets with boxes (\emph{differential PS's}), we define PS's as differential PS's satisfying some conditions (Definition~\ref{defin: PS}). More generally, \emph{differential $\circ$-PS's} are defined by induction on the depth: Definition~\ref{defin: diff ground-structure} concerns what happens at depth $0$.

We define the set $\mathbb{T}$ of types as follows: $\mathbb{T} ::= \one \: \vert \: \bot \vert \: (\mathbb{T} \tens \mathbb{T}) \: \vert \: (\mathbb{T} \parr \mathbb{T}) \: \vert \: \cod \mathbb{T} \: \vert \: \contr \mathbb{T}$. We set $\typesoflinks = $ $\{ \tens, $ $\parr, $ $\one, \bottom, \cod, \contr, \circ \}$. Pre-contractions ($\circ$-ports) are an artefact of our inductive definition on the depth and are used to ensure the canonicity of our syntactical objects (see Example~\ref{example: old versus new}).

\begin{definition}\label{defin: diff ground-structure}
A \emph{differential ground-structure} is a $6$-tuple $\mathcal{G} = (\mathcal{W}, \mathcal{P}, l, t, \mathcal{L}, \textsf{T})$, where
\begin{itemize}
\item $\mathcal{P}$ is a finite set; the elements of $\ports{\mathcal{G}}$ are the \emph{ports of $\mathcal{G}$};
\item $l$ is a function $\mathcal{P} \to \typesoflinks$;the element $l(p)$ of $\typesoflinks$ is the \emph{label of $p$ in $\mathcal{G}$};
\item $\mathcal{W}$ is a subset of $\{ p \in \mathcal{P} ; l(p) \not= \circ \}$; the elements of $\wires{\mathcal{G}}$ are the \emph{wires of $\mathcal{G}$};
\item $t$ is a function $\mathcal{W} \to \{ p \in \mathcal{P} ; l(p) \notin \{ \one, \bottom \} \}$ such that, for any port $p$ of $\mathcal{G}$, we have $(l(p) \in$ $\{ \tens, $ $\parr \}$ $\Rightarrow$ $\Card{\{ w \in \mathcal{W} ; t(w) = p \}} = 2)$; 
if $t(w) = p$, then $w$ is a \emph{premise of $p$}; the \emph{arity $\arity{\mathcal{G}}(p)$ of $p$} is the number of its premises;
\item $\mathcal{L}$ is a subset of $\{ w \in \mathcal{W}; l(t(w)) \in \{ \tens, \parr \} \}$ such that $(\forall p \in \mathcal{P})$ $(l(p) \in \{ \tens, \parr \}$ $\Rightarrow$ $\Card{{\{ w \in \mathcal{L} ; t(w) = p \}}} = 1)$; if $w \in \mathcal{L}$ s.t. $t(w) = p$, then $w$ is \emph{the left premise of $p$;}
\item and $\textsf{T}$ is a function $\mathcal{P} \to \mathbb{T}$ such that, for any $p \in \mathcal{P}$,
\begin{itemize}
\item if $l(p) \in \{ \one, \bottom \}$, then $\textsf{T}(p) = l(p)$;
\item if $l(p) = \tens$ (resp. $l(p) = \parr$), then, for any $w_1 \in \mathcal{W} \cap \mathcal{L}$ and any $w_2 \in \mathcal{W} \setminus \mathcal{L}$ such that $t(w_1) = p = t(w_2)$, we have $\textsf{T}(p) = (\textsf{T}(w_1) \tens \textsf{T}(w_2))$ (resp. $\textsf{T}(p) = (\textsf{T}(w_1) \parr \textsf{T}(w_2))$);
\item if $l(p) = \cod$, then $(\exists C \in \mathbb{T}) (\textsf{T}(p) = \cod C \wedge (\forall w \in \mathcal{W}) (t(w) = p \Rightarrow \textsf{T}(w) = C))$;
\item and if $l(p) \in \{ \circ, \contr \}$, then $(\exists C \in \mathbb{T}) (\textsf{T}(p) = \contr C \wedge (\forall w \in \mathcal{W}) (t(w) = p \Rightarrow \textsf{T}(w) = C))$.
\end{itemize}
\end{itemize}
We set 
$\wires{\mathcal{G}} = \mathcal{W}$, $\ports{\mathcal{G}} = \mathcal{P}$, $\labelofcell{\mathcal{G}} = l$, 
$\target{\mathcal{G}} = t$, $\leftwires{\mathcal{G}} = \mathcal{L}$, $\textsf{T}_{\mathcal{G}} = \textsf{T}$. The set $\conclusions{\mathcal{G}} = \mathcal{P} \setminus \mathcal{W}$ is the set of \emph{conclusions of $\mathcal{G}$}. For any $t \in \typesoflinks$, we set $\portsoftype{t}{\mathcal{G}} = \{ p \in \mathcal{P}; l(p) = t \}$; we set $\multiplicativeports{\mathcal{G}} = \portsoftype{\tens}{\mathcal{G}} \cup \portsoftype{\parr}{\mathcal{G}}$; the set $\exponentialports{\mathcal{G}}$ of \emph{exponential ports of $\mathcal{G}$} is $\portsoftype{\cod}{\mathcal{G}} \cup \portsoftype{\contr}{\mathcal{G}} \cup \portsoftype{\circ}{\mathcal{G}}$. 

A \emph{ground-structure} is a differential ground-structure $\mathcal{G}$ s.t. $\im{\target{\mathcal{G}}} \cap (\portsoftype{\cod}{\mathcal{G}} \cup \portsoftype{\circ}{\mathcal{G}}) = \emptyset$.
\end{definition}

Notice that, for any differential ground-structure $\mathcal{G}$, we have $\portsoftype{\circ}{\mathcal{G}} \subseteq \conclusions{\mathcal{G}}$.

\begin{example}
The ground-structure $\mathcal{G}$ defined by: 
$\wires{\mathcal{G}} = \{ p_3, p_4,$ $p_5 \}$; 
$\ports{\mathcal{G}} = \{ p_1, \ldots, p_5 \}$; 
$\labelofcell{\mathcal{G}}(p_1) = \bot$, $\labelofcell{\mathcal{G}}(p_2) = \lpar$, $\labelofcell{\mathcal{G}}(p_3) = \contr = \labelofcell{\mathcal{G}}(p_4)$, $\labelofcell{\mathcal{G}}(p_5) = \lone$; 
$\target{\mathcal{G}}(p_3) = p_2 = \target{\mathcal{G}}(p_4)$, $\target{\mathcal{G}}(p_5) = p_4$; 
$\leftwires{\mathcal{G}} = \{ p_3 \}$; 
and $\textsf{T}_{\mathcal{G}}(p_1) = \bot$, $\textsf{T}_{\mathcal{G}}(p_2) = (\contr 1 \lpar \contr 1)$, $\textsf{T}_{\mathcal{G}}(p_3) = \contr 1 = \textsf{T}_{\mathcal{G}}(p_4)$, $\textsf{T}_{\mathcal{G}}(p_5) =  \one$; 
is the ground-structure of the content of the box $o_1$ of $R$ (the leftmost box of Figure~\ref{fig: new_example}).
\end{example}

The content of every box of our \emph{differential $\circ$-PS's} is a \emph{$\circ$-PS}: every $\cod$-port inside is always the main door of some box.

\begin{definition}\label{defin: differential PS}
For any $d \in \Nat$, we define, by induction on $d$, the set of \emph{differential $\circ$-PS of depth $d$} (resp. the set of \emph{$\circ$-PS of depth $d$}). 
A \emph{differential $\circ$-PS of depth $d$} (resp. a \emph{$\circ$-PS of depth $d$}) is a 4-tuple $S = (\mathcal{G}, \mathcal{B}_0, B, b)$, where
\begin{itemize}
\item $\mathcal{G}$ is a differential ground-structure (resp. a ground-structure);
\item $\mathcal{B}_0 \subseteq \{ p \in \portsoftype{\cod}{\mathcal{G}} ; \arity{\mathcal{G}}(p) = 0 \}$ (resp. $\mathcal{B}_0 = \portsoftype{\cod}{\mathcal{G}}$) and is the set of \emph{boxes of $S$ at depth $0$}; 
\item $B$ is a function that associates with every $o \in \mathcal{B}_0$ a $\circ$-PS $B(o) = (\groundof{B(o)}, $ $\boxesatzero{B(o)}, $ $B_{B(o)}, $ $b_{B(o)})$ of depth $< d$ that enjoys the following property: if $d > 0$, then there exists $o \in \mathcal{B}_0$ s.t. $B(o)$ is a $\circ$-PS of depth $d-1$; the $!$-port $o$ is \emph{the main door of the box $B(o)$};
\item and $b$ is a function that associates with every $o \in \mathcal{B}_0$ a function $b(o) : \conclusions{\groundof{B(o)}} \to \{ o \} \cup \portsoftype{\contr}{\mathcal{G}} \cup \portsoftype{\circ}{\mathcal{G}}$ such that (resp. $\portsoftype{\circ}{\mathcal{G}} \subseteq \bigcup_{o \in \mathcal{B}_0} \im{b(o)}$ and), for any $o \in \mathcal{B}_0$, 
\begin{itemize}
\item $\restriction{b(o)}{\portsoftype{\circ}{\groundof{B(o)}}}$ is injective\footnote{So one cannot (pre-)contract several $\circ$-ports of the same box.} (resp.\footnote{This stronger condition on $\circ$-PS's is \emph{ad hoc}, but it allows to lighten the notations.} $\restriction{b(o)}{\portsoftype{\circ}{\groundof{B(o)}}} = id_{\portsoftype{\circ}{\groundof{B(o)}}}$); if $q = b(o)(p)$ with $p \in {\portsoftype{\circ}{\groundof{B(o)}}}$, then we set $q_{S, o} = p$;
\item $o \in \im{b(o)}$ and $\im{\restriction{b(o)}{\portsoftype{\circ}{\groundof{B(o)}}}} \cap \portsoftype{\cod}{\mathcal{G}} = \emptyset$;
\item for any $p \in \dom{b(o)} \cap \portsoftype{\circ}{\groundof{B_R(o)}}$, we have $\textsf{T}_{\mathcal{G}}(b(o)(p))$ $=$ $\textsf{T}_{\groundof{B(o)}}(p)$;
\item for any $p \in \dom{b(o)} \setminus \portsoftype{\circ}{\groundof{B_R(o)}}$, we have $\textsf{T}_{\mathcal{G}}(b(o)(p)) \in \{  \contr \textsf{T}_{\groundof{B(o)}}(p), \cod \textsf{T}_{\groundof{B(o)}}(p) \}$;
\end{itemize}
\end{itemize}
(resp. moreover no $p \in \ports{\mathcal{G}}$ is a sequence)\footnote{This condition on $\circ$-PS's is \emph{ad hoc}, but it allows to simplify Definition~\ref{defin: Taylor}.}.
For any differential $\circ$-PS $S = (\mathcal{G}, \mathcal{B}_0, B, b)$, we set $\groundof{S} = \mathcal{G}$, $\boxesatzero{S} = \mathcal{B}_0$ and $\boxes{S} = \boxesatzero{S} \cup \bigcup_{o \in \boxesatzero{S}} \{ \mbox{o:o'} ; o' \in \boxes{B_S(o)} \}$ is the set of \emph{boxes of $S$}. We denote by $B_S$ the extension of the function $B$ that associates with each $\mbox{o:o'} \in \boxes{S}$, where $o \in \boxesatzero{S}$, the $\circ$-PS $B_{B_S(o)}(o')$. 
We denote by $b_S$ the extension of the function $b$ that associates with each $\mbox{o:o'} \in \boxes{S}$, where $o \in \boxesatzero{S}$, the function $b_{B_S(o)}(o')$.

We set 
$\wiresatzero{S} = \wires{\groundof{S}}$ and $\portsatzero{S} = \ports{\groundof{S}}$; the elements of $\portsatzero{S}$ (resp. of $\wiresatzero{S}$) are the \emph{ports of $S$ at depth $0$} (resp. the \emph{wires of $S$ at depth $0$}). 
For any $l \in \typesoflinks \cup \{ \textit{m}, \textit{e} \}$, we set $\portsatzerooftype{l}{S} = \portsoftype{l}{\groundof{S}}$. 
We set $\conclusions{S} = \conclusions{\groundof{S}}$, $\conclusionscirc{S} = \portsoftype{\circ}{\groundof{S}}$ and $\conclusionsnotcirc{S} = \conclusions{S} \setminus \conclusionscirc{S}$; the elements of $\conclusions{S}$ are the \emph{conclusions of $S$} and the elements of $\conclusionscirc{S}$ are the \emph{$\circ$-conclusions of $S$}. For any relation $P \in \{ \geq, =, < \}$ on $\Nat$, for any $i \in \Nat$, we set $\mathcal{B}_0^{P i}(S) = \{ o \in \boxesatzero{S} ; \textit{depth}(B_S(o)) P i \}$ and $\mathcal{B}^{P i}(S) = \{ o \in \boxes{S} ; \textit{depth}(B_S(o)) P i \}$.
\end{definition}

\begin{figure}
\centering
\scalebox{\scalefactter}{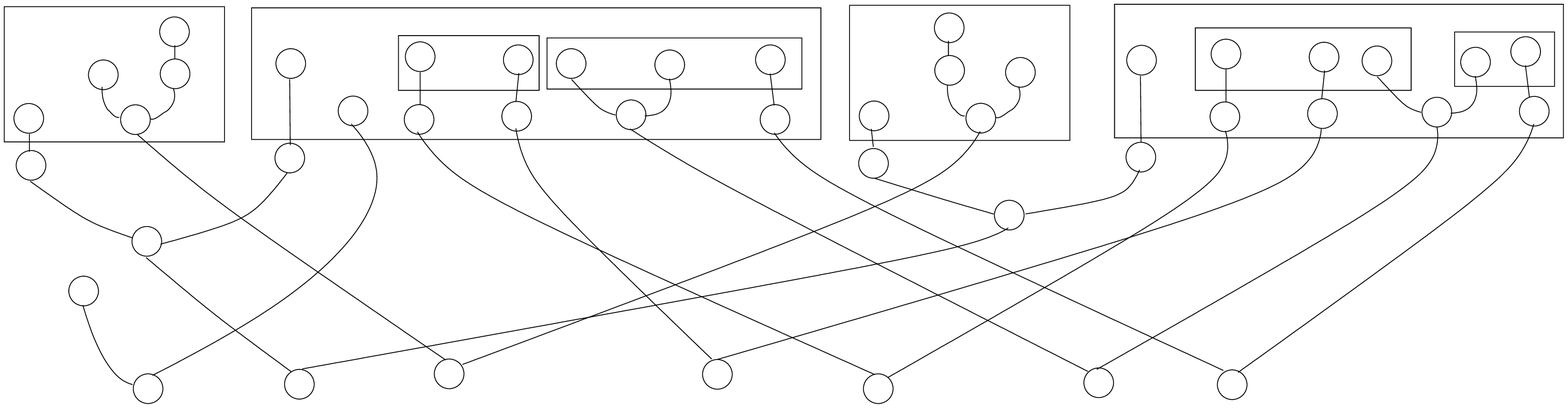}
\caption{PS $R$}
\label{fig: new_example}
\end{figure}

PS's are the MELL proof-nets studied in the present paper: there is no cut and no assumption of correctness property.

\begin{definition}\label{defin: PS}
A PS is a $\circ$-PS $R$ such that $\conclusionscirc{R} = \emptyset$.
\end{definition}

\begin{example}
Consider the PS $R$ of Figure~\ref{fig: new_example}. 
We have $\boxesatzero{R} = \{ o_1, o_2, o_3, o_4 \}$, $\boxes{R} = \{ o_1,$ $o_2,$ $o_3,$ $o_4,$ $(o_2, o),$ $(o_2, o'),$ $(o_4, o),$ $(o_4, o') \}$, $\exactboxes{R}{0} = \{ o_1, (o_2, o), (o_2, o'), o_3, (o_4, o), (o_4, o') \}$. We have $b_{R}(o_2)(o) = p_5$, $b_R(o_2)(p_4) = p_4$, $b_R(o_2)(p_6) = p_6$ and $b_R(o_2)(o') = p_7$.
\end{example}

\begin{example}~\label{example: old versus new}
In order to understand the role of the $\circ$-ports, consider how the proof-nets $O_1$ (Figure~\ref{figure: O1}) and $O_2$ (Figure~\ref{figure: O2}) in the ``old syntax'' (we denoted derelictions, contractions and auxiliary doors of the ``old syntax'' by \textsf{d}, \textsf{c} and \textsf{a}, respectively) are represented by the same PS $N$ (Figure~\ref{figure: N}). Roughly speaking, in our formalism, one pre-contracts (using $\circ$-ports) as soon as possible and one contracts (using $\contr$-ports) as late as possible.

\begin{figure}
\centering
\begin{minipage}{4.5cm}
\centering
\scalebox{\scalefactter}{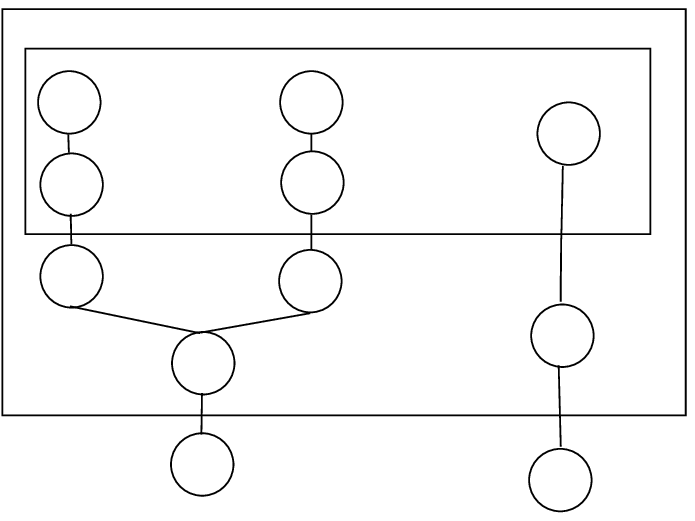}
\caption{$O_1$ (``old syntax'')}
\label{figure: O1}
\end{minipage}\hfill
\begin{minipage}{4.5cm}
\centering
\scalebox{\scalefactter}{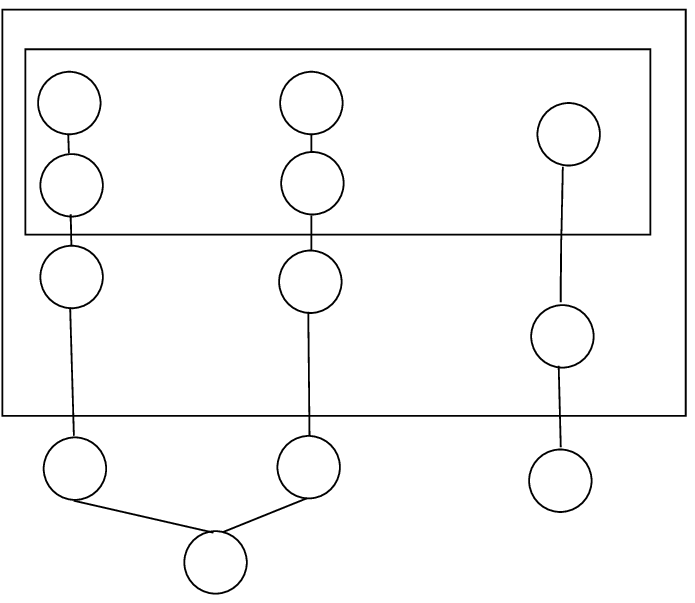}
\caption{$O_2$ (``old syntax'')}
\label{figure: O2}
\end{minipage}
\begin{minipage}{4.5cm}
\centering
\scalebox{\scalefactter}{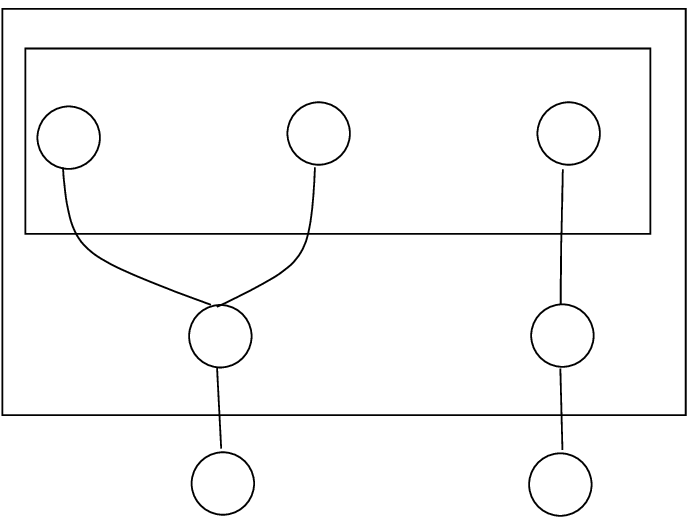}
\caption{PS $N$}
\label{figure: N}
\end{minipage}
\end{figure}
\end{example}


We write $R \simeq R'$ (resp. $R \equiv R'$) if $R$ and $R'$ are the same differential PS's up to the names of their ports (resp. that are not conclusions):

\begin{definition}
An \emph{isomorphism $\varphi: \mathcal{G} \simeq \mathcal{G'}$ of ground-structures} is a structure-preserving bijection $\portsatzero{\mathcal{G}} \simeq \portsatzero{\mathcal{G'}}$. We define, by induction on $\depthof{R}$, when $\varphi : R \simeq R'$ holds for two differential $\circ$-PS's $R$ and $R'$: it holds whenever $\varphi$ is a pair $(\varphi_\mathcal{G}, (\varphi_o)_{o \in \boxesatzero{R}})$ s.t. $\varphi_\mathcal{G} : \groundof{R} \simeq \groundof{R'}$, $\boxesatzero{R'} = \im{\restriction{\varphi_{\mathcal{G}}}{\portsatzerooftype{\cod}{R}}}$ and, 
for any $o \in \boxesatzero{R}$, 
$\varphi_o: B_R(o) \simeq B_{R'}(\varphi_\mathcal{G}(o))$ and $(\forall q \in \conclusions{B_R(o)}) b_{R'}(\varphi_{\mathcal{G}}(o))(\groundof{\varphi_o}(q)) = \varphi_{\mathcal{G}}(b_R(o)(q))$. 
We set $\groundof{\varphi} = \varphi_\mathcal{G}$ and, for any $o \in \boxesatzero{R}$, $\varphi(o) = \varphi_o$. 
We write $\varphi: R \equiv R'$ if $\varphi : R \simeq R'$ s.t. $\restriction{\groundof{\varphi}}{\conclusions{R}} = id_{\conclusions{R}}$.

We write $R \simeq R'$ (resp. $R \equiv R'$) if there exists $\varphi$ s.t. $\varphi : R \simeq R'$ (resp. $\varphi: R \equiv R'$).
\end{definition}

The \emph{arity $\arity{R}(q)$ of a port $q$ in a differential $\circ$-PS $R$} is computed  by ``ignoring'' the $\circ$-conclusions of the boxes of $R$:

\begin{definition}
Let $R$ be a differential $\circ$-PS. We define, by induction on $\depthof{R}$, the integers $\arity{R}(q)$ for any $q \in \portsatzero{R}$ and $\cosize{R}$: we set 
$\arity{R}(q) $ $=$ $\arity{\groundof{R}}(q) +$ $\sum_{o \in \boxesatzero{R}}$ $\Card{\left\lbrace p \in \conclusionsnotcirc{B_R(o)} ; b_{R}(o)(p) = q  \right\rbrace}$ $+$ $\sum_{\substack{o \in \boxesatzero{R}\\ q \in \im{\restriction{b_R(o)}{\conclusionscirc{B_R(o)}}}}} \arity{B_R(o)}(q_{R, o})$ and 
$\cosize{R}$ $=$ $\max(\{ \arity{R}(p) ; p \in \portsatzero{R} \} \cup \{ \cosize{B_R(o)} ; o \in \boxesatzero{R} \})$.
\end{definition}

\begin{example}
We have $\arity{R}(p_6) = 4$ (and not $2$) and $\cosize{R} = 4$ (see Figure~\ref{fig: new_example}).
\end{example}

\section{Experiments and their partial expansions}\label{section: Experiments}


When Jean-Yves Girard introduced proof-nets in \cite{ll}, he also introduced \emph{experiments of proof-nets}. Experiments (see our Definition~\ref{defin: experiment}) are a technology allowing to compute pointwise the interpretation $\llbracket R \rrbracket$ of a proof-net $R$ in the model directly on the proof-net rather than through some sequent calculus proof obtained from one of its sequentializations: the set of \emph{results} of all the experiments of a given proof-net is its interpretation $\llbracket R \rrbracket$. In an untyped framework, experiments correspond with type derivations and results correspond with intersection types. 

\begin{definition}\label{defin: experiment}
For any $C \in \mathbb{T}$, we define, by induction on $C$, the set $\sm{C}$: $\sm{1} = \{ \ast \} = \sm{\bot}$; $\sm{(C_1 \ltens C_2)} = \sm{C_1} \times \sm{C_2} = \sm{(C_1 \lpar C_2)}$; $\sm{\cod C} = \finitemultisets{\sm{C}} = \sm{\contr C}$.

Let $R$ be a differential $\circ$-PS. We define, by induction on $\textit{depth}(R)$, the set of \emph{experiments of $R$}: it is is the set of triples $(R, e_\mathcal{P}, e_\mathcal{B})$, where $e_\mathcal{P}$ is a function that associates with every $p \in \portsatzero{R}$ an element of $\sm{\textsf{T}_{\groundof{R}}(p)}$ and $e_\mathcal{B}$ is a function which associates to every $o \in \boxesatzero{R}$ a finite multiset of experiments of $B_R(o)$ such that
\begin{itemize}
\item for any $p \in \multiplicativeportsatzero{R}$, for any $w_1, w_2 \in \wiresatzero{R}$ such that $\target{\groundof{R}}(w_1) = p = \target{\groundof{R}}(w_2)$, $w_1 \in \leftwires{\groundof{R}}$ and $w_2 \notin \leftwires{\groundof{R}}$, we have $e_\mathcal{P}(p) = (e_\mathcal{P}(w_1), e_\mathcal{P}(w_2))$;
\item for any $p \in \exponentialportsatzero{R}$,
we have $e(p) = \sum_{\substack{w \in \wiresatzero{R}\\ \target{\groundof{R}}(w) = p}} [e_\mathcal{P}(w)] +  \sum_{o \in \boxesatzero{R}} \sum_{e' \in \supp{e_\mathcal{B}(o)}}$ $(\sum_{\substack{q \in \conclusionsnotcirc{B_R(o)}\\ b_R(o)(q) = p}} e_\mathcal{B}(o)(e') \cdot [{e'}_\mathcal{P}(q)] + \sum_{\substack{q \in \conclusionscirc{B_R(o)}\\ b_R(o)(q) = p}} e_\mathcal{B}(o)(e') \cdot {e'}_\mathcal{P}(q) )$.
\end{itemize}
For any experiment $e = (R, e_\mathcal{P}, e_\mathcal{B})$, we set $\ports{e} = e_\mathcal{P}$ and $\boxes{e} = e_\mathcal{B}$. We set $\sm{R} =$ $\{ \restriction{\ports{e}}{\conclusions{R}} ;$ $e \textit{ is an experiment of } R \}$.
\end{definition}

We encode in a more compact way the ``relevant'' information given by an experiment \emph{via} \emph{pseudo-experiments} and the functions $e^\#$:

\begin{definition}\label{defin: experiment induces pseudo-experiment}
For any differential $\circ$-PS $R$, we define, by induction on $\textit{depth}(R)$, the set of \emph{pseudo-experiments of $R$}: it is the set of functions that associate with every $o \in \boxesatzero{R}$ a finite set of pseudo-experiments of $B_R(o)$ and with $\emptysequence$ a pair $(R, m)$ for some $m \in \Nat$.

Given an experiment $e$ of some differential $\circ$-PS $R$, we define, by induction on $\textit{depth}(R)$, a pseudo-experiment $\overline{e}$ of $R$ as follows: $\overline{e}(\emptysequence) = (R, 1)$ and $\overline{e}(o) = \bigcup_{f \in \textit{Supp}(\boxes{e}(o))}$ $\{ \overline{f}[\emptysequence \mapsto (B_R(o), i)] ; 1 \leq i \leq \boxes{e}(o)(f) \}$ for any $o \in \boxesatzero{R}$.

Given a pseudo-experiment $e$ of a differential $\circ$-PS $R$, we define, by induction on $\textit{depth}(R)$, the function $e^\#: \boxes{R} \rightarrow \mathcal{P}_\textit{fin}(\Nat)$ as follows: for any $o \in \boxesatzero{R}$, $e^\#(o) = \{ \Card{e(o)} \}$ and, for any $o' \in \boxes{B_R(o)}$, $e^\#(\mbox{o:o'}) = \bigcup_{e' \in e(o)} {e'}^\#(o')$.
\end{definition}

There are different kinds of experiments:
\begin{itemize}
\item In \cite{injectcoh}, it was shown that given the result of an \emph{injective $k$-obsessional experiment} ($k$ big enough) of a cut-free proof-net in the fragment $A ::= X \vert \contr A \parr A \vert A \parr \contr A \vert A \otimes A \vert !A$, one can rebuild the entire experiment and, so, the entire proof-net. There, ``injective'' means that the experiment labels two different axioms with different atoms and ``obsessional'' means that different copies of the same axiom are labeled by the same atom.
\item In \cite{LPSinjectivity}, it was shown that for any two cut-free MELL proof-nets $R$ and $R'$, we have $LPS(R) = LPS(R')$ iff, for $k$ big enough\footnote{Interestingly, \cite{k=2}, following the approach of \cite{LPSinjectivity}, showed that, if these two proof-nets are assumed to be (recursively) connected, then we can take $k = 2$.}, there exist an \emph{injective $k$-experiment} of $R$ and an \emph{injective $k$-experiment} of $R'$ having the same result; as an immediate corollary we obtained the injectivity of the set of (recursively) connected proof-nets. There, ``injective'' means that not only the experiment labels two different axioms with different atoms, but it labels also different copies of the same axiom by different atoms. Given some proof-net $R$, there is exactly one injective $k$-experiment of $R$ up to the names of the atoms.
\item In the present paper we show that, for any two PS's $R$ and $R'$, given the result $\alpha$ of a \emph{$k$-injective experiment} of $R$ for $k$ big enough, if $\alpha \in \sm{R'}$, then $R'$ is the same PS as $R$. The conditions on $k$ are given by the result of a $1$-experiment, so we show that two (well-chosen) points are enough to determine a PS. The expression ``$k$-injective'' means that, for any two different occurrences of boxes, the experiment never takes the same number of copies: it takes $k^{j_1}$ copies and $k^{j_2}$ copies with $j_1 \not= j_2$ (\emph{a contrario}, in \cite{injectcoh} and \cite{LPSinjectivity}, the experiments always take the same number of copies). As shown by the proof-net S of Figure~\ref{fig: no reconstruction of the experiment}, it is impossible to rebuild the experiment from its result, since there exist four different $4$-injective experiments $e_1, e_2, e_3$ and $e_4$ such that, for any $i \in \{ 1, 2, 3, 4 \}$, we have $e_i(p) = (\ast, \ast)$, $e_i(o_1) = [\ast, \ast, \ast, \ast]$ and $e_i(p') = [[\underbrace{\ast, \ldots, \ast}_{4^2}], \ldots, [\underbrace{\ast, \ldots, \ast}_{4^6}]]$. For instance $e_1$ takes $4$ copies of the box $o_1$ and $16$ copies of the box $o_2$, while $e_2$ takes $4$ copies of the box $o_1$ and $64$ copies of the box $o_2$.
\end{itemize}

\begin{definition}\label{defin: k-injective}
Let $k > 1$. A pseudo-experiment $e$ of a $\circ$-PS $R$ is said to be \emph{$k$-injective} if
\begin{itemize}
\item for any $o \in \boxes{R}$, for any $m \in e^\#(o)$, there exists $j > 0$ such that $m = k^j$;
\item for any $o \in \boxesatzero{R}$, for any $o' \in \boxes{B_R(o)}$, we have $(\forall e_1, e_2 \in e(o))$ $({e_1}^\#(o') \cap {e_2}^\#(o') \not= \emptyset \Rightarrow e_1 = e_2)$;
\item and, for any $o_1, o_2 \in \boxes{R}{}$, we have $({e}^\#(o_1) \cap {e}^\#(o_2) \not= \emptyset \Rightarrow o_1 = o_2)$.
\end{itemize}
An experiment $e$ is said to be \emph{$k$-injective} if $\overline{e}$ is $k$-injective.
\end{definition}

\begin{example}\label{example: pseudo}
There exists a $10$-injective pseudo-experiment $f$ of the proof-net $R$ of Figure~\ref{fig: new_example} such that
$f^\#(o_1) = \{ 10^{223} \}$, 
$f^\#(o_2) = \{ 10 \}$, 
$f^\#(o_3) = \{ 10^{224} \}$, 
$f^\#(o_4) = \{ 100 \}$, 
$f^\#((o_2, o))$ $=$ $\{ 10^3, \ldots, 10^{12} \}$, 
$f^\#((o_2, o'))$ $=$ $\{ 10^{13}, \ldots, 10^{22} \}$, 
$f^\#((o_4, o))$ $=$ $\{ 10^{23}, \ldots, 10^{122} \}$ 
and $f^\#((o_4, o'))$ $=$ $\{ 10^{123}, \ldots, 10^{222} \}$.
\end{example}

In \cite{LPSinjectivity}, the interest for \emph{injective} experiments came from the remark that the result of an \emph{injective} experiment of a \emph{cut-free} proof-net can be easily identified with a differential net of its Taylor expansion in a sum of differential nets \cite{EhrhardRegnier:DiffNets} (it is essentialy the content of our Lemma~\ref{lem: taylor expansion}). Thus any proof using injective experiments can be straightforwardly expressed in terms of differential nets and conversely. Since this identification is trivial, besides the idea of considering injective experiments instead of obsessional experiments, the use of the terminology of differential nets does not bring any new insight\footnote{For proof-nets with cuts, the situation is completely different: the great novelty of differential nets is that differential nets have a cut-elimination; the differential nets appearing in the Taylor expansion of a proof-net with cuts have cuts, while the semantics does not see these cuts. But the proofs of the injectivity only consider cut-free proof-nets.}, it just superficially changes the presentation.That is why we decided in \cite{LPSinjectivity} to avoid introducing explicitely differential nets. In the present paper, we made the opposite choice for the following reason: 
the algorithm leading from the result of a $k$-injective experiment of $R$ to the entire rebuilding of $R$ is done in several steps: in the intermediate steps, we obtain a partial rebuilding where some boxes have been recovered but not all of them; a convenient way to represent this information is the use of ``differential nets with boxes'' (called ``differential PS's'' in the present paper). Now, the differential net representing the result and the proof-net $R$ are both instances of the more general notion of ``differential nets with boxes''.

The rebuilding of the proof-net $R$ is done in $d$ steps, where $d$ is the depth of $R$. We first rebuild the occurrences of the boxes of depth $0$ 
(the deepest ones) and next we rebuild the occurrences of the boxes of depth $1$ 
and so on... This can be formalized using differential nets (with boxes) as follows: if $e$ is an injective experiment of $R$, then $\taylor{\overline{e}}{i}$ is the differential net corresponding with $e$ in which only boxes of depth $\geq i$ are expanded,\footnote{Boxes \emph{of} depth $\geq i$ are boxes whose content is a proof-net of depth $\geq i$; the reader should not confuse \emph{boxes of depth $\geq i$} with \emph{boxes at depth $\geq i$}.} so $\taylor{\overline{e}}{0}$ is (essentially) the same as the result of the experiment and $\taylor{\overline{e}}{d} = R$; the first step of the algorithm builds $\taylor{\overline{e}}{1}$ from $\taylor{\overline{e}}{0}$, the second step builds $\taylor{\overline{e}}{2}$ from $\taylor{\overline{e}}{1}$, and so on... We thus reduced the problem of the injectivity to the problem of rebuilding $\taylor{\overline{e}}{i+1}$ from $\taylor{\overline{e}}{i}$ for any $k$-injective experiment $e$ ($k$ big enough).

\begin{definition}\label{defin: Taylor}
Let $R$ be a $\circ$-PS of depth $d$. Let $e$ be a pseudo-experiment of $R$. Let $i \in \Nat$. 
We define, by induction on $d$, a differential $\circ$-PS $\taylor{e}{i} = $ $((\mathcal{W}_{e, i},$ $\mathcal{P}_{e, i}, $ $l_{e, i}, $ $t_{e, i}, $ $\mathcal{L}_{e, i}),$ $\mathcal{B}_{e, i},$ $B_{e, i},$ $b_{e, i})$ of depth $\min \{ i, d \}$  s.t. 
$\conclusions{R} = \conclusions{\taylor{e}{i}}$ 
and $(\forall p \in \conclusions{R}) \labelofcell{\groundof{R}}(p) $ $=$ $\labelofcell{\groundof{\taylor{e}{i}}}(p)$ as follows: we set $\mathcal{P}_{e, i}^{\bullet} =$ $\bigcup_{o_1 \in \boxesatzerogeq{R}{i}} \bigcup_{e_1 \in e(o_1)}$ $\left\lbrace (o_1, e_1):p ; p \in \mathcal{P}_{e_1, i} \setminus \conclusionscirc{B_R(o_1)} \right\rbrace$;
\begin{itemize}
\item $\mathcal{W}_{e, i} =$ $\wiresatzero{R} \cup \mathcal{P}_{e, i}^{\bullet}$ and $\mathcal{P}_{e, i} = \portsatzero{R}$ $\cup \mathcal{P}_{e, i}^{\bullet}$;
\item $l_{e, i}(p) = \left\lbrace \begin{array}{ll} \labelofcell{\groundof{R}}(p) & \textit{if $p \in \portsatzero{R}$;}\\ l_{e_1, i}(p') & \textit{if $p = (o_1, e_1) : p'$ with $o_1 \in \boxesatzerogeq{R}{i}$;} \end{array} \right.$
\item $t_{e, i}$ is the extension of $\target{\groundof{R}}$ that associates with each $(o_1, e_1) : w' \in \mathcal{W}_{e, i}$, where $o_1 \in \boxesatzerogeq{R}{i}$, the port $\left\lbrace \begin{array}{ll}  
(o_1, e_1) : t_{e_1, i}(w') & \textit{if $w' \in \mathcal{W}_{e_1, i}$ and $t_{e_1, i}(w') \notin \conclusionscirc{B_R(o_1)}$;}\\
b_R(o_1)(t_{e_1, i}(w')) & \textit{if $w' \in \mathcal{W}_{e_1, i}$ and $t_{e_1, i}(w') \in \conclusionscirc{B_R(o_1)}$;} \\
b_R(o_1)(w') & \textit{if $w' \in \conclusionsnotcirc{B_R(o_1)}$;}
\end{array}\right.$
\item $\mathcal{L}_{e, i} = \leftwires{\groundof{R}} \cup \bigcup_{o_1 \in \boxesatzerogeq{R}{i}} \bigcup_{e_1 \in e(o_1)} \{ (o_1, e_1):p ; p \in \mathcal{L}_{e_1, i} \}$
\item $\mathcal{B}_{e, i} = \boxesatzerosmaller{R}{i} \cup \bigcup_{o_1 \in \boxesatzerogeq{R}{i}} \bigcup_{e_1 \in e(o_1)} \{ $ \mbox{$(o_1, e_1):o'$ ;} $o' \in \mathcal{B}_{e_1, i} \}$
\item $B_{e, i}(o) = \left\lbrace \begin{array}{ll} B_R(o) & \textit{if $o \in \boxesatzerosmaller{R}{i}$;}\\ B_{e_1, i}(o') & \textit{if \mbox{$o = (o_1, e_1):o'$} with $o_1 \in \boxesatzerogeq{R}{i}$;} \end{array} \right.$
\item $b_{e, i}$ is the extension of $\restriction{b_R}{\boxesatzerosmaller{R}{i}}$ that associates with each $(o_1, e_1):o' \in\mathcal{B}_{e, i}$, where $o_1 \in \boxesatzerogeq{R}{i}$, the function $p \mapsto \left\lbrace 
\begin{array}{ll} 
(o_1, e_1):b_{e_1, i}(o')(p) & \textit{if $b_{e_1, i}(o')(p) \notin \conclusionscirc{B_R(o_1)}$;}\\
b_{e_1, i}(o')(p) & \textit{if $b_{e_1, i}(o')(p) \in \conclusionscirc{B_R(o_1)}$.} \end{array} \right.$
\end{itemize}
\end{definition}

\begin{figure*}
\scalebox{\scalefactfive}{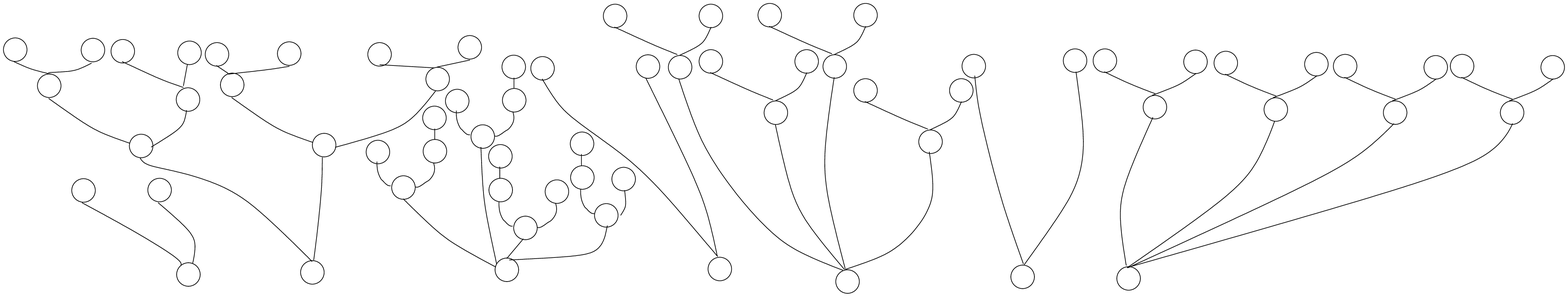}
\caption{$\mathcal{T}(f)[0]$}
\label{fig: T(f)[0]}
\end{figure*}

\begin{figure*}
\scalebox{\scalefacteight}{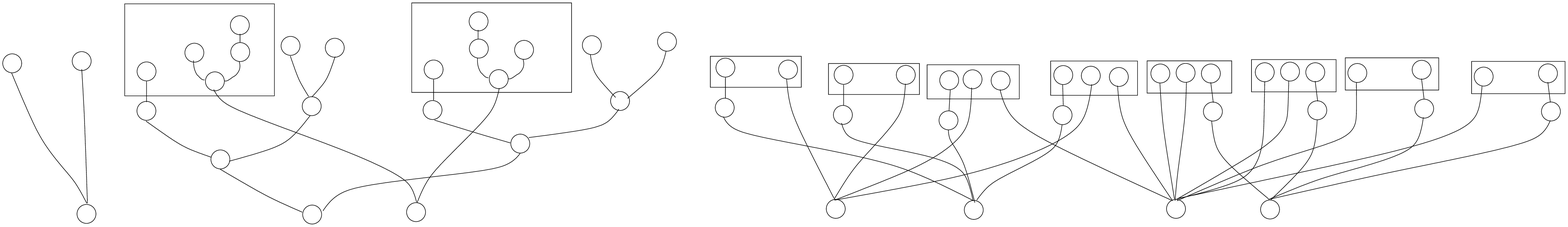}
\caption{$\mathcal{T}(f)[1]$}
\label{fig: T(f)[1]}
\end{figure*}

\begin{example}
If $f$ is a pseudo-experiment of the proof-net $R$ of Figure~\ref{fig: new_example} with $f^\#$ like in Example~\ref{example: pseudo}, then Figures~\ref{fig: T(f)[0]} and~\ref{fig: T(f)[1]} represent respectively $\taylor{f}{0}$ and $\taylor{f}{1}$.
\end{example}


The injectivity of the relational semantics for differential PS's of depth $0$ is trivial (one can proceed by induction on the cardinality of the set of ports). Since $\sm{\taylor{\overline{e}}{0}} = \{ \restriction{e}{\conclusions{R}} \}$, one can easily identify the result $\restriction{e}{\conclusions{R}}$ of an experiment $e$ with the differential net $\taylor{\overline{e}}{0}$:

\begin{lem}\label{lem: taylor expansion}
Let $R$ and $R'$ be two $\circ$-PS's such that $\conclusions{R} = \conclusions{R'}$. Let $e$ be an experiment of $R$ and let $e'$ be an experiment of $R'$ such that $\restriction{e}{\conclusions{R}} = \restriction{e'}{\conclusions{R'}}$. Then $\taylor{\overline{e}}{0} \equiv \taylor{\overline{e'}}{0}$.
\end{lem}

Now, the following fact shows that 
if we are able to recover $\taylor{\overline{e}}{\depthof{R}}$ from $\taylor{\overline{e}}{0}$, then we are done.

\begin{fact}\label{fact: Taylor[i] with i large}
Let $R$ be a $\circ$-PS. Let $e$ be a pseudo-experiment of $R$. Then $\taylor{e}{\depthof{R}} = R$.
\end{fact}

If $e$ is a $k$-injective experiment of $R$, then, for any $i \in \Nat$, there exists a bijection $\cod_{e, i} : \bigcup_{o \in \boxesgeq{R}{i}} \{ \log_k(m) ; m \in e^\#(o) \} \simeq \portsatzerooftype{\cod}{\taylor{e}{i}} \setminus \boxesatzero{\taylor{e}{i}}$ such that, for any $j \in \dom{\cod_{e, i}}$, we have  $(\arity{\taylor{e}{i}} \circ !_{e, i})(j) = k^j$. In Subsection~\ref{subsection: outline}, we will show how to recover $\bigcup_{o \in \boxesgeq{R}{i}} \{ \log_k(m) ; $ $m \in e^\#(o) \}$ from $\taylor{e}{0}$. There are two kinds of boxes of $\taylor{e}{i+1}$ at depth $0$: the ``new'' boxes of depth $i$ and the boxes of depth $< i$, which are the ``old'' boxes (i.e. that already were in $\taylor{e}{i}$) that do not go inside some ``new'' box:

\begin{fact}\label{fact : boxes of taylor{e}{i+1} of depth < i}
Let $R$ be a $\circ$-PS. Let $e$ be a pseudo-experiment of $R$. Let $i \in \Nat$. Then we have 
\begin{itemize}
\item $\boxesatzerosmaller{\taylor{e}{i+1}}{i} = \boxesatzero{\taylor{e}{i}} \cap \portsatzerooftype{\cod}{\taylor{e}{i+1}}$;
\item $\restriction{B_{\taylor{e}{i+1}}}{\boxesatzerosmaller{\taylor{e}{i+1}}{i}} = \restriction{B_{\taylor{e}{i}}}{\boxesatzerosmaller{\taylor{e}{i+1}}{i}}$; 
\item and $\restriction{b_{\taylor{e}{i+1}}}{\boxesatzerosmaller{\taylor{e}{i+1}}{i}} = \restriction{b_{\taylor{e}{i}}}{\boxesatzerosmaller{\taylor{e}{i+1}}{i}}$.
\end{itemize}
\end{fact}

The challenge is the rebuilding of the ``new'' boxes at depth $0$ of depth $i$.

\section{From $\taylor{e}{i}$ to $\taylor{e}{i+1}$}\label{section: one step}

\subsection{The outline of the boxes}\label{subsection: outline}

In this subsection we first show how to recover the set $\bigcup_{o \in \boxesgeq{R}{i}} \{ \log_k(m) ; m \in e^\#(o) \}$ and, therefore, the set $\portsatzerooftype{\cod}{\taylor{e}{i}} \setminus \boxesatzero{\taylor{e}{i}}$ (Lemma~\ref{lem: M_i}). Next, we show how to determine, from $\taylor{e}{i}$, the set $\exactboxesatzero{\taylor{e}{i+1}}{i}$ of ``new'' boxes and, for any such ``new'' box $o \in \exactboxesatzero{\taylor{e}{i+1}}{i}$, the set $\im{b_{\taylor{e}{i+1}}(o)}$ of exponential ports that are immediately below (Proposition~\ref{prop: critical ports below new boxes}). In particular, we have $\exactboxesatzero{\taylor{e}{i+1}}{i} = \im{\restriction{!_{e, i}}{\mathcal{N}_i(e)}}$, where the set $\mathcal{N}_i(e) \subseteq \Nat$ is defined from the set $\mathcal{M}_0(e)$ of the numbers of copies of boxes taken by the pseudo--experiment $e$:



\begin{definition}\label{defin: the algorithm (b)}
Let $R$ be a differential $\circ$-PS. Let $k > 1$. Let $e$ be a $k$-injective pseudo-experiment of $R$. For any $i \in \Nat$, we define, by induction on $i$, $\mathcal{M}_i(e) \subseteq \Nat \setminus \{ 0 \}$ and $(m_{i, j}(e))_{j \in \Nat} \in \{ 0, \ldots, k-1 \}^\Nat$ as follows. 
We set $\mathcal{M}_0(e) = \bigcup_{o \in \boxes{R}{}} \{ j \in \Nat; k^j \in e^\#(o) \}$ and we write $\Card{\mathcal{M}_i(e)}$ in base $k$: $\Card{\mathcal{M}_i(e)} = \sum_{j \in \Nat} m_{i, j}(e) \cdot k^j$; we set $\mathcal{M}_{i+1}(e) =$ $\{ j > 0 ; m_{i, j}(e) \not= 0 \}$.

For any $i \in \Nat$, we set $\mathcal{N}_i(e) =  \mathcal{M}_i(e) \setminus \mathcal{M}_{i+1}(e)$.
\end{definition}

Notice that all the sets $\mathcal{M}_i(e)$ and $\mathcal{N}_i(e)$ can be computed from $\taylor{e}{0}$, since
 we have $\mathcal{M}_0(e) =$ $\{ \arity{\taylor{e}{0}}(p) ;$ $p \in \portsatzerooftype{\cod}{\taylor{e}{0}} \}$.

\begin{example}\label{example: N_i(e)}
If $f$ is a $10$-injective pseudo-experiment as in Example~\ref{example: pseudo}, then $\mathcal{M}_0(f) = \{ 1, \ldots, 224 \}$. We have $\Card{\mathcal{M}_0(f)} = 4 + 2 \cdot 10^1 + 2 \cdot 10^2$, hence $\mathcal{M}_1(f) = \{ 1, 2 \}$ and $\mathcal{N}_0(f) = \{ 3, \ldots, 224 \}$. We have $\Card{\mathcal{M}_1(f)} = 2$, hence $\mathcal{M}_2(f) = \emptyset$ and $\mathcal{N}_1(f) = \{ 1, 2 \}$.  
\end{example}

The following lemma shows that, for any $k$-injective pseudo-experiment $e$ of $R$, for any $i \in \Nat$, the function $\cod_{e, i}$ is actually a bijection $\mathcal{M}_i(e) \to \portsatzerooftype{\cod}{\taylor{e}{i}} \setminus \boxesatzero{\taylor{e}{i}}$ such that, for any $j \in \mathcal{M}_i(e)$,  we have  $(\arity{\taylor{e}{i}} \circ \cod_{\taylor{e}{i}})(j) = k^j$. 

\begin{lem}\label{lem: M_i}
Let $R$ be a $\circ$-PS. 
Let $k > \Card{\boxes{R}{}}$. 
For any $k$-injective pseudo-experiment $e$ of $R$, for any $i \in \Nat$, we have $\mathcal{M}_i(e) = \bigcup_{o \in \boxesgeq{R}{i}} \{ j \in \Nat ; k^j \in e^\#(o) \}$, hence $\mathcal{N}_i(e) = \bigcup_{o \in \exactboxes{R}{i}} \{ j \in \Nat ; k^j \in e^\#(o) \}$.
\end{lem}

\begin{example}~\label{example: N_i(e) - 2}
(Continuation of Example~\ref{example: N_i(e)}) We thus have $\mathcal{M}_1(f) = \{ 1, 2 \}$ and $\portsatzerooftype{\cod}{\taylor{f}{i}} \setminus \boxesatzero{\taylor{f}{1}} = \{ o_2, o_4 \}$ with $\arity{\taylor{f}{1}}(o_2)$ $=$ $10^1$ and $\arity{\taylor{f}{1}}(o_4)$ $=$ $10^2$ (see Figure~\ref{fig: T(f)[1]}).
\end{example}

The set $\criticalports{S}{k}{\mathcal{N}_i(e)}$ of ``critical ports'' is a set of exponential ports that will play a crucial role in our algorithm.


\begin{definition}\label{definition: critical}
Let $S$ be a differential $\circ$-PS. Let $ k > 1$. For any $p \in \portsatzero{S}$, we define the sequence $(m_{k, j}(S)(p))_{j \in \Nat} \in \{ 0, \ldots, k-1 \}^{\Nat}$ as follows: $\arity{S}(p) = \sum_{j \in \Nat} m_{k, j}(S)(p) \cdot k^j$. For any $j \in \Nat$, we set $\criticalports{S}{k}{j} = $ $\{ p \in \portsatzero{S} ; m_{k, j}(S)(p) \not= 0 \} \cap \exponentialports{\groundof{S}}$ and, for any $J \subseteq \Nat$, we set $\criticalports{S}{k}{J} =$  $\bigcup_{j \in J} \criticalports{S}{k}{j}$. 
\end{definition}

\begin{example}
We have $\criticalports{S}{10}{1} = \{ p_1, p_4, p_5, p_6, p_7, o_2 \}$ and $\criticalports{S}{10}{2} = \{ p_4, p_5, p_6, p_7, o_4 \}$, where $S$ is the PS of Figure~\ref{fig: T(f)[1]}. So we have $\criticalports{S}{10}{\{ 1, 2\}} = \{ p_1, p_4, p_5, p_6, p_7, o_2, o_4 \}$.
\end{example}

Critical ports are defined by their arities. We show that they are exponential ports that are immediately below the ``new'' boxes:

\begin{prop}\label{prop: critical ports below new boxes}
Let $R$ be a $\circ$-PS. 
Let $k > \Card{\boxes{R}{}}, \cosize{R}$. Let $e$ be a $k$-injective pseudo-experiment of $R$ and let $i \in \Nat$. 
Then we have $\exactboxesatzero{\taylor{e}{i+1}}{i} = \im{\restriction{!_{e, i}}{\mathcal{N}_{i, e}}}$. Furthermore, for any $j \in \mathcal{N}_i(e)$, we have $\im{b_{\taylor{e}{i+1}}(!_{e, i}(j))} = \criticalports{\taylor{e}{i}}{k}{j}$ and, if $!_{e, i}(j) \notin \exactboxesatzero{R}{i}$, then there exist $o_1 \in \boxesatzerogeq{R}{i+1}$ and $e_1 \in e(o_1)$ such that $j \in \mathcal{N}_i(e_1)$. 
In particular, we have $\criticalports{\taylor{e}{i}}{k}{\mathcal{N}_i(e)} \subseteq \exponentialportsatzero{\taylor{e}{i+1}}$.
\end{prop}

\begin{example}
(Continuation of Example~\ref{example: N_i(e) - 2}) We thus have $\im{\restriction{!_{f, 1}}{\mathcal{N}_1(f)}}$ $= \{ o_2, o_4 \}$; and indeed $o_2$ and $o_4$ are the boxes of depth $1$ at depth $0$ of $\taylor{f}{2} \equiv R$ (see Figure~\ref{fig: new_example}). Moreover we have $\criticalports{\taylor{f}{1}}{10}{1} = \{ p_1, p_4, p_5, p_6, p_7, o_2 \}$ and $\criticalports{\taylor{f}{1}}{10}{2} = \{ p_4, p_5, p_6, p_7, o_4 \}$; and indeed, in Figure~\ref{fig: new_example}, we have $\im{b_{\taylor{f}{2}}(o_2)}$ $=$ $\{ p_1, p_4, p_5, p_6, p_7, o_2 \}$ and $\im{b_{\taylor{f}{2}}(o_4)}$ $=$ $\{ p_4, p_5, p_6, p_7, o_4 \}$.
\end{example}

As the following example shows, the information we obtain is already strong, but not strong enough.

\begin{example}
The PS's $R_1$, $R_2$ and $R_3$ of Figures~\ref{fig: R_1}, \ref{fig: R_2} and \ref{fig: R_3} respectively have the same LPS. But if we know that $p \in \im{b_R(o_1)}$, then we know that $R \not= R_3$. Still we are not able to distinguish between $R_1$ and $R_2$.

\begin{figure}
\centering
\begin{minipage}{0.3\textwidth}
\centering
\scalebox{\scalefactter}{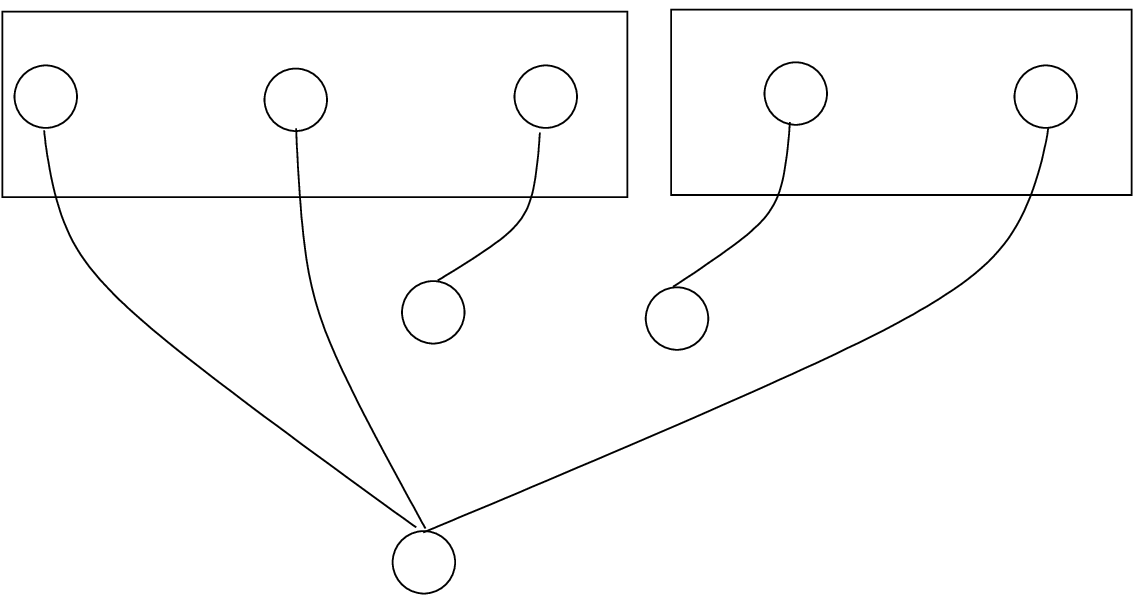}
\caption{$R_1$}
\label{fig: R_1}
\end{minipage}\hfill
\begin{minipage}{0.3\textwidth}
\centering
\scalebox{\scalefactter}{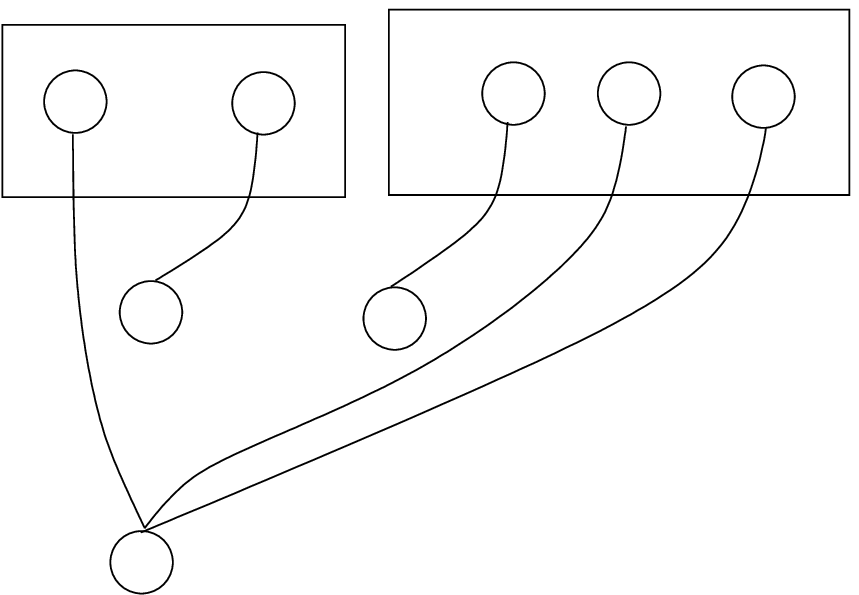}
\caption{$R_2$}
\label{fig: R_2}
\end{minipage}
\begin{minipage}{0.3\textwidth}
\centering
\scalebox{\scalefactter}{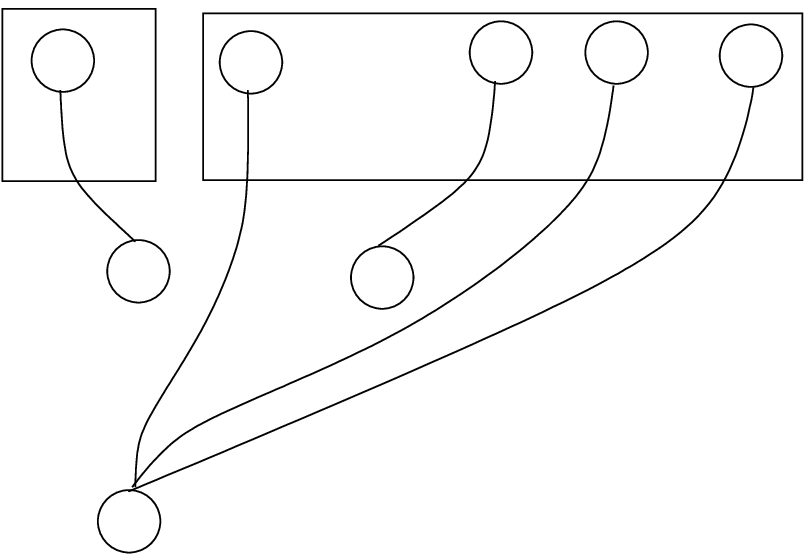}
\caption{$R_3$}
\label{fig: R_3}
\end{minipage}
\end{figure}
\end{example}

\subsection{Connected components}

In order to rebuild the content of the boxes, we introduce our notion of \emph{connected component} (Definition~\ref{defin: connected components}), which uses the auxiliary notions of \emph{substructure} (Definition~\ref{defin: substructure}) and \emph{connected substructure} (Definition~\ref{defin: connected substructure}). A differential $\circ$-PS $R$ is a substructure of a differential $\circ$-PS $S$ (we write $R \sqsubseteq S$) if $R$ is obtained from $S$ by erasing some ports and wires. More precisely:

\begin{definition}\label{defin: substructure}
Let $R$ and $S$ be two differential $\circ$-PS's. Let $\mathcal{Q} \subseteq \exponentialportsatzero{S}$. We write $R \sqsubseteq_{\mathcal{Q}} S$ to denote that $\portsatzero{R} \subseteq \portsatzero{S}$, $\wiresatzero{R} = \{ w \in \wiresatzero{S} \cap (\portsatzero{R} \setminus \mathcal{Q}) ; \target{\groundof{S}}(w) \in \portsatzero{R} \}$, $\labelofcell{\groundof{R}} = \restriction{\labelofcell{\groundof{S}}}{\portsatzero{R}}$, $\target{\groundof{R}} = \restriction{\target{\groundof{S}}}{\wiresatzero{R}}$, $\leftwires{\groundof{R}} = \leftwires{\groundof{S}} \cap \{ w \in \wiresatzero{S} ; \target{\groundof{S}}(w) \in \multiplicativeportsatzero{R} \}$, $\textsf{T}_{\groundof{R}}= \restriction{\textsf{T}_{\groundof{S}}}{\portsatzero{R}}$, 
$\boxesatzero{R} = \boxesatzero{S} \cap \portsatzero{R}$, $B_R = \restriction{B_S}{\boxesatzero{R}}$ and $b_R = \restriction{b_S}{\boxesatzero{R}}$. We write $R \sqsubseteq S$ if there exists $\mathcal{Q}$ such that $R \sqsubseteq_{\mathcal{Q}} S$.  
\end{definition}

\begin{rem}
Let $R \sqsubseteq S$ and $\mathcal{Q} \subseteq \exponentialportsatzero{S}$. We have $R \sqsubseteq_{\mathcal{Q}} S$ iff $\mathcal{Q} \cap \portsatzero{R} \subseteq \conclusions{R}$.
\end{rem}

\begin{rem}
If $R, R' \sqsubseteq_{\mathcal{Q}} S$ and $\portsatzero{R} = \portsatzero{R'}$, then $R = R'$.
\end{rem}

Let us explain with the following example why we sometimes need to erase some wires (so the notion of $\sqsubseteq_{\emptyset}$ is not enough).  

\begin{example}
We need to erase some wires whenever there exist a box $o$ and $p, q \in \im{b_{R'}(o)}$ such that $\target{\groundof{R'}}(q) = p$. Consider, for instance, Figure~\ref{fig: R'}. If $e'$ is a $k$-injective pseudo-experiment of $R'$, then we want to be able to consider the PS $U$ of Figure~\ref{fig: U} as a substructure of $\taylor{e'}{1}$, so we need to erase the wire $q$; we thus have $U \sqsubseteq_{\{ p, q, o \}} \taylor{e'}{1}$.
\begin{figure}
\begin{minipage}{0.25\textwidth}
\centering
\scalebox{\scalefactbis}{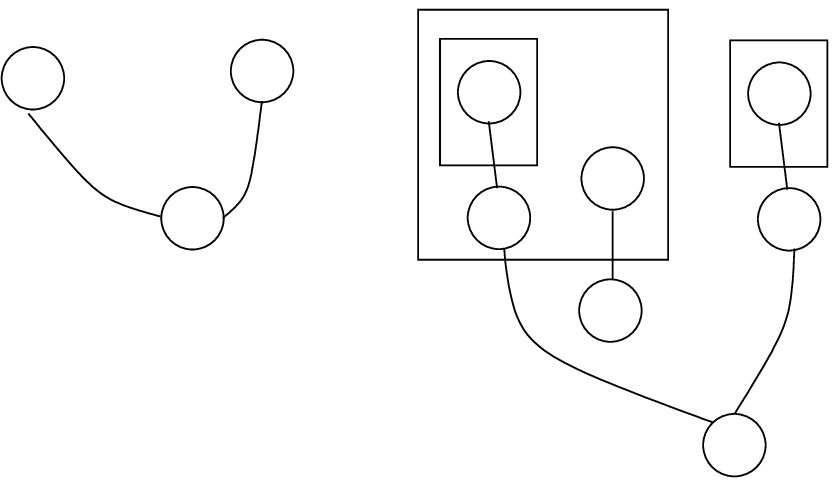}
\caption{PS $S$}
\label{fig: no reconstruction of the experiment}
\end{minipage}
\centering
\begin{minipage}{0.17\textwidth}
\centering
\scalebox{\scalefactnine}{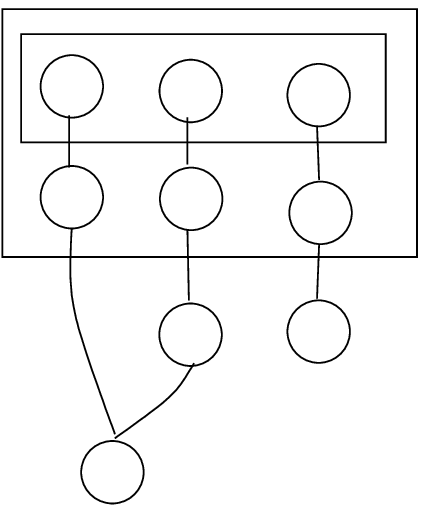}
\caption{PS $R'$}
\label{fig: R'}
\end{minipage}\hfill
\begin{minipage}{0.16\textwidth}
\centering
\scalebox{\scalefactnine}{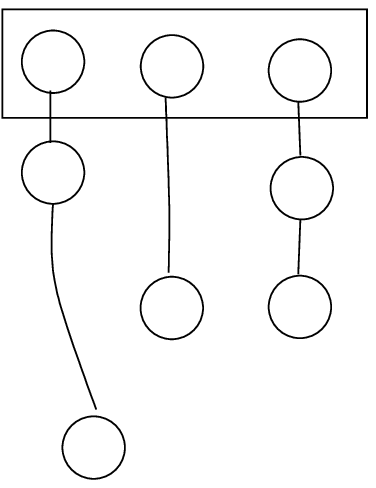}
\caption{PS $U$}
\label{fig: U}
\end{minipage}
\begin{minipage}{0.21\textwidth}
\centering
\scalebox{\scalefactter}{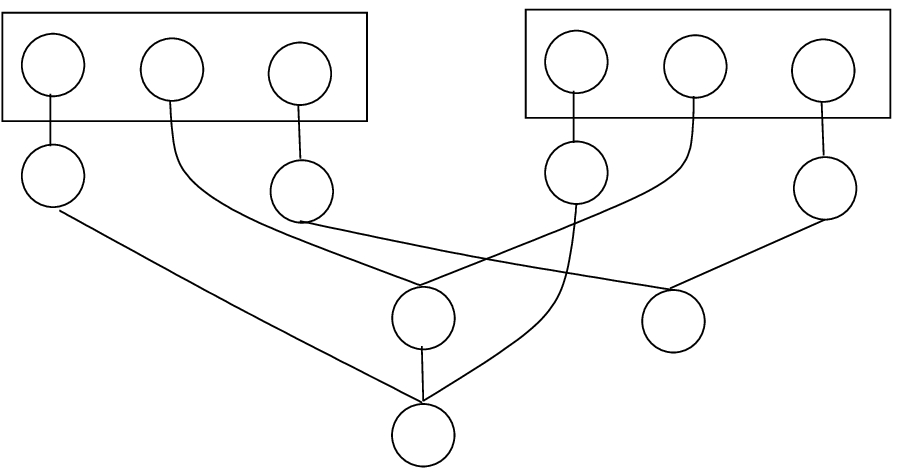}
\caption{PS $\taylor{e'}{1}$}
\label{fig: T(e')[1]}
\end{minipage}
\begin{minipage}{0.14\textwidth}
\centering
\scalebox{\scalefactten}{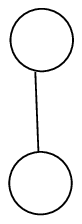}
\caption{PS $T$}
\label{fig: T}
\end{minipage}\hfill
\end{figure}
\end{example}

The relation $\coh_S$ formalizes the notion of ``connectness" between two ports of $S$ at depth $0$. But be aware that, here, ``connected'' has nothing to do with ``connected'' in the sense of \cite{LPSinjectivity}: here, any two doors of the same box are always ``connected''.

\begin{definition}
Let $S$ be a differential $\circ$-PS. We define the binary relation $\coh_S$ on $\portsatzero{S}$ as follows: for any $p, p' \in \portsatzero{S}$, we have $p \coh_S p'$ iff $(p \in \wiresatzero{S}$ and $p' = \target{\groundof{S}}(p))$ or $(p' \in \wiresatzero{S}$ and $p = \target{\groundof{S}}(p'))$ or $((\exists o \in \boxesatzero{S}) \{ p, p' \} \subseteq \im{b_S(o)})$.
\end{definition}

\begin{definition}\label{defin: connected substructure}
Let $S$ and $T$ be two differential $\circ$-PS's. Let $\mathcal{Q} \subseteq \exponentialportsatzero{S}$ such that $T \sqsubseteq_{\mathcal{Q}} S$. We write $T \trianglelefteq_\mathcal{Q} S$ if, for any $p, p' \in \portsatzero{T}$, there exists a finite sequence $(p_0, \ldots, p_n)$ of elements of $\portsatzero{T}$ such that $p_0 = p$, $p_n = p'$ and, for any $j \in \{ 0, \ldots, n-1 \}$, we have $p_j \coh_S p_{j+1}$ and $(p_j \in \mathcal{Q} \Rightarrow j = 0)$.
\end{definition}

\begin{rem}
If $T \trianglelefteq_{\mathcal{Q}} S$, $\mathcal{Q'} \subseteq \exponentialportsatzero{S}$ and $\portsatzero{T} \cap \mathcal{Q'} \subseteq \mathcal{Q}$, then $T \trianglelefteq_{\mathcal{Q'}} S$.
\end{rem}

The sets $\nontrivialconnected{S}{(\mathcal{Q}, \mathcal{Q}_0)}{k}$ of ``components $T$ of $S$ above $\mathcal{Q}$ and $\mathcal{Q}_0$ that are connected \emph{via} other ports than $\mathcal{Q}$ and such that $\cosize{T} < k$'' will play a crucial role in the algorithm of the rebuilding of $\taylor{\overline{e}}{i+1}$ from $\taylor{\overline{e}}{i}$. The reader already knows that, here, ``connected'' has nothing to do with the ``connected proof-nets" of \cite{LPSinjectivity}: there, the crucial tool used was rather the ``bridges'' that put together two doors of the same copy of some box only if they are connected in the LPS of the proof-net.

\begin{definition}\label{defin: connected components}
Let $k \in \Nat$. 
Let $S$ be a differential $\circ$-PS. Let $\mathcal{Q} \subseteq \exponentialportsatzero{S}$ and $\mathcal{Q}_0 \subseteq \portsatzero{S}$. We set 
$\nontrivialconnected{S}{(\mathcal{Q}, \mathcal{Q}_0)}{k} = \left\lbrace T \trianglelefteq_\mathcal{Q} S ; \begin{array}{l} 
 \cosize{T} < k \textit{ and } \conclusions{T} \subseteq \mathcal{Q} \cup \mathcal{Q}_0 \textit{ and} \\ \portsatzero{T} \setminus \mathcal{Q} \not= \emptyset 
\textit{ and } (\forall p \in \portsatzero{T}) (\forall q \in \portsatzero{S}) \\
((p \coh_S q \textit{ and } q \notin \portsatzero{T}) \Rightarrow p \in \mathcal{Q})  \end{array} \right\rbrace$. 
We write also $\nontrivialconnected{S}{\mathcal{Q}}{k}$ instead of $\nontrivialconnected{S}{(\mathcal{Q}, \emptyset)}{k}$ and $\connectedcomponents{S}{k} = \nontrivialconnected{S}{(\conclusionscirc{S}, \conclusionsnotcirc{S})}{k}$.
\end{definition}

A port at depth $0$ of $S$ that is not in $\mathcal{Q}$ cannot belong to two different components:

\begin{fact}\label{fact: two components with a common port}
Let $k \in \Nat$. Let $S$ be a differential $\circ$-PS. Let $\mathcal{Q}, \mathcal{Q}_0 \subseteq \portsatzero{S}$. Let $T, T' \in \nontrivialconnected{S}{(\mathcal{Q}, \mathcal{Q}_0)}{k}$ such that $(\portsatzero{T} \cap \portsatzero{T'}) \setminus \mathcal{Q} \not= \emptyset$. Then $ T = T'$.
\end{fact}

\begin{example}
We have $\Card{\nontrivialconnected{S}{\{ p_1, p_4, p_5, p_6, p_7, o_2 \}}{10}} = 241$ and $\Card{\nontrivialconnected{S}{\{ p_4, p_5, p_6, p_7, o_4 \}}{10}} = 320$, where $S$ is the PS of Figure~\ref{fig: T(f)[1]}.
\end{example}

The operator $\sum$ glues together several $\circ$-PS's that share only $\circ$-conclusions:

\begin{definition}
Let $\mathcal{U}$ be a set of $\circ$-PS's. We say that $\mathcal{U}$ is \emph{gluable} if, for any 
$R, S \in \mathcal{U}$ s.t. $R \not= S$, we have $\portsatzero{R} \cap \portsatzero{S} \subseteq \conclusionscirc{R} \cap \conclusionscirc{S}$. 
If $\mathcal{U}$ is gluable, then $\sum \mathcal{U}$ is the $\circ$-PS such that $\conclusionscirc{\sum \mathcal{U}} = \bigcup \{ \conclusionscirc{R} ; R \in \mathcal{U} \}$ obtained by glueing all the elements of $\mathcal{U}$.
\end{definition}


The set $\connectedcomponents{R}{k}$ (for $k$ big enough) is an alternative way to describe a $\circ$-PS $R$:

\begin{fact}\label{fact: connected components}
Let $R$ be a $\circ$-PS. Let $k > \cosize{R}$. We have $R = \sum \connectedcomponents{R}{k}$.
\end{fact}

Definition~\ref{defin: overline} allows to formalize the operation of ``putting a connected component inside a box'', which will be useful for building the boxes of depth $i$  of $\taylor{\overline{e}}{i+1}$: from some \emph{boxable} differential $\circ$-PS $R \sqsubseteq \taylor{\overline{e}}{i}$, we build a $\circ$-PS $\overline{R}$ such that, for some $o \in \exactboxesatzero{\taylor{\overline{e}}{i+1}}{i}$, there exists $T \in \connectedcomponents{B_{\taylor{\overline{e}}{i+1}}(o)}{k}$ such that $T \simeq \overline{R}$.


\begin{definition}\label{defin: overline}
Let $R$ be a differential $\circ$-PS. If $\conclusions{R} \subseteq \exponentialportsatzero{R}$, $\conclusions{R} \cap \boxesatzero{R} = \emptyset$ and $\portsatzerooftype{\cod}{R} \setminus \boxesatzero{R} \subseteq \conclusions{R}$, then one says that $R$ is \emph{boxable} and we define a $\circ$-PS $\overline{R}$ s.t. $\portsatzero{\overline{R}} \subseteq \portsatzero{R}$, $\conclusionsnotcirc{\overline{R}} \subseteq \wiresatzero{R}$ and $\conclusionscirc{\overline{R}} = \conclusions{R} \cap \portsatzero{\overline{R}}$ as follows:
\begin{itemize}
\item $\portsatzero{\overline{R}} = \wiresatzero{R} \cup \bigcup_{o \in \boxesatzero{R}} (\im{b_R(o)} \cap \conclusions{R})$;
\item $\wiresatzero{\overline{R}} = \{ w \in \wiresatzero{R} ; \target{\groundof{R}}(w) \in \wiresatzero{R} \}$
\item $\labelofcell{\groundof{\overline{R}}}(p) = \left\lbrace \begin{array}{ll} 
\labelofcell{\groundof{R}}(p) & \textit{if $p \in \wiresatzero{R}$;}\\
\circ & \textit{otherwise;}
\end{array} \right.$
\item $\target{\groundof{\overline{R}}} = \restriction{\target{\groundof{R}}}{\wiresatzero{\overline{R}}}$; $\leftwires{\groundof{\overline{R}}} = \leftwires{\groundof{R}} \cap \wiresatzero{\overline{R}}$; $\boxesatzero{\overline{R}} = \boxesatzero{R}$; $b_{\overline{R}} = b_R$.
\end{itemize}
If $\mathcal{U}$ is a set of boxable differential $\circ$-PS's, 
then we set $\overline{\mathcal{U}} = \{ \overline{R} ; R \in \mathcal{U} \}$.
\end{definition}

In the proof of the following proposition, we finally describe the complete algorithm leading from $\taylor{e}{i}$ to $\taylor{e}{i+1}$.
 Informally: for every $j_0 \in \mathcal{N}_i(e)$, for every equivalence class $\mathfrak{T} \in \nontrivialconnected{\taylor{e}{i}}{\criticalports{\taylor{e}{i}}{k}{j_0}}{k}_{/_\equiv}$, if $\Card{\mathfrak{T}} = \sum_{j \in \Nat} m_j \cdot k^j$ (with $0 \leq m_j <k$), then we remove $m_{j_0} \cdot k^{j_0}$ elements of $\mathfrak{T}$ from $\groundof{\taylor{e}{i}}$ and we put $m_{j_0}$ such elements inside the (new) box $!_{e, i}(j_0)$ of depth $i$. For every $j_0 \in \mathcal{N}_i(e)$, the set $\mathcal{U}_{j_0}$ of the proof is the union of the sets of such $m_{j_0}$ elements for all the equivalence classes $\mathfrak{T} \in \nontrivialconnected{\taylor{e}{i}}{\criticalports{\taylor{e}{i}}{k}{j_0}}{k}_{/_\equiv}$.

\begin{prop}\label{prop: from i to i+1}
Let $R$ and $R'$ be two PS's. Let $k > \cosize{R}$, $\cosize{R'}$, $\Card{\boxes{R}}$, $\Card{\boxes{R'}}$. Let $e$ be a $k$-injective pseudo-experiment of $R$ and let $e'$ be a $k$-injective pseudo-experiment of $R'$ s.t. $\taylor{e}{0} \equiv \taylor{e'}{0}$. Then, for any $i \in \Nat$, we have $\taylor{e}{i} \equiv \taylor{e'}{i}$.
\end{prop}

\begin{proof}
(Sketch) By induction on $i$. We assume that $\taylor{e}{i} \equiv \taylor{e'}{i}$. We set $\mathcal{M} = \mathcal{M}_i(e) = \mathcal{M}_i(e')$ and $\mathcal{N} = \mathcal{N}_i(e) = \mathcal{N}_i(e')$.

Let $S \equiv \taylor{e}{i}$. There is a bijection $\cod : \mathcal{M} \to \portsatzerooftype{\cod}{S} \setminus \boxesatzero{S}$ such that, for any $j \in \mathcal{M}$,  we have  $(\arity{S} \circ \cod)(j) = k^j$. 
For any $j \in \mathcal{N}$, we set $\mathcal{K}_j = \criticalports{S}{k}{j}$ and $\mathcal{T}_{j} = \nontrivialconnected{S}{\mathcal{K}_{j}}{k}$. We set $\mathcal{T} = \bigcup_{j \in \Nat} \mathcal{T}_j$. For any $T \in \mathcal{T}$, we define $(m_j^T)_{j \in \Nat} \in {\{ 0, \ldots, k-1 \}}^\Nat$ as follows: $\Card{\{ T' \in \mathcal{T} ; T \equiv T'\}} = \sum_{j \in \Nat} m_j^T \cdot k^j$. We set $\mathcal{P} =  \{ p \in \portsatzero{S} ; $ $p \notin \bigcup_{j \in \mathcal{N}} \bigcup_{\begin{array}{c} T \in \nontrivialconnected{S}{\mathcal{K}_j}{k} \end{array}} \portsatzero{T} \}$. 
For any $j \in \mathcal{N}$, we are given $\mathcal{U}_j \subseteq \mathcal{T}_j$ such that, for any $T \in \mathcal{T}_j$, we have $\Card{\{ T' \in \mathcal{U}_j ; T' \equiv T \}} = m_j^T$. Let $S'$ be some differential PS such that $\groundof{\restriction{S}{\mathcal{P}}} \sqsubseteq_\emptyset \groundof{S'} \sqsubseteq_\emptyset \groundof{S}$, where $\restriction{S}{\mathcal{P}}$ is the unique $S_0 \sqsubseteq_{\emptyset} S'$ s.t. $\portsatzero{S_0} = \mathcal{P}$, and:
\begin{itemize}
\item $\bigcup_{j_0 \in \mathcal{N}} \mathcal{T'}_{j_0} \subseteq \bigcup_{j_0 \in \mathcal{N}} \mathcal{T}_{j_0}$
\item for any $T \in \mathcal{T}$, we have 
$\Card{\{ T' \in \bigcup_{j \in \Nat} \nontrivialconnected{S'}{\mathcal{K}_{j}}{k} ; T' \equiv T \}} = \sum_{j \notin \mathcal{N}} m_j^T \cdot k^j$;
\item $\boxesatzero{S'} = (\boxesatzero{S} \cap \portsatzero{S'}) \cup \im{\restriction{\cod}{\mathcal{N}}}$
\item for any $o_1 \in \boxesatzero{S} \cap \portsatzero{S'}$, we have $B_{S'}(o_1) = B_S(o_1)$ and $b_{S'}(o_1) = b_S(o_1)$
\item for any $j \in \mathcal{N}$, there exists $\rho^j : B_{S'}(\cod(j)) \simeq \sum \overline{\mathcal{U}_j}$ such that $b_{S'}(\cod(j))(q) =$\\ $\left\lbrace \begin{array}{ll} 
{\groundof{\rho^{j}}}(q) & \textit{if $q \in \conclusionscirc{B_{S'}(\cod(j))}$;}\\
\target{\groundof{S}}(\groundof{\rho^{j}}(q)) & \textit{if $q \in \conclusionsnotcirc{B_{S'}(\cod(j))}$.}
\end{array} \right.$
\end{itemize}
Then 
one can show that $\taylor{e}{i+1} \equiv S' \equiv \taylor{e'}{i+1}$.
\end{proof}

\begin{example}
Consider Figure~\ref{fig: T(f)[1]}. We set $j_0 = 1$. We have $o_2 = \cod_{f, 1}(j_0)$. We set $\mathcal{T}_{j_0} = \nontrivialconnected{\taylor{f}{1}}{\criticalports{\taylor{f}{1}}{10}{1}}{10}$. Let $T$ be the $\circ$-PS of Figure~\ref{fig: T}. We have $T \in \mathcal{T}$. We have $\Card{\{ T' \in \mathcal{T} ; T' \equiv T \}} = 1 \cdot 10^0 + 1 \cdot 10^1$. We set $\mathcal{U}_{j_0} = \{ T \}$. The $\circ$-PS $\sum \overline{\mathcal{U}_{j_0}} = \overline{T}$ is the $\circ$-PS of depth $0$ that consists of only one port: the port $p$ labelled by $\bot$. There exists $\rho^{j_0}: B_{\taylor{f}{2}}(o_2) \simeq \sum \overline{\mathcal{U}_{j_0}}$ such that $b_{\taylor{f}{2}}(o_2)(q_2) = p_1 = \target{\groundof{\taylor{f}{1}}}(q)$ defined by $\groundof{\rho^{j_0}}(q_2) = q$ (see Figure~\ref{fig: new_example}).
\end{example}

\subsection{Injectivity}

\begin{theorem}\label{thm: injectivity}
Let $R$ and $R'$ be two PS's s.t. $\conclusions{R} = \conclusions{R'}$. If $\sm{R} = \sm{R'}$, then $R \equiv R'$.
\end{theorem}

\begin{proof}
We set $d = $ $\max{\{ \depthof{R}, \depthof{R'} \}}$. 
For any $k > 1$, there exist a $k$-injective experiment $e$ of $R$ and a $k$-injective experiment $e'$ of $R'$\footnote{This is not necessary true for the multiset based coherent semantics.} such that $\restriction{e}{\conclusions{R}} = \restriction{e'}{\conclusions{R'}} \in \sm{R} \cap \sm{R'}$. By Lemma~\ref{lem: taylor expansion}, we have $\taylor{\overline{e}}{0} \equiv \taylor{\overline{e'}}{0}$. Therefore if $k > \cosize{R}, \Card{\boxes{R}}$, then, by Proposition~\ref{prop: from i to i+1}, we have $\taylor{\overline{e}}{d} \equiv \taylor{\overline{e'}}{d}$. Now, by Fact~\ref{fact: Taylor[i] with i large}, we have $R = \taylor{\overline{e}}{d} \equiv \taylor{\overline{e'}}{d} = R'$.
\end{proof}

\begin{rem}
From any $1$-point of $\sm{R}$ (i.e. the result of some $1$-experiment of $R$) one can recover $\cosize{R}$ and $\Card{\boxes{R}}$. This remark shows that for characterizing $R$, two points are enough: a $1$-point of its interpretation, from which one can bound $\cosize{R}$ and $\Card{\boxes{R}}$, and a $k$-injective point of its interpretation with $k > \cosize{R},  \Card {\boxes{R}}$.
\end{rem}

\begin{rem}\label{remark: with axioms}
If we want to extend our theorem to PS's with axioms, then we assume that the interpretation of any ground type is an infinite set and we consider a $k$-injective experiment $e$ that is ``injective'' in the sense that every atom (the atoms are the elements of the interpretations of the ground types) occurring in $\im{\restriction{e}{\conclusions{R}}}$ occurs exactly twice.
\end{rem}

\begin{rem}\label{remark: untyped framework}
In an untyped framework with axioms, we need to add the constraint on the injective $k$-injective point one considers to be {$\llbracket R \rrbracket$-atomic}, i.e. a point of $\llbracket R \rrbracket$ that cannot be obtained from another point of $\llbracket R \rrbracket$ by some substitution that is not a renaming. Atomic points are results of atomic experiments (experiments that label axioms with atoms) - the converse does not necessarily hold.
\end{rem}

\bibliography{ll_new}

\onecolumn

\appendix

\section{Proof of Lemma~\ref{lem: M_i}}


\begin{fact}\label{fact: M_0 pairwise disjoint}
Let $R$ be a $\circ$-PS. 
Let $k > 1$. Let $e$ be a $k$-injective pseudo-experiment of $R$. Then 
\begin{itemize}
\item for any $o, o' \in \boxesatzero{R}$, for any $e' \in e(o')$, we have $(\forall j \in \mathcal{M}_0(e')) e^\#(o) \not= \{ k^j \}$;
\item and, for any $o_1, o_2 \in \boxesatzero{R}$, for any $e_1 \in e(o_1)$, for any $e_2 \in e(o_2)$, we have $\mathcal{M}_0(e_1) \cap \mathcal{M}_0(e_2) \not= \emptyset \Rightarrow (o_1 = o_2 \textit{ and } e_1 = e_2)$.
\end{itemize}
\end{fact}

\begin{proof}
Let $o, o' \in \boxesatzero{R}$ and let $e' \in e(o')$. For any $j \in \mathcal{M}_0(e')$, there exists $o'' \in \boxes{B_R(o')}$ such that $k^j \in e'^\#(o'')$. So, if there exists $j \in \mathcal{M}_0(e')$ such that $e^\#(o) = \{ k^j \}$, then there exists $o'' \in \boxes{R}(o')$ such that $e^\#(o':o'') \cap e^\#(o) \not= \emptyset$; but, by the definition of \emph{$k$-injective experiment} (Definition~\ref{defin: k-injective}), this entails that $o = o':o''$; now, since $\emptysequence \notin \boxes{B_R(o'')}$, we obtain a contradiction with the following requirement of the definition of $\circ$-PS's (Definition~\ref{defin: differential PS}): no $p \in \ports{\groundof{R}}$ is a sequence.
\end{proof}

\begin{lem}\label{lem: M_1}
Let $k > \Card{\boxes{R}{}}$. 
For any $k$-injective pseudo-experiment $e$ of $R$, we have 
\begin{itemize}
\item if $\textit{depth}(R) \leq 1$, then $\mathcal{M}_1(e) = \emptyset$
\item $m_{0, 0}(e) = \Card{\boxesatzero{R}}$
\item 
  $\mathcal{M}_1(e) = \bigcup_{o \in \boxesatzerogeq{R}{1}} \{ j \in \Nat ; e^\#(o) = \{ k^j \} \} \cup \bigcup_{o \in \boxesatzerogeq{R}{2}} \bigcup_{e' \in e(o)} \mathcal{M}_1(e')$
\item $\mathcal{M}_1(e) \subseteq \mathcal{M}_0(e)$.
\end{itemize}
\end{lem}

\begin{proof}
The first item follows from the fact that if $\textit{depth}(R) = 0$, then $\Card{\mathcal{M}_0(e)} = 0$ and if $\textit{depth}(R) = 1$, then $\Card{\mathcal{M}_0(e)} = \Card{\boxes{R}{}} < k$. 

We prove the three other items by induction on $\textit{depth}(R)$. We have:
\begin{eqnarray*}
& & \Card{\mathcal{M}_0(e)} \\
& = & \Card{\bigcup_{o \in \boxesatzero{R}} \{ j \in \Nat ; k^j \in e^\#(o) \}} \allowdisplaybreaks \\
& = & \sum_{o \in \boxes{R}{}} \Card{\{ j \in \Nat; k^j \in e^\#(o) \}} \allowdisplaybreaks \\
& = & \sum_{o \in \boxesatzero{R}} \Card{\{ j \in \Nat ; k^j \in e^\#(o) \}} + \sum_{o \in \boxesatzerogeq{R}{1}} \sum_{o' \in \boxes{B_R(o)}{}} \Card{\{ j \in \Nat; k^j \in e^\#(o:o') \}} \allowdisplaybreaks \\
& = & \Card{\boxesatzero{R}} + \sum_{o \in \boxesatzerogeq{R}{1}} \sum_{o' \in \boxes{B_R(o)}{}} \Card{\bigcup_{e' \in e(o)} \{ j \in \Nat; k^j \in {e'}^\#(o') \}} \allowdisplaybreaks \\
& = & \Card{\boxesatzero{R}} + \sum_{o \in \boxesatzerogeq{R}{1}} \sum_{o' \in \boxes{B_R(o)}{}} \sum_{e' \in e(o)} \Card{\{ j \in \Nat; k^j \in {e'}^\#(o') \}} \allowdisplaybreaks \\
& = & \Card{\boxesatzero{R}} + \sum_{o \in \boxesatzerogeq{R}{1}} \sum_{e' \in e(o)} \sum_{o' \in \boxes{B_R(o)}{}} \Card{\{ j \in \Nat; k^j \in {e'}^\#(o') \}} \allowdisplaybreaks \\
& = & \Card{\boxesatzero{R}} + \sum_{o \in \boxesatzerogeq{R}{1}} \sum_{e' \in e(o)} \Card{\bigcup_{o' \in \boxes{B_R(o)}{}} \{ j \in \Nat; k^j \in {e'}^\#(o') \}} \\
& & \textit{(since $(\forall o \in \boxesatzerogeq{R}{1}) (\forall e' \in e(o)) (\forall o' \in \boxes{B_R(o)}{}) {e'}^\#(o') \subseteq e^\#(o:o')$)} \allowdisplaybreaks \\
& = & \Card{\boxesatzero{R}} + \sum_{o \in \boxesatzerogeq{R}{1}} \sum_{e' \in e(o)} \Card{\mathcal{M}_0(e')} \allowdisplaybreaks \\
& = & \Card{\boxesatzero{R}} + \sum_{o \in \boxesatzerogeq{R}{1}} \sum_{e' \in e(o)} (m_{0, 0}(e') + \sum_{j \in \mathcal{M}_1(e')} m_{0, j}(e') \cdot k^j) \allowdisplaybreaks \\
& = & \Card{\boxesatzero{R}} + \sum_{o \in \boxesatzerogeq{R}{1}} \sum_{e' \in e(o)} (\Card{\boxesatzero{B_R(o)}} + \sum_{j \in \mathcal{M}_1(e')} m_{0, j}(e') \cdot k^j)  \\
& & \textit{(by the induction hypothesis)} \allowdisplaybreaks \\
& = & \Card{\boxesatzero{R}} + \sum_{o \in \boxesatzerogeq{R}{1}} (\Card{\boxesatzero{B_R(o)}} \cdot \Card{e(o)} + \sum_{e' \in e(o)} \sum_{j \in \mathcal{M}_1(e')} m_{0, j}(e') \cdot k^j) \allowdisplaybreaks \\
& = & \Card{\boxesatzero{R}} + \sum_{o \in \boxesatzerogeq{R}{1}} \Card{\boxesatzero{B_R(o)}} \cdot \Card{e(o)} + \sum_{o \in \boxesatzerogeq{R}{2}} \sum_{e' \in e(o)} \sum_{j \in \mathcal{M}_1(e')} m_{0, j}(e') \cdot k^j \\
& & \textit{(by the first item).}
\end{eqnarray*}
By the induction hypothesis, for any $o \in \boxesatzero{R}$, for any $e' \in e(o)$, we have $\mathcal{M}_1(e') \subseteq \mathcal{M}_0(e')$, hence, by Fact~\ref{fact: M_0 pairwise disjoint},
\begin{itemize}
\item for any $o \in \boxesatzero{R}$, for any $o' \in \boxesatzero{R}$, for any $e' \in e(o')$, we have $(\forall j \in \mathcal{M}_1(e')) e^\#(o) \not= \{ k^j \}$;
\item and, for any $o_1, o_2 \in \boxesatzero{R}$, for any $e_1 \in e(o_1)$, for any $e_2 \in e(o_2)$, we have $\mathcal{M}_1(e_1) \cap \mathcal{M}_1(e_2) \not= \emptyset \Rightarrow (o_1 = o_2 \textit{ and } e_1 = e_2)$.
\end{itemize}
We obtain $\mathcal{M}_1(e) = \bigcup_{o \in \boxesatzerogeq{R}{1}} \{ j \in \Nat ; e^\#(o) = \{ k^j \} \} \cup \bigcup_{o \in \boxesatzerogeq{R}{2}} \bigcup_{e' \in e(o)} \mathcal{M}_1(e')$. Since, by the induction hypothesis, for any $o \in \boxesatzerogeq{R}{2}$, for any $e' \in e(o)$, we have $\mathcal{M}_1(e') \subseteq \mathcal{M}_0(e')$, we obtain $\mathcal{M}_1(e) \subseteq \mathcal{M}_0(e)$.
\end{proof}


\begin{lem}\label{lemma: M_j(e)}
Let $R$ be a $\circ$-PS. 
Let $k > \Card{\boxes{R}{}}$. 
For any $k$-injective pseudo-experiment $e$ of $R$, for any $i \in \Nat$, we have 
\begin{itemize}
\item for any $o, o' \in \boxesatzero{R}$, for any $e' \in e(o')$, we have $(\forall j \in \mathcal{M}_i(e')) e^\#(o) \not= \{ k^j \}$;
\item for any $o_1, o_2 \in \boxesatzero{R}$, for any $e_1 \in e(o_1)$, for any $e_2 \in e(o_2)$, we have $\mathcal{M}_i(e_1) \cap \mathcal{M}_i(e_2) \not= \emptyset \Rightarrow (o_1 = o_2 \textit{ and } e_1 = e_2)$;
\item if $\textit{depth}(R) \leq i+1$, then $\mathcal{M}_{i+1}(e) = \emptyset$;
\item $m_{i, 0}(e) = \Card{\boxesatzerogeq{R}{i}}$;
\item $\mathcal{M}_{i+1}(e) = \bigcup_{o \in \boxesatzerogeq{R}{i+1}} \{ j \in \Nat ; e^\#(o) = \{ k^j \} \} \cup \bigcup_{o \in \boxesatzerogeq{R}{i+2}} \bigcup_{e' \in e(o)} \mathcal{M}_{i+1}(e')$;
\item and $\mathcal{M}_{i+1}(e) \subseteq \mathcal{M}_{i}(e)$.
\end{itemize}
\end{lem}

\begin{proof}
By induction on $i$. 

If $i = 0$, then the items hold by Fact~\ref{fact: M_0 pairwise disjoint} and Lemma~\ref{lem: M_1}. Now, we assume that they hold for some $i \in \Nat$.

Since $\mathcal{M}_{i+1}(e) \subseteq \mathcal{M}_i(e)$, we have 
\begin{enumerate}
\item for any $o, o' \in \boxesatzero{R}$, for any $e' \in e(o')$, we have $(\forall j \in \mathcal{M}_{i+1}(e')) e^\#(o) \not= \{ k^j \}$;
\item and, for any $o_1, o_2 \in \boxesatzero{R}$, for any $e_1 \in e(o_1)$, for any $e_2 \in e(o_2)$, we have $\mathcal{M}_{i+1}(e_1) \cap \mathcal{M}_{i+1}(e_2) \not= \emptyset \Rightarrow (o_1 = o_2 \textit{ and } e_1 = e_2)$.
\end{enumerate}
If $\textit{depth}(R) \leq i+2$, then, since, $\mathcal{M}_{i+1}(e) = \bigcup_{o \in \boxesatzerogeq{R}{i+1}} \{ j \in \Nat ; e^\#(o) = \{ k^j \} \} \cup \bigcup_{o \in \boxesatzerogeq{R}{i+2}} \bigcup_{e' \in e(o)} \mathcal{M}_{i+1}(e')$, 
we obtain $\Card{\mathcal{M}_{i+1}(e)} = \Card{\bigcup_{o \in \boxesatzerogeq{R}{i+1}} \{ j \in \Nat ; e^\#(o) = \{ k^j \} \}} < k$, hence $\mathcal{M}_{i+2}(e) = \emptyset$.

Now, we prove, by by induction on $\textit{depth}(R)$, that:
\begin{itemize}
\item $m_{i+1, 0}(e) = \Card{\boxesatzerogeq{R}{i+1}}$
\item $\mathcal{M}_{i+2}(e) = \bigcup_{o \in \boxesatzerogeq{R}{i+2}} \{ j \in \Nat ; e^\#(o) = \{ k^j \} \} \cup \bigcup_{o \in \boxesatzerogeq{R}{i+3}} \bigcup_{e' \in e(o)} \mathcal{M}_{i+2}(e')$
\item and $\mathcal{M}_{i+2}(e) \subseteq \mathcal{M}_{i+1}(e)$.
\end{itemize}
We have
\begin{eqnarray*}
& & \Card{\mathcal{M}_{i+1}(e)} \\
& = & \Card{\boxesatzerogeq{R}{i+1}} + \sum_{o \in \boxesatzerogeq{R}{i+2}} \sum_{e' \in e(o)} \Card{\mathcal{M}_{i+1}(e')} \\
& & \textrm{(by 1. and 2.)}\\
& = & \Card{\boxesatzerogeq{B_R(o)}{i+1}} + \sum_{o \in \boxesatzerogeq{R}{i+2}} \sum_{e' \in e(o)} (m_{i+1, 0}(e') + \sum_{j \in \mathcal{M}_{i+2}(e')} m_{i+1, j}(e') \cdot k^j)\\
& = & \Card{\boxesatzerogeq{R}{i+1}} + \sum_{o \in \boxesatzerogeq{R}{i+2}} \sum_{e' \in e(o)} (\Card{\boxesatzerogeq{B_R(o)}{i+1}} + \sum_{j \in \mathcal{M}_{i+2}(e')} m_{i+1, j}(e') \cdot k^j)\\
& & \textrm{(by the induction hypothesis)} \allowdisplaybreaks \\
& = & \Card{\boxesatzerogeq{R}{i+1}} + \sum_{o \in \boxesatzerogeq{R}{i+2}} (\Card{\boxesatzerogeq{B_R(o)}{i+1}} \cdot \Card{e(o)} + \sum_{e' \in e(o)} \sum_{j \in \mathcal{M}_{i+2}(e')} m_{i+1, j}(e') \cdot k^j) \allowdisplaybreaks \\
& = & \Card{\boxesatzerogeq{R}{i+1}} + \sum_{o \in \boxesatzerogeq{R}{i+2}} \Card{\boxesatzerogeq{B_R(o)}{i+1}} \cdot \Card{e(o)} \\
& & + \sum_{o \in \boxesatzerogeq{R}{i+3}} \sum_{e' \in e(o)} \sum_{j \in \mathcal{M}_{i+2}(e')} m_{i+1, j}(e') \cdot k^j \\
& & \textit{(since, if $\textit{depth}(B_R(o)) = i+2$ and $e' \in e(o)$, then $\mathcal{M}_{i+2}(e') = \emptyset$).}
\end{eqnarray*}
By the induction hypothesis, for any $o \in \boxesatzerogeq{R}{i+3}$, for any $e' \in e(o)$, we have $\mathcal{M}_{i+2}(e') \subseteq \mathcal{M}_{i+1}(e')$, hence, by 1. and 2., 
\begin{itemize}
\item for any $o, o' \in \boxesatzero{R}$, for any $e' \in e(o')$, we have $(\forall j \in \mathcal{M}_{i+2}(e')) e^\#(o) \not= \{ k^j \}$;
\item and, for any $o_1, o_2 \in \boxesatzero{R}$, for any $e_1 \in e(o_1)$, for any $e_2 \in e(o_2)$, we have $\mathcal{M}_{i+2}(e_1) \cap \mathcal{M}_{i+2}(e_2) \not= \emptyset \Rightarrow (o_1 = o_2 \textit{ and } e_1 = e_2)$.
\end{itemize}
We obtain 
\begin{itemize}
\item $m_{i+1, 0}(e) = \Card{\boxesatzerogeq{R}{i+1}}$
\item and $\mathcal{M}_{i+2}(e) = \bigcup_{o \in \boxesatzerogeq{R}{i+2}} \{ j \in \Nat ; e^\#(o) = \{ k^j \} \} \cup \bigcup_{o \in \boxesatzerogeq{R}{i+3}} \bigcup_{e' \in e(o)} \mathcal{M}_{i+2}(e')$.
\end{itemize}
Since, by the induction hypothesis, for any $o \in \boxesatzerogeq{R}{i+3}$, for any $e' \in e(o)$, we have $\mathcal{M}_{i+2}(e') \subseteq \mathcal{M}_{i+1}(e')$, we obtain $\mathcal{M}_{i+2}(e) \subseteq \mathcal{M}_{i+1}(e)$.
\end{proof}

\textbf{Proof of Lemma~\ref{lem: M_i}}

\begin{proof}
By induction on $\textit{depth}(R)$. We have
\begin{eqnarray*}
\mathcal{M}_i(e) & = & \bigcup_{o \in \boxesatzerogeq{R}{i}} \{ j \in \Nat ; e^\#(o) = \{ k ^j \} \} \cup \bigcup_{o \in \boxesatzerogeq{R}{i+1}} \bigcup_{e' \in e(o)} \mathcal{M}_i(e')\\
& & \textit{(by Lemma~\ref{lemma: M_j(e)})} \allowdisplaybreaks\\
& = & \bigcup_{o \in \boxesatzerogeq{R}{i}} \{ j \in \Nat ; e^\#(o) = \{ k^j \} \} \cup \bigcup_{o \in \boxesatzerogeq{R}{i+1}} \bigcup_{e' \in e(o)} \bigcup_{o' \in \boxesgeq{B_R(o)}{i}} \{ j \in \Nat ; k^j \in {e'}^\#(o') \}\\
& & \textit{(by the induction hypothesis)} \allowdisplaybreaks\\
& = & \bigcup_{o \in \boxesatzerogeq{R}{i}} \{ j \in \Nat ; e^\#(o) = \{ k^j \} \} \cup \bigcup_{o \in \boxesatzerogeq{R}{i+1}} \bigcup_{o' \in \boxesgeq{B_R(o)}{i}} \{ j \in \Nat ; k^j \in {e}^\#(o:o') \}\\
& = & \bigcup_{o \in \boxesgeq{R}{i}} \{ j \in \Nat ; k^j \in e^\#(o) \} .
\end{eqnarray*}
\end{proof}

\section{Proof of Proposition~\ref{prop: critical ports below new boxes}}

We prove Proposition~\ref{prop: critical ports below new boxes} through Lemma~\ref{lem: arity{taylor{e}{i}}(c) for c at depth 0}.

\begin{definition}
Let $R$ be a $\circ$-PS. 
For any $q \in \portsatzero{R}$, for any $i \in \Nat$, we set 
$\arity{R, i}(q) = \arity{\groundof{R}}(q) + \sum_{\substack{o \in \boxesatzerosmaller{R}{i}\\ q \in \conclusionscirc{B_R(o)}}} \arity{B_R(o)}(q)$\\
$+ \sum_{o \in \boxesatzerosmaller{R}{i}} \Card{\left\lbrace p \in \conclusionsnotcirc{B_R(o)} ; b_{R}(o)(p) = q  \right\rbrace}$. 

For any $p \in \exponentialports{\groundof{R}}$, for any $i \in \Nat$, we define, by induction on $\depthof{R}$, a subset $\mathcal{B}_R^{\geq i}(p)$ of $\boxesgeq{R}{i}$: we set 
\begin{eqnarray*}
\mathcal{B}_R^{\geq i}(p) & = & \{ o_1 \in \boxesatzerogeq{R}{i} ; p \in \im{\restriction{b_R(o_1)}{\conclusionsnotcirc{B_R(o_1)}}}  \\
& & \textit{ or } (p \in \conclusionscirc{B_R(o_1)} \textit{ and } a_{B_R(o_1), i}(p) > 0) \}\\ 
& & \cup \bigcup_{\substack{o_1 \in \boxesatzerogeq{R}{i+1}\\ p \in \conclusionscirc{B_R(o_1)}}} \{ o_1 : o; o \in \mathcal{B}_{B_R(o_1)}^{\geq i}(p) \}
\end{eqnarray*}
\end{definition}

\begin{rem}
If $o_1 \in \exactboxesatzero{R}{i}$,  then $\im{b_R(o_1)} = \{ p \in \exponentialports{\groundof{R}} ; $ $o_1 \in \mathcal{B}_R^{\geq i}(p) \}$.
\end{rem}

\begin{lem}\label{lem: arity{taylor{e}{i}}(c) for c at depth 0}
Let $R$ be a $\circ$-PS. Let $k > \cosize{R}$. Let $e$ be a $k$-injective pseudo-experiment of $R$. Let $q \in \portsatzero{R}$. Let $i \in \Nat$. Then, for any $j > 0$, we have $q \in \criticalports{\taylor{e}{i}}{k}{j}$ if, and only if, there exists $o \in \mathcal{B}_R^{\geq i}(q)$ such that $k^j \in e^\#(o)$. Moreover we have $\arity{\taylor{e}{i}}(q) \textit{ mod } k = \arity{R, i}(q)$.
\end{lem}

\subsection{Proof of Lemma~\ref{lem: arity{taylor{e}{i}}(c) for c at depth 0}}

\begin{fact}\label{fact: arity}
Let $R$ be a $\circ$-PS. Let $e$ be a pseudo-experiment of $R$. Let $p \in \portsatzero{R}$. Let $i \in \Nat$. 
If $p \notin \conclusionscirc{R}$, then
\begin{eqnarray*}
\arity{\taylor{e}{i}}(p) & = & \arity{R, i}(p) \\
& & + \sum_{o_1 \in \boxesatzerogeq{R}{i}} \Card{\{ q \in \conclusionsnotcirc{B_R(o_1)} ; b_R(o_1)(q) = p \}\}} \cdot \Card{e(o_1)}
\end{eqnarray*}
If $p \in \conclusionscirc{R}$, then
\begin{eqnarray*}
\arity{\taylor{e}{i}}(p) & = & \arity{R, i}(p) \\
& & + \sum_{o_1 \in \boxesatzerogeq{R}{i}} \Card{\{ q \in \conclusionsnotcirc{B_R(o_1)} ; b_R(o_1)(q) = p \}\}} \cdot \Card{e(o_1)}\\
& & + \sum_{o_1 \in \boxesatzerogeq{R}{i}} \sum_{e_1 \in e(o_1)} \arity{\taylor{e_1}{i}}(p)
\end{eqnarray*}
\end{fact}

\begin{proof}
For any $o_1 \in \boxesatzerogeq{R}{i}$, for any $e_1 \in e(o_1)$, we have
\begin{align*}
& \{ p \in \conclusionsnotcirc{B_{\taylor{e}{i}}((o_1, e_1):o')} ; b_{\taylor{e}{i}}((o_1, e_1):o')(p) = c \} \\
= & \bigcup_{\begin{array}{c} q \in \conclusionscirc{B_R(o_1)}\\ b_R(o_1)(p) = c \end{array}} \{ p \in \conclusionsnotcirc{B_{\taylor{e_1}{i}}(o')} ; b_{\taylor{e_1}{i}}(o')(p) = q \}
\end{align*}
hence
\begin{align*}
& \Card{\{ p \in \conclusionsnotcirc{B_{\taylor{e}{i}}((o_1, e_1):o')} ; b_{\taylor{e}{i}}((o_1, e_1):o')(p) = c \}} \\
= & \sum_{\begin{array}{c} q \in \conclusionscirc{B_R(o_1)}\\ b_R(o_1)(p) = c \end{array}} \Card{\{ p \in \conclusionsnotcirc{B_{\taylor{e_1}{i}}(o')} ; b_{\taylor{e_1}{i}}(o')(p) = q \}}
\end{align*}
\end{proof}

\begin{lem}\label{lemma: m_{k, 0}(c)}
Let $R$ be a $\circ$-PS. Let $k > \cosize{R}$. Let $e$ be a $k$-injective pseudo-experiment of $R$. Let $p \in \portsatzero{R}$. Let $i \in \Nat$. Then $\arity{\taylor{e}{i}}(p) \textit{ mod }k = \arity{R, i}(p)$.
\end{lem}

\begin{proof}
First, notice that 
$$\arity{R, i}(p) \leq \arity{R}(p) \leq \cosize{R} < k \: \: \: (\ast)$$
Now, we prove the statement by induction on $\depthof{R}$. If $\depthof{R} = 0$, then we just apply $(\ast)$. If $\depthof{R} > 0$, then, by induction hypothesis, we have
\begin{eqnarray*}
& & (\sum_{o_1 \in \boxesatzerogeq{R}{i}} \sum_{e_1 \in e(o_1)} \sum_{\begin{array}{c} q \in \conclusionscirc{B_R(o_1)}\\ b_R(o_1)(q) = p \end{array}} \arity{\taylor{e_1}{i}}(q)) \textit{ mod } k\\
& = & (\sum_{o_1 \in \boxesatzerogeq{R}{i}} \sum_{e_1 \in e(o_1)} \sum_{\begin{array}{c} q \in \conclusionscirc{B_R(o_1)}\\ b_R(o_1)(q) = p \end{array}} \arity{B_R(o_1)}(q)) \textit{ mod } k\\
& = & (\sum_{o_1 \in \boxesatzerogeq{R}{i}} (\sum_{\begin{array}{c} q \in \conclusionscirc{B_R(o_1)}\\ b_R(o_1)(q) = p \end{array}} \arity{B_R(o_1)}(q)) \cdot \Card{e(o_1)}) \textit{ mod } k\\
& = & 0
\end{eqnarray*}
hence
\begin{eqnarray*}
\arity{\taylor{e}{i}}(p) \textit{ mod } k & = & (\arity{R, i}(p)\\
& & + \sum_{o_1 \in \boxesatzerogeq{R}{i}} \Card{\{ q \in \conclusionsnotcirc{B_R(o_1)} ; b_R(o_1)(q) = p \}\}} \cdot \Card{e(o_1)}\\
& & + \sum_{o_1 \in \boxesatzerogeq{R}{i}} \sum_{e_1 \in e(o_1)} \sum_{\begin{array}{c} q \in \conclusionscirc{B_R(o_1)}\\ b_R(o_1)(q) = p \end{array}} \arity{\taylor{e_1}{i}}(q)) \textit{ mod } k\\
& & \textit{(by Fact~\ref{fact: arity})}\\
& = & \arity{R, i}(p)
\end{eqnarray*}
\end{proof}

\textbf{Proof of Lemma~\ref{lem: arity{taylor{e}{i}}(c) for c at depth 0}:}

\begin{proof}
By Lemma~\ref{lemma: m_{k, 0}(c)}, we already know that $\arity{\taylor{e}{i}}(q) \textit{ mod } k = \arity{R, i}(q)$. We prove, by induction on $\depthof{R}$, that, for any $j > 0$, we have $m_{k, j}(\taylor{e}{i})(q) > 0$ if, and only if, there exists $o \in \mathcal{B}_R^{\geq i}(q)$ such that $k^j \in e^\#(o)$. If $\depthof{R} > 0$, then we have
\begin{eqnarray*}
& & \arity{\taylor{e}{i}}(q)\\
& = & \arity{\groundof{R}}(q)\\
& & + \sum_{o_1 \in \boxesatzerosmaller{R}{i}} \Card{ \{ p \in \conclusionsnotcirc{B_R(o_1)} ; b_R(o_1)(p) = q \}}\\
& & + \sum_{\begin{array}{c} o_1 \in \boxesatzerosmaller{R}{i}\\ q \in \conclusionscirc{B_R(o_1)}  \end{array}} \arity{B_R(o_1)}(q)\\
& & + \sum_{o_1 \in \boxesatzerogeq{R}{i}} (\Card{e(o_1)} \cdot \Card{\{ p \in \conclusionsnotcirc{B_R(o_1)} ; b_R(o_1)(p) = q \}})\\
& & + \sum_{\begin{array}{c} o_1 \in \exactboxesatzero{R}{i}\\ q \in \conclusionscirc{B_R(o_1)} \end{array}} \sum_{e_1 \in e(o_1)} \arity{\taylor{e_1}{i}}(q) + \sum_{\begin{array}{c} o_1 \in \boxesatzerogeq{R}{i+1}\\ q \in \conclusionscirc{B_R(o_1)} \end{array}
} \sum_{e_1 \in e(o_1)}  \arity{\taylor{e_1}{i}}(q)\allowdisplaybreaks\\
& = & \arity{\groundof{R}}(q)\\
& & + \sum_{o_1 \in \boxesatzerosmaller{R}{i}} (\Card{ \{ p \in \conclusionsnotcirc{B_R(o_1)} ; b_R(o_1)(p) = q \}} \\
& & + \sum_{\begin{array}{c} o_1 \in \boxesatzerosmaller{R}{i}\\ q \in \conclusionscirc{B_R(o_1)}  \end{array}} \arity{B_R(o_1)}(q)\\
& & + \sum_{o_1 \in \boxesatzerogeq{R}{i}} (\Card{e(o_1)} \cdot \Card{\{ p \in \conclusionsnotcirc{B_R(o_1)} ; b_R(o_1)(p) = q \}})\\
& & + \sum_{\begin{array}{c} o_1 \in \exactboxesatzero{R}{i}\\ q \in \conclusionscirc{B_R(o_1)} \end{array}} (\Card{e(o_1)} \cdot  \arity{B_R(o_1)}(q))\\
& &  + \sum_{\begin{array}{c} o_1 \in \boxesatzerogeq{R}{i+1}\\ q \in \conclusionscirc{B_R(o_1)} \end{array}
}  (a_{B_R(o_1), i}(q)  \cdot \Card{e(o_1)} + \sum_{e_1 \in e(o_1)} \sum_{o \in \mathcal{B}_{B_R(o_1)}^{\geq i}(q)} \sum_{\begin{array}{c} j \in \Nat\\ k^j \in e_1^\#(o) \end{array}} m_{k, j}(\taylor{e_1}{i})(q) \cdot k^j)    \\
& & \textit{(by Lemma~\ref{lemma: m_{k, 0}(c)} and by induction hypothesis)} \allowdisplaybreaks\\
\end{eqnarray*}
\end{proof}

\subsection{Proof of Proposition~\ref{prop: critical ports below new boxes}}

\begin{proof}
(Sketch) The proof is in four steps, using Lemmas~\ref{lem: M_i} and \ref{lem: arity{taylor{e}{i}}(c) for c at depth 0}. 
\begin{enumerate}
\item\label{item: step1} We first prove that, for any $j_0 \in \Nat \setminus \{ 0 \}$, for any $o_1 \in \boxesatzerogeq{R}{i}$, for any $e_1 \in e(o_1)$, we have $\criticalports{\taylor{e_1}{i}}{k}{j_0} \cap \conclusionscirc{B_R(o_1)} \subseteq$ $\criticalports{\taylor{e}{i}}{k}{j_0}$.
\item\label{item: step2} Second we prove, by induction on $\depthof{R}$, that
\begin{itemize}
\item for any $j \in \mathcal{N}_i(e)$, there exists $o \in \exactboxes{R}{i}$ such that $k^j \in e^\#(o)$ and
\begin{itemize}
\item either $o \in \exactboxesatzero{R}{i}$; moreover, in this case, $e^\#(o) = \{ k^j \}$ and $\criticalports{\taylor{e}{i}}{k}{j} = \im{b_R(o)}$;
\item or there exist $o_1 \in \boxesatzerogeq{R}{i+1}$ and $o' \in \exactboxes{B_R(o_1)}{i}$ such that $o = $ \mbox{$o_1:o'$;}  moreover, in this case, there exists $e_1 \in e(o_1)$ such that $j \in \mathcal{N}_i(e_1)$ and $\criticalports{\taylor{e}{i}}{k}{j} = $ 
$\{ (o_1, e_1) : p ;$ $p \in \criticalports{\taylor{e_1}{i}}{k}{j} \setminus \conclusionscirc{B_R(o_1)} \}$ $\cup$ $(\criticalports{\taylor{e_1}{i}}{k}{j} \cap \conclusionscirc{B_R(o_1)})$ $\cup$ $\{ b_R(o_1)(q) ;$ $q \in \mathcal{P} \cap \conclusions{B_R(o_1)} \textit{ and } \arity{\taylor{e_1}{i}}(q) < k \}$;
\end{itemize}
\item and, for any $j > 0$ such that $\criticalports{\taylor{e}{i}}{k}{j} \not= \emptyset$, we have $j \in \mathcal{M}_i(e)$.
\end{itemize}
\item\label{item: step3} We prove, by induction on $\depthof{R}$, that, for any $o \in \exactboxesatzero{\taylor{e}{i+1}}{i}$, there exists $j \in \mathcal{N}_{i}(e)$ such that $\im{b_{\taylor{e}{i+1}}(o)} = \criticalports{\taylor{e}{i}}{k}{j}$. 
\item\label{item: step4} Finally we prove, by induction on $\depthof{R}$, that, for any $j \in \mathcal{N}_i(e)$, we have $!_{e, i}(j) \in \exactboxesatzero{\taylor{e}{i+1}}{i}$. 
\end{enumerate}
\end{proof}

\section{Rebuilding the boxes: Proposition~\ref{prop: rebuilding boxes}}\label{appendix: rebuilding boxes}

\begin{prop}\label{prop: rebuilding boxes}
Let $R$ be a $\circ$-PS. Let $k > \Card{\boxes{R}{}}, \cosize{R}$. 
Let $e$ be a $k$-injective pseudo-experiment of $R$. Let $i \in \Nat$. Let $j_0 \in \mathcal{N}_i(e)$. 
We set $\mathcal{T} = \nontrivialconnected{\taylor{e}{i}}{\criticalports{\taylor{e}{i}}{k}{j_0}}{k}$. 
For any $T \in \mathcal{T}$, let $(m_j^T)_{j \in \Nat} \in \{ 0, \ldots, k-1 \}^\Nat$ such that $\Card{\{ T' \in \mathcal{T} ; T' \equiv T \}}$ $=$ $\sum_{j \in \Nat} m_j^T \cdot k^j$. 
Let $\mathcal{U} \subseteq \mathcal{T}$ such that, 
for any $T \in \mathcal{U}$,  $\Card{\{ T' \in \mathcal{U} ; T \equiv T' \}}$ $=$ $m_{j_0}^T$. 
Then there exists $\rho: $ $B_{\taylor{e}{i+1}}(\cod_{e, i}(j_0))$ $\simeq$ $\sum \overline{\mathcal{U}}$ such that 
$b_{\taylor{e}{i+1}}(\cod_{e, i}(j_0))(q) =$\\ $\left\lbrace \begin{array}{ll} 
\groundof{\rho}(q) & \textit{if $q \in \conclusionscirc{B_{\taylor{e}{i+1}}(\cod_{e, i}(j_0))}$;}\\
\target{\groundof{\taylor{e}{i}}}(\groundof{\rho}(q)) & \textit{if $q \in \conclusionsnotcirc{B_{\taylor{e}{i+1}}(\cod_{e, i}(j_0))}$.}
\end{array} \right.$
\end{prop}

The proof of Proposition~\ref{prop: rebuilding boxes} uses Proposition~\ref{prop: crucial}, which justifies that in the algorithm leading from $\taylor{e}{i}$ to $\taylor{e}{i+1}$, for every $j_0 \in \mathcal{N}_i(e)$, for every equivalence class $\mathfrak{T} \in \nontrivialconnected{\taylor{e}{i}}{\criticalports{\taylor{e}{i}}{k}{j_0}}{k}_{/_\equiv}$, if $\Card{\mathfrak{T}} = \sum_{j \in \Nat} m_j \cdot k^j$ (with $0 \leq m_j <k$), then we remove $m_{j_0} \cdot k^{j_0}$ elements of $\mathfrak{T}$ from $\groundof{\taylor{e}{i}}$. Its proof is quite technical, we just state the right induction hypothesis and give some remarks. 

\begin{definition}
Let $R$ be a $\circ$-PS. Let $o_1 \in \boxesatzerogeq{R}{i}$. Let $k > 1$. Let $e$ be a $k$-injective pseudo-experiment of $R$ and let $e_1 \in e(o_1)$. Let $i \in \Nat$. Let $\mathcal{P} \subseteq \portsatzero{\taylor{e_1}{i}}$. We denote by $\mathcal{P}[R, o_1, e_1, i]_k$ the following subset of $\portsatzero{\taylor{e}{i}}$: 
\begin{eqnarray*}
& & \mathcal{P}[R, o_1, e_1, i]_k \\
& = & \{ (o_1, e_1) : p ; p \in \mathcal{P} \setminus \conclusionscirc{B_R(o_1)} \} \cup (\mathcal{P} \cap \conclusionscirc{B_R(o_1)})\\
& & \cup \left\{ b_R(o_1)(q) ; q \in \mathcal{P} \cap \conclusions{B_R(o_1)} \textit{ and } \arity{\taylor{e_1}{i}}(q) < k \right\}
\end{eqnarray*}
\end{definition}


\begin{fact}\label{fact: adequate_extended}
Let $S$ be a differential $\circ$-PS. Let $\mathcal{P} \subseteq \portsatzero{S}$ and let $\mathcal{Q} \subseteq \exponentialportsatzero{S}$ such that
\begin{enumerate} 
\item $(\forall w \in \wiresatzero{S} \setminus \mathcal{P}) (\target{\groundof{S}}(w) \in \mathcal{P} \Rightarrow \target{\groundof{S}}(w) \in \exponentialportsatzero{S})$
\item $(\forall w \in \wiresatzero{S} \cap \mathcal{Q}) (\target{\groundof{S}}(w) \in \mathcal{P} \Rightarrow \target{\groundof{S}}(w) \in \mathcal{Q})$
\item and $(\forall o_1 \in \boxesatzero{S} \cap \mathcal{P}) \im{b_S(o_1)} \subseteq \mathcal{P}$.
\end{enumerate}
Then there exists a substructure $R$ of $S$ such that $\portsatzero{R} = \mathcal{P}$ and $R \sqsubseteq_\mathcal{Q} S$.
\end{fact}

\begin{proof}
We set $\mathcal{G} = (\mathcal{W}, \mathcal{P}, l, t, \mathcal{L}, \textsf{T})$ and $R = (\mathcal{G}, \mathcal{B}_0, B, b)$ with:
\begin{itemize}
\item $\mathcal{W} = \{ w \in \wiresatzero{S} \cap (\mathcal{P} \setminus \mathcal{Q}) ; \target{\groundof{S}}(w) \in \mathcal{P} \}$
\item $l = \restriction{\labelofcell{\groundof{S}}}{\mathcal{P}}$
\item $t = \restriction{\target{\groundof{S}}}{\mathcal{W}}$
\item $\mathcal{L} = \leftwires{\groundof{S}} \cap \{ w \in \wiresatzero{S}; \target{\wiresatzero{S}}(w) \in \mathcal{P} \}$
\item $\textsf{T} = \restriction{\textsf{T}_{\groundof{S}}}{\mathcal{P}}$
\item $\mathcal{B}_0 = \boxesatzero{S} \cap \mathcal{P}$
\item $B = \restriction{B_S}{\mathcal{B}_0}$
\item and $b = \restriction{b_S}{\mathcal{B}}$.
\end{itemize}
\end{proof}

If the pair $(\mathcal{P}, \mathcal{Q})$ satisfies the conditions of the previous fact, then we say that the pair $(\mathcal{P}, \mathcal{Q})$ is \emph{adequate} with respect to $S$ and we denote by $\restriction{S}{(\mathcal{P}, \mathcal{Q})}$ the unique substructure $R$ of $S$ such that $\portsatzero{R} = \mathcal{P}$ and $R \sqsubseteq_\mathcal{Q} S$. 
We have $\restriction{S}{\mathcal{P}} = \restriction{S}{(\mathcal{P}, \emptyset)}$.

\begin{fact}\label{fact: from taylor{e1}{i} to taylor{e}{i}}
Let $R$ be a $\circ$-PS. Let $k > 1$. Let $e$ be a $k$-injective pseudo-experiment of $R$. Let $i \in \Nat$. 
Let $o_1 \in \boxesatzerogeq{R}{i}$. Let $e_1 \in e(o_1)$. 
Let $\mathcal{P} \subseteq \portsatzero{\taylor{e_1}{i}}$ and $\mathcal{Q} \subseteq \exponentialportsatzero{\taylor{e_1}{i}}$ such that the pair $(\mathcal{P}, \mathcal{Q})$ is adequate with respect to $\taylor{e_1}{i}$. Then the pair $(\mathcal{P}[R, o_1, e_1, i]_k, \mathcal{Q}[R, o_1, e_1, i]_k)$ is adequate with respect to $\taylor{e}{i}$. Moreover we have 
\begin{itemize}
\item $\conclusions{\restriction{\taylor{e}{i}}{\mathcal{P}[R, o_1, e_1, i]_k}} \subseteq \im{b_R(o_1)} \cup \{ (o_1, e_1) : p ; p \in \conclusions{\restriction{\taylor{e_1}{i}}{\mathcal{P}}} \setminus \conclusions{B_R(o_1)} \} \cup  \{ (o_1, e_1) : p ; p \in \conclusions{\restriction{\taylor{e_1}{i}}{\mathcal{P}}} \cap \conclusionsnotcirc{B_R(o_1)} \textit{ and } \arity{\taylor{e_1}{i}}(p) \geq k \}$;
\item $\boxesatzero{\restriction{\taylor{e}{i}}{\mathcal{P}[R, o_1, e_1, i]_k}} = \{ (o_1, e_1) : o ; o \in \boxesatzero{\restriction{\taylor{e_1}{i}}{\mathcal{P}}} \}$;
\item and, for any $p \in \mathcal{P} \setminus \conclusionscirc{B_R(o_1)}$, $\arity{T}((o_1, e_1):p) = \arity{\restriction{\taylor{e_1}{i}}{\mathcal{P}}}(p)$.
\end{itemize}
\end{fact}

This fact allows the following definition:

\begin{definition}\label{defin: T[R, e, i]}
Let $R$ be a $\circ$-PS. Let $e$ be a pseudo-experiment of $R$. Let $i, k \in \Nat$. 
Let $o_1 \in \boxesatzerogeq{R}{i}$. Let $e_1 \in e(o_1)$. 
Let $T \sqsubseteq \taylor{e_1}{i}$. We set $T[R, e, i, o_1, e_1]_k = \restriction{\taylor{e}{i}}{(\portsatzero{T}[R, o_1, e_1, i]_k, \conclusions{T}[R, o_1, e_1, i]_k)}$.
\end{definition}

Before of the proof of Proposition~\ref{prop: crucial}, notice that the proof of Proposition~\ref{prop: critical ports below new boxes} actually proves something more than its statement: it proves also that, for any $j \in \mathcal{N}_i(e)$ such that $!_{e, i}(j) \notin \exactboxesatzero{R}{i}$, there exist $o_1 \in \boxesatzerogeq{R}{i+1}$ and $e_1 \in e(o_1)$ such that $j \in \mathcal{N}_i(e_1)$ and $\criticalports{\taylor{e}{i}}{k}{j} = \criticalports{\taylor{e_1}{i}}{k}{j}[R, o_1, e_1, i]_k$. From now on, whenever we refer to Proposition~\ref{prop: critical ports below new boxes}, we refer to the statement thus completed.

\begin{prop}\label{prop: crucial}
Let $R$ be a $\circ$-PS. Let $k > \Card{\boxes{R}{}}, \cosize{R}$. 
Let $e$ be a $k$-injective pseudo-experiment of $R$. Let $i \in \Nat$. 
Let $T \in \nontrivialconnected{\taylor{e}{i}}{(\criticalports{\taylor{e}{i}}{k}{\mathcal{N}_i(e)}, \conclusions{R})}{k}$. We set 
$$\mathcal{T} = \{ T' \in \nontrivialconnected{\taylor{e}{i}}{(\criticalports{\taylor{e}{i}}{k}{\mathcal{N}_i(e)},  \conclusions{R })}{k} ; T \equiv T' \}$$
and
$$\mathcal{T'} = \{ T' \in \nontrivialconnected{\taylor{e}{i+1}}{(\criticalports{\taylor{e}{i}}{k}{\mathcal{N}_i(e)},  \conclusions{R })}{k} ; T \equiv T' \}$$
Let $(m_j)_{j \in \Nat}, (m'_j)_{j \in \Nat} \in \{ 0, \ldots, k-1 \}^\Nat$ such that $\Card{\mathcal{T}}  = \sum_{j \in \Nat} m_j \cdot k^j$ 
and $ \Card{\mathcal{T'}} = \sum_{j \in \Nat} m'_j \cdot k^j$. Then 
\begin{itemize}
\item $\{ j \in \Nat \setminus \{ 0 \} ; m_j \not= 0 \} \subseteq \mathcal{M}_i(e)$
\item $\{ j \in \Nat \setminus \{ 0 \} ; m'_j \not= 0 \} \subseteq \mathcal{M}_{i+1}(e)$
\item and $(\forall j \in \mathcal{M}_{i+1}(e)) m'_j = m_j$.
\end{itemize}
Moreover, if $\{ T' \in \mathcal{T} ; T' \sqsubseteq R \} \not= \emptyset$, then $m'_0 = \Card{ \{ T' \in \mathcal{T} ; T' \sqsubseteq R \} } = m_0$.
\end{prop}

\begin{proof}
By induction on $\depthof{R}$. 
We distinguish between two cases:
\begin{enumerate}
\item Case $\conclusions{T} \setminus \portsatzero{R} \not= \emptyset$: by Proposition~\ref{prop: critical ports below new boxes}, there exist $o_1 \in \boxesatzerogeq{R}{i+1}$, $e_1 \in e(o_1)$ and $p' \in \criticalports{\taylor{e_1}{i}}{k}{\mathcal{N}_i(e_1)} \setminus \conclusionscirc{B_R(o_1)}$ such that $(o_1, e_1):p' \in \conclusions{T}$. 
For any $T' \in \mathcal{T}$, we set $\mathcal{P}_{T'} = \{q \in \portsatzero{\taylor{e_1}{i}}; $ \mbox{$(o_1, e_1):q$} $\in \portsatzero{T'} \}$ and $$U_{T'} = \restriction{\taylor{e_1}{i}}{(\mathcal{P}_{T'} \cup \bigcup_{o \in \mathcal{P}_{T'} \cap \boxesatzero{B_R(o_1)}} \im{b_{B_R(o_1)}(o)}, \criticalports{\taylor{e_1}{i}}{k}{\mathcal{N}_i(e_1)} \cup \conclusions{B_R(o_1)})}$$ we have $T' = U_{T'}[R, e, i, o_1, e_1]_k$, hence $\cosize{U_{T'}} \leq \cosize{T} < k$; therefore $U_{T'} \in \nontrivialconnected{\taylor{e_1}{i}}{(\criticalports{\taylor{e_1}{i}}{k}{\mathcal{N}_i(e_1)}, \conclusions{B_R(o_1)})}{k}$. 
\begin{itemize}
\item Now, notice that if there exists $q \in \conclusions{U_T} \cap \conclusionsnotcirc{B_R(o_1)}$ such that $q \notin \criticalports{\taylor{e_1}{i}}{k}{\mathcal{N}_i(e_1)}$, then: 
$$(\forall T' \in \mathcal{T}) U_{T'} \in \nontrivialconnected{B_R(o_1)}{(\criticalports{\taylor{e_1}{i}}{k}{\mathcal{N}_i(e_1)}, \conclusions{B_R(o_1)})}{k}$$ we thus have $(\forall T' \in \mathcal{T}) T' \sqsubseteq R$ and $\Card{\mathcal{T}} \leq \arity{R}{(b_R(o_1)(q))} < k$; we obtain $m_0 = \Card{\{ T' \sqsubseteq R ; T' \equiv T \}}$ and $\{ j \in \Nat \setminus \{ 0 \} ; m_j \not= 0 \} = \emptyset$. In the same way, we have
$$(\forall T' \in \mathcal{T'}) U_{T'} \in \nontrivialconnected{B_R(o_1)}{(\criticalports{\taylor{e_1}{i}}{k}{\mathcal{N}_i(e_1)}, \conclusions{B_R(o_1)})}{k}$$ we thus have $(\forall T' \in \mathcal{T'}) T' \sqsubseteq R$ and $\Card{\mathcal{T'}} \leq \arity{R}{(b_R(o_1)(q))} < k$; we obtain $m'_0 = \Card{\{ T' \sqsubseteq R ; T' \equiv T \}}$ and $\{ j \in \Nat \setminus \{ 0 \} ; m'_j \not= 0 \} = \emptyset$.
\item Otherwise, for any $T' \in \mathcal{T}$, we have $\conclusions{U_{T'}} = \conclusions{U_T}$, hence $U_{T'} \equiv U_T$. So, 
$$\Card{\mathcal{T}} = \Card{\{ U \in \nontrivialconnected{\taylor{e_1}{i}}{(\criticalports{\taylor{e_1}{i}}{k}{\mathcal{N}_i(e_1)},  \conclusions{B_R(o_1)})}{k} ; U \equiv U_T \}}$$
and
$$\Card{\mathcal{T'}} = \Card{\{ U \in \nontrivialconnected{\taylor{e_1}{i+1}}{(\criticalports{\taylor{e_1}{i}}{k}{\mathcal{N}_i(e_1)},  \conclusions{B_R(o_1) })}{k} ; U \equiv U_T \}}$$
We apply the induction hypothesis and we obtain:
\begin{itemize}
\item $\{ j \in \Nat \setminus \{ 0 \} ; m_j \not= 0 \} \subseteq \mathcal{M}_i(e_1) \subseteq \mathcal{M}_i(e)$
\item $\{ j \in \Nat \setminus \{ 0 \} ; m'_j \not= 0 \} \subseteq \mathcal{M}_{i+1}(e_1) \subseteq \mathcal{M}_{i+1}(e)$
\item and $(\forall j \in \mathcal{M}_{i+1}(e_1)) m'_j = m_j$.
\end{itemize}
\end{itemize}
\item Case $\conclusions{T} \setminus \portsatzero{R} = \emptyset$: Let $o_1 \in \boxesatzerogeq{R}{i}$. We set $\mathcal{T}_{o_1} = \{ U \in \connectedcomponents{B_R(o_1)}{k} ; U[R, e, i, o_1, e_1]_k \equiv T \textit{ and } \conclusionsnotcirc{U} \not= \emptyset \}$. We set $m^{o_1} = \Card{\mathcal{T}_{o_1}}$. We have $m^{o_1} < k$. Let $e_1 \in e(o_1)$. 

Notice that, for any $U \in \nontrivialconnected{\taylor{e_1}{i}}{(\criticalports{\taylor{e_1}{i}}{k}{\mathcal{N}_i(e_1)},  \conclusions{B_R(o_1) })}{k}$ such that $U[R, e, i, o_1, e_1]_k \equiv T$, we have $U \in \connectedcomponents{\taylor{e_1}{i}}{k}$.

Now, for any  $U \in \connectedcomponents{\taylor{e_1}{i}}{k}$ such that $\conclusionsnotcirc{U} \not= \emptyset$, we have $U \sqsubseteq B_R(o_1)$. Indeed assume that $U \in \connectedcomponents{\taylor{e_1}{i}}{k}$ such that $\conclusionsnotcirc{U} \not= \emptyset$ and $\portsatzero{U} \setminus \portsatzero{B_R(o_1)} \not= \emptyset$; let $q \in \portsatzero{U} \setminus \portsatzero{B_R(o_1)}$ and let $q' \in \conclusionsnotcirc{U}$; there exist $p_0, \ldots p_{n+1} \in \portsatzero{U}$ such that $p_0 = q'$, $p_{n+1} = q$, $(\forall j \in \{ 0, \ldots, n \}) p_j \coh_{\taylor{e_1}{i}} p_{j+1}$ and $(\forall j \in \{ 0, \ldots, n \}) (p_j \in \conclusionscirc{U} \Rightarrow j \in \{ 0, n+1 \})$. Let $j_0 = \max \{ j \in \{ 0, \ldots, n \} ; p_j \in \portsatzero{B_R(o_1)} \}$. We have $\arity{B_R(o_1)}(p_{j_0}) = \arity{\taylor{e_1}{i}}(p_{j_0}) \geq k > \cosize{B_R(o_1)}$, which is contradictory.


We thus obtain $\Card{ \left\lbrace U \in \connectedcomponents{\taylor{e_1}{i}}{k} ; \conclusionsnotcirc{U} \not= \emptyset \textit{ and } U[R, e, i, o_1, e_1]_k \equiv T \right\rbrace } = m^{o_1}$. 
We distinguish between two cases:
\begin{itemize}
\item $\left\lbrace U \in \connectedcomponents{\taylor{e_1}{i}}{k} ; \conclusionsnotcirc{U} = \emptyset \textit{ and } U[R, e, i, o_1, e_1]_k \equiv T \right\rbrace = \emptyset$: for any $j \in \mathcal{M}_i(e_1)$, we set $m_j^{(o_1, e_1)} = 0$.
\item $\left\lbrace U \in \connectedcomponents{\taylor{e_1}{i}}{k} ; \conclusionsnotcirc{U} = \emptyset \textit{ and } U[R, e, i, o_1, e_1]_k \equiv T \right\rbrace \not= \emptyset$: let $U \in \connectedcomponents{\taylor{e_1}{i}}{k}$ such that $\conclusionsnotcirc{U} = \emptyset$ and $U[R, e, i, o_1, e_1]_k \equiv T$; let $(m_j^{(o_1, e_1)})_{j \in \mathbb{N}}, ({m'_j}^{(o_1, e_1)})_{j \in \mathbb{N}} \in \{ 0, \ldots, k-1 \}^{\mathbb{N}}$ such that
$$\Card{\{ U' \in \nontrivialconnected{\taylor{e_1}{i}}{(\criticalports{\taylor{e_1}{i}}{k}{\mathcal{N}_i(e_1)},  \conclusions{B_R(o_1) })}{k} ; U' \equiv U \}} = \sum_{j \in \mathbb{N}} {m}_j^{(o_1, e_1)} \cdot k^j$$
and
$$\Card{\{ U' \in \nontrivialconnected{\taylor{e_1}{i+1}}{(\criticalports{\taylor{e_1}{i}}{k}{\mathcal{N}_i(e_1)},  \conclusions{B_R(o_1) })}{k} ; U' \equiv U \}} = \sum_{j \in \mathbb{N}} {m'}_j^{(o_1, e_1)} \cdot k^j$$
By induction hypothesis, we have
\begin{itemize}
\item $\{ j \in \Nat \setminus \{ 0 \} ; m_j^{(o_1, e_1)} \not= 0 \} \subseteq \mathcal{M}_i(e_1)$
\item $\{ j \in \Nat \setminus \{ 0 \} ; {m'}_j^{(o_1, e_1)} \not= 0 \} \subseteq \mathcal{M}_{i+1}(e_1)$
\item and $(\forall j \in \mathcal{M}_{i+1}(e_1)) {m'}_j^{(o_1, e_1)} = m_j^{(o_1, e_1)}$.
\end{itemize}
For any  $U \in \connectedcomponents{\taylor{e_1}{i}}{k}$ such that $\conclusionsnotcirc{U} = \emptyset$, we have $\neg U \sqsubseteq B_R(o_1)$, hence ${m'}_0^{(o_1, e_1)} = 0 = {m}_0^{(o_1, e_1)}$ and $m^{o_1} = \Card{\{ T' \in \mathcal{T} ; T' \sqsubseteq R \}}$. 
\end{itemize}
Finally:
\begin{eqnarray*}
& & \Card{\{ T' \in \mathcal{T} ; \neg T' \sqsubseteq R \}} \\
& = & \Card{\bigcup_{o_1 \in \boxesatzerogeq{R}{i}} \bigcup_{e_1 \in e(o_1)} \left\lbrace U' \in \nontrivialconnected{\taylor{e_1}{i}}{(\criticalports{\taylor{e_1}{i}}{k}{\mathcal{N}_i(e_1)}, \conclusions{B_R(o_1)})}{k} ;  U'[R, e, i, o_1, e_1]_k \in \mathcal{T} \right\rbrace}\\
& = & \Card{\bigcup_{o_1 \in \boxesatzerogeq{R}{i}} \bigcup_{e_1 \in e(o_1)} \left\lbrace U' \in 
\connectedcomponents{\taylor{e_1}{i}}{k} ;  U'[R, e, i, o_1, e_1]_k \in \mathcal{T} \right\rbrace}\\
& = & \sum_{o_1 \in \boxesatzerogeq{R}{i}} \sum_{e_1 \in e(o_1)} \Card{\left\lbrace U' \in 
\connectedcomponents{\taylor{e_1}{i}}{k} ;  U'[R, e, i, o_1, e_1]_k \in \mathcal{T} \right\rbrace}\\
& = & \sum_{o_1 \in \boxesatzerogeq{R}{i}} \sum_{e_1 \in e(o_1)} \left( m^{o_1} + \Card{\left\lbrace U' \in 
\connectedcomponents{\taylor{e_1}{i}}{k} ;  U'[R, e, i, o_1, e_1]_k \in \mathcal{T} \textit{ and } \conclusionsnotcirc{U'} = \emptyset \right\rbrace} \right)\\
& = & \sum_{o_1 \in \boxesatzerogeq{R}{i}} \sum_{e_1 \in e(o_1)} (m^{o_1} + \sum_{j \in \mathcal{M}_i(e_1)} m_j^{(o_1, e_1)} \cdot k^j) \\
& = & \sum_{o_1 \in \boxesatzerogeq{R}{i}} (m^{o_1} \cdot \Card{e(o_1)} + \sum_{e_1 \in e(o_1)} \sum_{j \in \mathcal{M}_i(e_1)} m_j^{(o_1, e_1)} \cdot k^j) 
\end{eqnarray*}
and
\begin{eqnarray*}
& & \Card{\{ T' \in \mathcal{T'} ; \neg T' \sqsubseteq R \}}\\
& = & \Card{\bigcup_{o_1 \in \boxesatzerogeq{R}{i+1}} \bigcup_{e_1 \in e(o_1)} \left\lbrace U' \in \nontrivialconnected{\taylor{e_1}{i+1}}{(\criticalports{\taylor{e_1}{i}}{k}{\mathcal{N}_i(e_1)}, \conclusions{B_R(o_1)})}{k} ;  U'[R, e, i, o_1, e_1]_k \in \mathcal{T'} \right\rbrace}\\
& = & \Card{\bigcup_{o_1 \in \boxesatzerogeq{R}{i+1}} \bigcup_{e_1 \in e(o_1)} \left\lbrace U' \in 
\connectedcomponents{\taylor{e_1}{i+1}}{k} ;  U'[R, e, i, o_1, e_1]_k \in \mathcal{T'} \right\rbrace}\\
& = & \sum_{o_1 \in \boxesatzerogeq{R}{i+1}} \sum_{e_1 \in e(o_1)} \Card{\left\lbrace U' \in 
\connectedcomponents{\taylor{e_1}{i+1}}{k} ;  U'[R, e, i, o_1, e_1]_k \in \mathcal{T'} \right\rbrace}\\
& = & \sum_{o_1 \in \boxesatzerogeq{R}{i+1}} \sum_{e_1 \in e(o_1)} \left( m^{o_1} + \Card{\left\lbrace U' \in 
\connectedcomponents{\taylor{e_1}{i+1}}{k} ;  U'[R, e, i, o_1, e_1]_k \in \mathcal{T'} \textit{ and } \conclusionsnotcirc{U'} = \emptyset \right\rbrace} \right)\\
& = & \sum_{o_1 \in \boxesatzerogeq{R}{i+1}} \sum_{e_1 \in e(o_1)} (m^{o_1} + \sum_{j \in \mathcal{M}_{i+1}(e_1)} m_j^{(o_1, e_1)} \cdot k^j) \\
& = & \sum_{o_1 \in \boxesatzerogeq{R}{i+1}} (m^{o_1} \cdot \Card{e(o_1)} + \sum_{e_1 \in e(o_1)} \sum_{j \in \mathcal{M}_{i+1}(e_1)} m_j^{(o_1, e_1)} \cdot k^j) 
\end{eqnarray*}
Therefore $\{ j \in \Nat \setminus \{ 0 \} ; {m}_j \not= 0 \} \subseteq \bigcup_{o_1 \in \boxesatzerogeq{R}{i}} (\{ \log_k(\Card{e(o_1)} \} \cup \bigcup_{e_1 \in e(o_1)}  \mathcal{M}_i(e_1)) \subseteq \mathcal{M}_i(e)$ and $\{ j \in \Nat \setminus \{ 0 \} ; {m'}_j \not= 0 \} \subseteq \bigcup_{o_1 \in \boxesatzerogeq{R}{i+1}} (\{ \log_k(\Card{e(o_1)} \} \cup \bigcup_{e_1 \in e(o_1)}  \mathcal{M}_{i+1}(e_1)) \subseteq \mathcal{M}_{i+1}(e)$. Moreover $m_0 = m'_0 = \Card{\{ T' \in \mathcal{T} ; T' \sqsubseteq R \}}$. Lastly, for any $j \in \mathcal{M}_{i+1}(e)$, we have $m'_j = \sum_{e_1 \in e(o_1)} \sum_{j \in \mathcal{M}_i(e_1)} m_j^{(o_1, e_1)}$.
\end{enumerate}
\end{proof}

We need a variant of the notion of equivalence denoted by $\equiv$:

\begin{definition}
Let $R$ be a $\circ$-PS. Let $k > 1$. For any $k$-injective pseudo-experiment $e$ of $R$, for any $i \in \Nat$, for any $o \in \exactboxesatzero{\taylor{e}{i+1}}{i}$, for any $T' \in \connectedcomponents{B_{\taylor{e}{i+1}}(o)}{k}$, for any $\circ$-PS $T''$ such that $\conclusionsnotcirc{T''} \subseteq \wiresatzero{\taylor{e}{i}}$, we write $\varphi :T' \equiv_{e, i, o} T''$ if $\varphi: T' \simeq T''$ such that 
$(\forall q \in \conclusionscirc{T'}) \groundof{\varphi}(q) = b_{\taylor{e}{i+1}}(o)(q) \textit{ and }(\forall q \in \conclusionsnotcirc{T'}) \target{\groundof{\taylor{e}{i}}}(\groundof{\varphi}(q)) =  b_{\taylor{e}{i+1}}(o)(q)$ and $T' \equiv_{e, i, o} T''$ if there exists $\varphi$ such that $\varphi :T' \equiv_{e, i, o} T''$.
\end{definition}

\textbf{Sketch of the proof of Proposition~\ref{prop: rebuilding boxes}}

\begin{proof}(Sketch) 
We set $\mathcal{W} = \connectedcomponents{B_{\taylor{e}{i+1}}(\cod_{e, i}(j_0))}{k}$. 
Let $\tau : \mathcal{W}_{/\equiv} \to \mathcal{W}$ such that, for any $\mathcal{V} \in \mathcal{W}_{/\equiv}$, we have $\tau(\mathcal{V}) \in \mathcal{V}$. 
For any $\mathcal{V} \in \mathcal{W}_{/\equiv}$, there exist $\sigma(\mathcal{V}) \in \mathcal{T}$ and $\varphi_\mathcal{V} : \tau(\mathcal{V}) \simeq \overline{\sigma(\mathcal{V})}$ such that, for any $q \in \conclusions{B_{\taylor{e}{i+1}}(\cod_{e, i}(j_0))} \cap \portsatzero{\tau(\mathcal{V})}$, we have $b_{\taylor{e}{i+1}}(\cod_{e, i}(j_0))(q) =$  $\left\lbrace \begin{array}{ll} 
\groundof{\varphi_\mathcal{V}}(q) & \textit{if $q \in \conclusionscirc{B_{\taylor{e}{i+1}}(\cod_{e, i}(j_0))}$;}\\
\target{\groundof{\taylor{e}{i}}}(\groundof{\varphi_\mathcal{V}}(q)) & \textit{if $q \in \conclusionsnotcirc{B_{\taylor{e}{i+1}}(\cod_{e, i}(j_0))}$.} \end{array} \right.$ We check the existence of such $\sigma(\mathcal{V})$ and $\varphi(\mathcal{V})$ by induction on $\depthof{R}$:
\begin{itemize}
\item In the case $!_{e, i}(j_0) \in \exactboxesatzero{R}{i}$, for any $e_1 \in e(!_{e, i}(j_0))$, we have $\tau(\mathcal{V}) \in \connectedcomponents{\taylor{e_1}{i}}{k}$, hence one can set $\sigma(\mathcal{V}) = \tau(\mathcal{V})[R, e, i, !_{e, i}(j_0), e_1]_k$ for some $e_1 \in e(!_{e, i}(j_0))$.
\item In the case $!_{e, i}(j_0) = (o_1, e_1):!_{e_1, i}(j_0)$ for some $o_1 \in \boxesatzerogeq{R}{i+1}$ and some $e_1 \in e(o_1)$, by induction hypothesis, there exists $\sigma'(\mathcal{V}) \in \nontrivialconnected{\taylor{e_1}{i}}{\criticalports{\taylor{e_1}{i}}{k}{j_0}}{k}$ and $\varphi'_\mathcal{V} : \tau(\mathcal{V}) \simeq \overline{\sigma'(\mathcal{V})}$ such that, for any $q \in \conclusions{B_{\taylor{e}{i+1}}(\cod_{e, i}(j_0))} \cap \portsatzero{\tau(\mathcal{V})}$, we have $b_{\taylor{e_1}{i+1}}(\cod_{e_1, i}(j_0))(q) =$  $\left\lbrace \begin{array}{ll} 
\groundof{\varphi'_\mathcal{V}}(q) & \textit{if $q \in \conclusionscirc{B_{\taylor{e}{i+1}}(\cod_{e, i}(j_0))}$;}\\
\target{\groundof{\taylor{e_1}{i}}}(\groundof{\varphi'_\mathcal{V}}(q)) & \textit{if $q \in \conclusionsnotcirc{B_{\taylor{e}{i+1}}(\cod_{e, i}(j_0))}$.} \end{array} \right.$ We set $\sigma(\mathcal{V}) = \sigma'(\mathcal{V})[R, e, i, !_{e, i}(j_0), e_1]_k$. Moreover we set 
$(\varphi_\mathcal{V})_{\mathcal{G}} : \begin{array}{rcl} 
\portsatzero{\tau(\mathcal{V})} & \to & \portsatzero{\overline{\sigma(\mathcal{V})}} \\ 
p & \mapsto & \left\lbrace \begin{array}{ll} \groundof{\varphi'_{\mathcal{V}}}(p) & \textit{if $\groundof{\varphi'_{\mathcal{V}}}(p) \in \conclusionscirc{B_R(o_1)}$}\\ (o_1, e_1):\groundof{\varphi'_{\mathcal{V}}}(p) & \textit{otherwise;} \end{array} \right. \end{array}$ and $\varphi(\mathcal{V}) = ((\varphi_{\mathcal{V}})_{\mathcal{G}}, (\varphi'(\mathcal{V})(o))_{o \in \boxesatzero{\tau(\mathcal{V})}})$.
\end{itemize}

We set $\mathcal{U}_{\mathcal{V}} = \{ T' \in \mathcal{U} ; T' \equiv \sigma(\mathcal{V}) \}$.

Let us show that there exists a bijection $\epsilon_{\mathcal{V}} : \mathcal{V} \to \mathcal{U}_{\mathcal{V}}$: 
there exist a partition $\mathfrak{U}$ of $\{ U' \in \mathcal{T} \setminus \nontrivialconnected{\taylor{e}{i+1}}{\criticalports{\taylor{e}{i}}{k}{\mathcal{N}_i(e)}}{k} ; U' \equiv \sigma(\mathcal{V}) \}$ and a bijection $\delta: \{ \cod_{e, i}(j) ; j \in \mathcal{N}_i(e) \textit{ and } n_j \not= 0 \} \to \mathfrak{U}$ such that, for any $j \in \mathcal{N}_i(e)$ such that $n_j \not= 0$, we have $\Card{\delta(\cod_{e, i}(j))} = n_j \cdot k^j$, where $n_j$ is the integer $\Card{\{ T' \in \mathcal{W} ; T' \equiv_{e, i, o} \overline{\sigma(\mathcal{V})} \}}$. We have $\Card{\bigcup \mathfrak{U}} = \sum_{j \in \mathcal{N}_i(e)} n_j \cdot k^j$. By Proposition~\ref{prop: crucial}, we have $\Card{\bigcup \mathfrak{U}}= \sum_{j \in \mathcal{N}_i(e)} m_j^{\sigma(\mathcal{V})} \cdot k^j$. We thus have $(\forall j \in \mathcal{N}_i(e)) n_j = m_j^{\sigma(\mathcal{V})}$.

For any $V \in \mathcal{V}$, there exists $\rho_{V} : V \simeq \overline{\epsilon_\mathcal{V}(V)}$ such that, for any $q \in \conclusions{B_{\taylor{e}{i+1}}(\cod_{e, i}(j_0))} \cap \portsatzero{V}$, $b_{\taylor{e}{i+1}}(\cod_{e, i}(j_0))(q) = $ $\left\lbrace \begin{array}{ll} 
\groundof{\rho_{V}}(q) & \textit{if $q \in \conclusionscirc{B_{\taylor{e}{i+1}}(\cod_{e, i}(j_0))}$;}\\
\target{\groundof{\taylor{e}{i}}}(\groundof{\rho_{V}}(q)) & \textit{if $q \in \conclusionsnotcirc{B_{\taylor{e}{i+1}}(\cod_{e, i}(j_0))}$.} \end{array} \right.$
We have: for any $\mathcal{V} \in \mathcal{W}_{/ \equiv}$, there exists $\rho_{\mathcal{V}} : \sum_{V \in \mathcal{V}} V \simeq \sum_{V \in \mathcal{V}} \overline{\epsilon_\mathcal{V}(V)}$ such that, for any $V_1, V_2 \in \mathcal{V}$, $\restriction{\groundof{\rho_{V_1}}}{\conclusionscirc{V_1} \cap \conclusionscirc{V_2}} = \restriction{\groundof{\rho_{\mathcal{V}}}}{\conclusionscirc{V_1} \cap \conclusionscirc{V_2}} 
= \restriction{\groundof{\rho_{V_2}}}{\conclusionscirc{V_1} \cap \conclusionscirc{V_2}}$; 
the sets $\{ \sum_{V \in \mathcal{V}} V ; \mathcal{V} \in \mathcal{W}_{/ \equiv} \}$ and $\{ \sum_{V \in \mathcal{V}} \overline{\epsilon_{\mathcal{V}}(V)} ; \mathcal{V} \in \mathcal{W}_{/ \equiv} \}$ are gluable; 
there exists $\rho : \sum_{\mathcal{V} \in \mathcal{W}_{/ \equiv}} \sum_{V \in \mathcal{V}} V  \simeq \sum_{\mathcal{V} \in \mathcal{W}_{/ \equiv}} \sum_{V \in \mathcal{V}} \overline{\epsilon_\mathcal{V}(V)}$ such that, 
for any $q \in \conclusions{B_{\taylor{e}{i+1}}(\cod_{e, i}(j_0))}$, we have $b_{\taylor{e}{i+1}}(\cod_{e, i}(j_0))(q) = $ $\left\lbrace \begin{array}{ll} 
\groundof{\rho}(q) & \textit{if $q \in \conclusionscirc{B_{\taylor{e}{i+1}}(\cod_{e, i}(j_0))}$;}\\
\target{\groundof{\taylor{e}{i}}}(\groundof{\rho}(q)) & \textit{if $q \in \conclusionsnotcirc{B_{\taylor{e}{i+1}}(\cod_{e, i}(j_0))}$.} \end{array} \right.$
But we have 
\begin{eqnarray*}
\sum_{\mathcal{V} \in \mathcal{W}_{/ \equiv}} \sum_{V \in \mathcal{V}} V & = & \sum \mathcal{W} \\
& = & B_{\taylor{e}{i+1}}(\cod_{e, i}(j_0)) \textit{(by Fact~\ref{fact: connected components})}
\end{eqnarray*}
and $\sum_{\mathcal{V} \in \mathcal{W}_{/ \equiv}} \sum_{V \in \mathcal{V}} \overline{\epsilon_\mathcal{V}(V)} = \sum_{\mathcal{V} \in \mathcal{W}_{/ \equiv}} \sum \overline{\mathcal{U}_{\mathcal{V}}} = \sum \overline{\mathcal{U}}$.
\end{proof}

\section{Proposition~\ref{prop: from i to i+1}}

In the sketch of the proof of Proposition~\ref{prop: from i to i+1}, we gave the complete formalized algorithm leading from $\taylor{e}{i}$ to $\taylor{e}{i+1}$ (up to the names of the ports). We only gave some arguments in favour of its correctness. In Section~\ref{appendix: rebuilding boxes}, we already stated Proposition~\ref{prop: crucial}. This proposition is actually used in the proof of Proposition~\ref{prop: rebuilding boxes} and again in Proposition~\ref{prop: from i to i+1}. Subsections~\ref{subsection: deterministic}, ~\ref{subsection: cond_sup_1} and~\ref{subsection: cond_sup_2} give some more arguments in favour of the correctness of the algorithm. 

\subsection{The algorithm is deterministic}\label{subsection: deterministic}

We should show that, if $S''$ is some differential PS that enjoys the conditions given on $S'$ in the proof of Proposition~\ref{prop: from i to i+1}, then $S'' \equiv S'$. The proof is quite easy but tedious.

\subsection{$\groundof{\restriction{\taylor{e}{i}}}{\mathcal{P}} \sqsubseteq \groundof{\taylor{e}{i+1}}$}\label{subsection: cond_sup_1}

As said in its informal description before the statement of the proposition, in order to obtain the differential PS $\taylor{e}{i+1}$, we consider the set $\bigcup_{j \in \mathcal{N}_i(e)} \nontrivialconnected{\taylor{e}{i}}{\criticalports{\taylor{e}{i}}{k}{j}}{k}$ and we remove some of these elements and put some other of these elements inside new boxes. This means in particular that we keep all the ports that are not caught by some element of $\bigcup_{j \in \mathcal{N}_i(e)} \nontrivialconnected{\taylor{e}{i}}{\criticalports{\taylor{e}{i}}{k}{j}}{k}$, what is formally stated by:
$$\groundof{\restriction{\taylor{e}{i}}}{\mathcal{P}} \sqsubseteq \groundof{\taylor{e}{i+1}}$$
with 
$$\mathcal{P} =  \{ p \in \portsatzero{\taylor{e}{i+1}} ; p \notin \bigcup_{j \in \mathcal{N}_i(e)} \bigcup_{\begin{array}{c} T \in \nontrivialconnected{\taylor{e}{i}}{\criticalports{\taylor{e}{i}}{k}{j}}{k} \end{array}} \portsatzero{T} \}$$

It is essentially what Lemma~\ref{lem: P_k subseteq ports of taylor} says. 

\begin{lem}\label{lem: P_k subseteq ports of taylor}
Let $R$ be a $\circ$-PS. Let $k > \Card{\boxes{R}{}}, \cosize{R}$. 
Let $e$ be a $k$-injective pseudo-experiment of $R$. Let $i \in \Nat$. 
For any $p \in \portsatzero{\taylor{e}{i}} \setminus \portsatzero{\taylor{e}{i+1}}$, there exist $j \in \mathcal{N}_{i}(e)$ and $T \in \nontrivialconnected{\taylor{e}{i}}{\criticalports{\taylor{e}{i}}{k}{j}}{k}$ such that, for any $q \in \conclusionscirc{T}$, we have $\arity{T}(q) \leq \arity{R}(q)$.
\end{lem}

\begin{proof}
(Sketch) 
By induction on $\depthof{R}$:

If $\depthof{R} = 0$, then $\portsatzero{\taylor{e}{i}} \setminus \portsatzero{\taylor{e}{i+1}} = \emptyset$.

If $\depthof{R} > 0$, then we distinguish between two cases:
\begin{itemize}
\item In the case there exist $o_1 \in \exactboxesatzero{R}{i}$, $e_1 \in e(o_1)$ and $p' \in \portsatzero{\taylor{e_1}{i}} \setminus \conclusionscirc{B_R(o_1)}$ such that $p = (o_1, e_1):p'$, 
there exists $T \in \nontrivialconnected{\taylor{e}{i}}{\im{b_R(o_1)}}{k}$ such that $p \in \portsatzero{T}$ and, for any $q \in \conclusionscirc{T}$, we have $\arity{T}(q) \leq \arity{R}(q)$. 
Now, by Proposition~\ref{prop: critical ports below new boxes}, we have $\im{b_R(o_1)} = \criticalports{\taylor{e}{i}}{k}{j}$ with $e^\#(o_1) = \{ k^j \}$.
\item In the case there exist $o_1 \in \boxesatzerogeq{R}{i+1}$, $e_1 \in e(o_1)$ and $p' \in \portsatzero{\taylor{e_1}{i}} \setminus \portsatzero{\taylor{e_1}{i+1}}$ such that $p = (o_1, e_1):p'$, by induction hypothesis, there exists $j \in \mathcal{N}_i(e_1)$ and $T' \in \nontrivialconnected{\taylor{e_1}{i}}{\criticalports{\taylor{e_1}{i}}{k}{j}}{k}$ such that $p' \in \portsatzero{T'} \setminus \conclusions{T'}$ and, for any $q \in \conclusionscirc{T'}$, we have $\arity{T'}(q) \leq \arity{B_R(o_1)}(q)$. 

By Fact~\ref{fact: from taylor{e1}{i} to taylor{e}{i}}, we can set $T = T'[R, e, i, o_1, e_1]_k$. By 
Proposition~\ref{prop: critical ports below new boxes}, we have $T \trianglelefteq_{\criticalports{\taylor{e}{i}}{k}{j}} \taylor{e}{i}$.

Let $q \in \portsatzero{T} \setminus \{ (o_1, e_1) : q' ; q' \in \portsatzero{T'} \setminus \conclusionscirc{B_R(o_1)} \}$. Notice that $q \in \portsatzero{R}$. We have 
\begin{eqnarray*}
\arity{\groundof{T}}(q) & \leq & \Card{\{ w \in \wiresatzero{R} ; \target{\groundof{R}}(w) = q \} } \\
& & + \Card{\{ p \in \conclusionsnotcirc{B_R(o_1)} ; b_R(o_1)(p) = q \}}
\end{eqnarray*}
Futhermore, notice that
 we have $\boxesatzerosmaller{R}{i} \cap \boxesatzero{T} = \emptyset$; so we have:
\begin{itemize}
\item if $q \notin \conclusionscirc{B_R(o_1)}$, then 
\begin{eqnarray*}
& & \sum_{o \in \boxesatzero{T}} \Card{\{ p \in \conclusionsnotcirc{B_T(o)} ; b_T(o)(p) = q \}} \\
& & + \sum_{\begin{array}{c} o \in \boxesatzero{T}\\ q \in \im{\restriction{b_T(o)}{\conclusionscirc{B_T(o)}}}\end{array}} \arity{B_T(o)}(q_{T, o})\allowdisplaybreaks\\
& = & 0
\end{eqnarray*}
\item if $q \in \conclusionscirc{B_R(o_1)}$, then
\begin{eqnarray*}
& & \sum_{o \in \boxesatzero{T}} \Card{\{ p \in \conclusionsnotcirc{B_T(o)} ; b_T(o)(p) = q \}} \\
& & + \sum_{\begin{array}{c} o \in \boxesatzero{T}\\ q \in \im{\restriction{b_T(o)}{\conclusionscirc{B_T(o)}}}\end{array}} \arity{B_T(o)}(q_{T, o})\allowdisplaybreaks\\
& = & 
\sum_{o' \in \boxesatzero{T'}} \Card{\{ p \in \conclusionsnotcirc{B_{T'}(o')} ; b_{T'}(o')(p) = q \}}\\
& & + \sum_{\begin{array}{c} o' \in \boxesatzero{T'}\\
q \in \im{\restriction{b_{T'}(o')}{\conclusionscirc{B_{T'}(o')}}} \end{array}} \arity{B_{T'}(o')}(q_{T', o'})\allowdisplaybreaks\\
& \leq & 
\arity{T'}(q) \allowdisplaybreaks\\
& \leq & 
\arity{B_R(o_1)}(q)
\end{eqnarray*}
\end{itemize}
hence, in the two cases, $\arity{T}(q) \leq \arity{R}(q)$.

Lastly, let $q \in \conclusions{T}$. By Fact~\ref{fact: from taylor{e1}{i} to taylor{e}{i}}, we can distinguish between the two following cases:
\begin{itemize}
\item there exists $q' \in \conclusions{T'} \setminus \conclusionscirc{B_R(o_1)}$ such that $q = (o_1, e_1) : q'$: by Proposition~\ref{prop: critical ports below new boxes}, we have $q \in \criticalports{\taylor{e}{i}}{k}{j}$;
\item $q \in \im{b_R(o_1)}$: since, for any $q' \in \criticalports{\taylor{e_1}{i}}{k}{j}$, we have $\arity{\taylor{e_1}{i}}(q') \geq k$, we obtain $q \in \conclusionscirc{T'}$, hence $q \in \criticalports{\taylor{e_1}{i}}{k}{j}$; by Proposition~\ref{prop: critical ports below new boxes}, we obtain $q \in \criticalports{\taylor{e}{i}}{k}{j}$.
\end{itemize}
\end{itemize}
\end{proof}

\subsection{$\bigcup_{j \in \mathcal{N}_i(e)} \nontrivialconnected{\taylor{e}{i+1}}{\criticalports{\taylor{e}{i}}{k}{j}}{k} \subseteq \bigcup_{j \in \mathcal{N}_i(e)} \nontrivialconnected{\taylor{e}{i}}{\criticalports{\taylor{e}{i}}{k}{j}}{k}$}\label{subsection: cond_sup_2}

We said also that we do not add new connected components (the new boxes are never caught by the elements of the $\nontrivialconnected{\taylor{e}{i+1}}{\criticalports{\taylor{e}{i}}{k}{\mathcal{N}_i(e)}}{k}$). It is the content of Lemma~\ref{lem: second condition}.

\begin{fact}\label{fact: coh[i]}
Let $R$ be a $\circ$-PS. Let $k > \Card{\boxes{R}{}}, \cosize{R}$. 
Let $e$ be a $k$-injective pseudo-experiment of $R$. Let $i \in \Nat$. 
Let $T \sqsubseteq \taylor{e}{i}, \taylor{e}{i+1}$. 
Let $p, p' \in \portsatzero{T}$ such that $p \coh_{\taylor{e}{i}} p'$.
 Then we have $p \coh_{\taylor{e}{i+1}} p'$.
\end{fact}

\begin{proof}
We distinguish between four cases:
\begin{itemize}
\item $p \in \wiresatzero{\taylor{e}{i}}$ and $p' = \target{\groundof{\taylor{e}{i}}}(p)$: we have $p' = \target{\groundof{\taylor{e}{i+1}}}(p)$;
\item $p' \in \wiresatzero{\taylor{e}{i}}$ and $p = \target{\groundof{\taylor{e}{i}}}(p')$: we have $p = \target{\groundof{\taylor{e}{i+1}}}(p')$;
\item $\{ p, p' \} \in \axioms{\groundof{\taylor{e}{i}}}$: we have $\{ p, p' \} \in \axioms{\groundof{\taylor{e}{i+1}}}$;
\item there exists $o \in \boxesatzerosmaller{\taylor{e}{i}}{i}$ such that $p, p' \in \im{b_{\taylor{e}{i}}(o)}$: by Fact~\ref{fact : boxes of taylor{e}{i+1} of depth < i}, we obtain $p \coh_{\taylor{e}{i+1}} p'$.
\end{itemize}
\end{proof}

\begin{fact}\label{fact: coh[i+1]}
Let $R$ be a $\circ$-PS. Let $k > \Card{\boxes{R}{}}, \cosize{R}$. 
Let $e$ be a $k$-injective pseudo-experiment of $R$. Let $i \in \Nat$. 
Let $T \sqsubseteq \taylor{e}{i}, \taylor{e}{i+1}$. 
Let $p, p' \in \portsatzero{T}$ such that $p \coh_{\taylor{e}{i+1}} p'$.
 Then we have $p \coh_{\taylor{e}{i}} p'$ or $p, p' \in \criticalports{\taylor{e}{i}}{k}{\mathcal{N}_i(e)}$.
\end{fact}

\begin{proof}
We distinguish between three cases:
\begin{itemize}
\item ($p \in \wiresatzero{\taylor{e}{i+1}}$ and $p' = \target{\groundof{\taylor{e}{i+1}}}(p)$) or ($p' \in \wiresatzero{\taylor{e}{i+1}}$ and $p = \target{\groundof{\taylor{e}{i+1}}}(p')$) or $\{ p, p' \} \in \axioms{\groundof{\taylor{e}{i+1}}}$: we have $p \coh_{\taylor{e}{i}} p'$;
\item there exists $o \in \boxesatzerosmaller{\taylor{e}{i+1}}{i}$ such that $p, p' \in \im{b_{\taylor{e}{i+1}}(o)}$: by  Fact~\ref{fact : boxes of taylor{e}{i+1} of depth < i}, we have $p \coh_{\taylor{e}{i}} p'$;
\item there exists $o \in \exactboxesatzero{\taylor{e}{i+1}}{i}$ such that $p, p' \in \im{b_{\taylor{e}{i+1}}(o)}$: by  Proposition~\ref{prop: critical ports below new boxes}, we obtain $p, p' \in \criticalports{\taylor{e}{i}}{k}{\mathcal{N}_i(e)}$.
\end{itemize}
\end{proof}

\begin{fact}\label{fact: critical ports for M => crtical ports for j}
Let $R$ be a $\circ$-PS. Let $k > \Card{\boxes{R}{}}, \cosize{R}$. 
Let $e$ be a $k$-injective pseudo-experiment of $R$. Let $i \in \Nat$. Let $j \in \mathcal{N}_i(e)$. Let $T \in \nontrivialconnected{\taylor{e}{i+1}}{\criticalports{\taylor{e}{i}}{k}{j}}{k}$ such that $\depthof{T} \leq i$. Then we have $\criticalports{\taylor{e}{i}}{k}{\mathcal{N}_i(e)} \cap \portsatzero{T} \subseteq \criticalports{\taylor{e}{i}}{k}{j}$.
\end{fact}

\begin{proof}
Let $p \in \criticalports{\taylor{e}{i}}{k}{\mathcal{N}_i(e)} \cap \portsatzero{T}$. By Proposition~\ref{prop: critical ports below new boxes}, there exists $o \in \exactboxesatzero{\taylor{e}{i+1}}{i}$ such that $p \in \im{b_{\taylor{e}{i+1}}(o)}$. We have $p \coh_{\taylor{e}{i+1}} o$, hence $p \notin \criticalports{\taylor{e}{i}}{k}{j} \Rightarrow o \in \portsatzero{T}$. But, since $\depthof{T} \leq i$, we have $o \notin \portsatzero{T}$.
\end{proof}

\begin{lem}\label{lem: second condition}
Let $R$ be a $\circ$-PS. Let $k > \Card{\boxes{R}{}}, \cosize{R}$. 
Let $e$ be a $k$-injective pseudo-experiment of $R$. Let $i \in \Nat$. 
For any $j \in  \mathcal{N}_i(e)$, we have $ \nontrivialconnected{\taylor{e}{i+1}}{\criticalports{\taylor{e}{i}}{k}{j}}{k} \subseteq \nontrivialconnected{\taylor{e}{i}}{\criticalports{\taylor{e}{i}}{k}{j}}{k}$.
\end{lem}

\begin{proof}
(Sketch) 
Let $j \in \mathcal{N}_i(e)$ and $T \in \nontrivialconnected{\taylor{e}{i+1}}{\criticalports{\taylor{e}{i}}{k}{j}}{k}$. 
By assumption we have $\portsatzero{T} \setminus \criticalports{\taylor{e}{i}}{k}{j} \not= \emptyset$.

First, notice that $\depthof{T} \leq i$. Indeed, assume that $\depthof{T} = i+1$. Let $o \in \exactboxesatzero{T}{i} = \exactboxesatzero{\taylor{e}{i+1}}{i} \cap \portsatzero{T}$. For any $p \in \portsatzero{T}$ such that $o \coh_T p$, we have $p \in \conclusions{T}$. Indeed, consider the case $o \in \exactboxesatzero{R}{i}$: we have $b_{\taylor{e}{i+1}}(o) = b_R(o)$; by Proposition~\ref{prop: critical ports below new boxes}, we have $\{ q \in \portsatzero{\taylor{e}{i}} ; q \in \im{b_R(o)} \} = \im{b_R(o)} = \criticalports{\taylor{e}{i}}{k}{j}$. This shows that $\portsatzero{T} = \conclusions{T}$, which contradicts $T \in \nontrivialconnected{\taylor{e}{i+1}}{\criticalports{\taylor{e}{i}}{k}{j}}{k}$. 

Now, using Fact~\ref{fact : boxes of taylor{e}{i+1} of depth < i}, we could easily show that $T \sqsubseteq \taylor{e}{i}$. 
So, by Fact~\ref{fact: coh[i+1]} and Fact~\ref{fact: critical ports for M => crtical ports for j}, we have $T \trianglelefteq_{\criticalports{\taylor{e}{i}}{k}{j}} \taylor{e}{i}$. 
Finally, by Fact~\ref{fact: coh[i]}, we obtain $T \in \nontrivialconnected{\taylor{e}{i}}{\criticalports{\taylor{e}{i}}{k}{j}}{k}$.
\end{proof}

\section{With axioms (Remark~\ref{remark: with axioms})}

With axioms, we need to slightly modify Definition~\ref{defin: experiment induces pseudo-experiment}, since different experiments can induce the same pseudo-experiment:

\begin{definition}\label{defin: experiment induces pseudo-experiment - with axioms}
Given an experiment $e$ of some differential $\circ$-PS $R$, we define, by induction on $\textit{depth}(R)$, a pseudo-experiment $\overline{e}$ of $R$ as follows: $\overline{e}(\emptysequence) = (R, 1)$ and 
$$\overline{e}(o) = \left\lbrace \overline{f}[\emptysequence \mapsto (B_R(o), i)] ; f \in \textit{Supp}(\boxes{e}(o)) \textit{ and } 1 \leq i \leq \sum_{\begin{array}{c} g \in \textit{Supp}(\boxes{e}(o))\\ \overline{g} = \overline{f} \end{array}} \mathcal{B}(e)(o)(g) \right\rbrace$$ for any $o \in \boxesatzero{R}$
\end{definition}

Notice that, if there is no axiom, Definitions~\ref{defin: experiment induces pseudo-experiment} and~\ref{defin: experiment induces pseudo-experiment - with axioms} induce the same pseudo-experiment $\overline{e}$ for an experiment $e$.

\section{Untyped framework (Remark~\ref{remark: untyped framework})}

Since there is no type, we define (differential) ground-structures \emph{via} the auxiliary definition of \emph{(differential) pre-ground-structures}:

\begin{definition}\label{defin: diff pre-ground-structure}
A \emph{differential pre-ground-structure} is a $6$-tuple $\mathcal{G} = (\mathcal{W}, \mathcal{P}, l, t, \mathcal{L}, \mathcal{A})$, where
\begin{itemize}
\item $\mathcal{P}$ is a finite set; the elements of $\ports{\mathcal{G}}$ are the \emph{ports of $\mathcal{G}$};
\item $\mathcal{W}$ is a subset of $\mathcal{P}$; the elements of $\wires{\mathcal{G}}$ are the \emph{wires of $\mathcal{G}$};
\item $l$ is a function $\mathcal{P} \to \typesoflinks$ such that $(\forall w \in \mathcal{W}) l(w) \not= \circ$; the element $l(p)$ of $\typesoflinks$ is the \emph{label of $p$ in $\mathcal{G}$};
\item $t$ is a function $\mathcal{W} \to \mathcal{P}$ such that, for any $p \in \mathcal{P}$, we have
\begin{itemize}
\item $l(p) \in \{ \tens, \parr \} \Rightarrow \Card{\{ w \in \mathcal{W} ; t(w) = p \}} = 2$;
\item $l(p) \in \{ \one, \bottom, \textit{ax} \} \Rightarrow \Card{\{ w \in \mathcal{W} ; t(w) = p \}} = 0$;
\end{itemize}
if $t(w) = p$, then $w$ is a \emph{premise of $p$};
\item $\mathcal{L}$ is a subset of $\{ w \in \mathcal{W}; l(t(w)) \in \{ \tens, \parr \} \}$ such that, for any $p \in \mathcal{P}$ such that $l(p) \in \{ \tens, \parr \}$, we have $\Card{\{ w \in \mathcal{L} ; t(w) = p \}} = 1$; if $w \in \mathcal{L}$ such that $t(w) = p$, we say that $w$ is a \emph{left premise of $p$;}
\item and $\mathcal{A}$ is a partition of $\{ p \in \mathcal{P} ; l(p) = \textit{ax} \}$ such that, for any $a \in \mathcal{A}$, $\Card{a} = 2$; the elements of $\mathcal{A}$ are the \emph{axioms of $\mathcal{G}$}.
\end{itemize}
We set 
$\wires{\mathcal{G}} = \mathcal{W}$, $\ports{\mathcal{G}} = \mathcal{P}$, $\labelofcell{\mathcal{G}} = l$, 
$\target{\mathcal{G}} = t$, $\leftwires{\mathcal{G}} = \mathcal{L}$, $\axioms{\mathcal{G}} = \mathcal{A}$ and $\conclusions{\mathcal{G}} = \mathcal{P} \setminus \mathcal{W}$. The elements of $\conclusions{\mathcal{G}}$ are the \emph{conclusions of $\mathcal{G}$}.

We set $\portsoftype{\one}{\mathcal{G}} = \{ p \in \mathcal{P}; l(p) = \one \}$, $\portsoftype{\bot}{\mathcal{G}} = \{ p \in \mathcal{P}; l(p) = \bot \}$, 
$\portsoftype{\cod}{\mathcal{G}} = \{ p \in \mathcal{P}; l(p) = \cod \}$, $\portsoftype{\contr}{\mathcal{G}} = \{ p \in \mathcal{P} ; l(p) = \contr \}$, $\portsoftype{\circ}{\mathcal{G}} = \{ p \in \mathcal{P} ; l(p) = \circ \}$, 
$\exponentialports{\mathcal{G}} = \portsoftype{\cod}{\mathcal{G}} \cup \portsoftype{\contr}{\mathcal{G}} \cup \portsoftype{\circ}{\mathcal{G}}$, $\portsoftype{\tens}{\mathcal{G}} = \{ p \in \mathcal{P}; l(p) = \tens \}$, $\portsoftype{\parr}{\mathcal{G}} = \{ p \in \mathcal{P}; l(p) = \parr \}$ and $\multiplicativeports{\mathcal{G}} = \portsoftype{\tens}{\mathcal{G}} \cup \portsoftype{\parr}{\mathcal{G}}$.

A \emph{pre-ground-structure} is a differential pre-ground-structure $\mathcal{G}$ such that $\im{\target{\mathcal{G}}} \cap (\portsoftype{\cod}{\mathcal{G}} \cup \portsoftype{\circ}{\mathcal{G}}) = \emptyset$.

A \emph{differential ground-structure} (resp. a \emph{ground-structure}) is a differential pre-ground structure (resp. a pre-ground structure) $\mathcal{G}$ such that the reflexive transitive closure $<_\mathcal{G}$ of the binary relation $<$ on $\ports{\mathcal{G}}$ defined by $p < p' \textit{ iff } p = \target{\mathcal{G}}(p')$ is antisymmetric.
\end{definition}

For the semantics of PS's, we are given a set $A$ that does not contain any couple nor any $3$-tuple and such that $\ast \notin A$. We define, by induction on $n$, the set $D_{A, n}$ for any $n \in \Nat$:
\begin{itemize}
\item $D_{A, 0} = \{ +, - \} \times (A \cup \{ \ast \})$
\item $D_{A, n+1} = D_{A, 0} \cup (\{ +, - \} \times D_{A, n} \times D_{A, n}) \cup (\{ +, - \} \times \finitemultisets{D_{A, n}})$
\end{itemize}
We set $D_A = \bigcup_{n \in \Nat} D_{A, n}$.

Definition~\ref{defin: experiment in an untyped framework} is an adaptation of Definition~\ref{defin: experiment} in an untyped framework.

\begin{definition}
For any $\alpha \in D_A$, we define $\alpha^\perp \in D_A$ as follows:
\begin{itemize}
\item if $\alpha \in A$ and $\delta \in \{ +, - \}$, then $(\delta, \alpha)^\perp = (\delta^\perp, \alpha)$;
\item if $\alpha = (\delta, \ast)$ with $\delta \in \{ +, - \}$, then $\alpha^\perp = (\delta^\perp, \ast)$;
\item if $\alpha = (\delta, \alpha_1, \alpha_2)$ with $\delta \in \{ +, - \}$ and $\alpha_1, \alpha_2 \in D_A$, then $\alpha^\perp = (\delta^\perp, {\alpha_1}^\perp, {\alpha_2}^\perp)$;
\item if $\alpha = (\delta, [\alpha_1, \ldots, \alpha_m])$ with $\delta \in \{ +, - \}$ and $\alpha_1, \ldots, \alpha_m \in D_A$, then $\alpha^\perp = (\delta^\perp, [{\alpha_1}^\perp, \ldots, {\alpha_m}^\perp])$;
\end{itemize}
where $+^\perp = -$ and $-^\perp = +$.
\end{definition}

\begin{definition}\label{defin: experiment in an untyped framework}
For any differential $\circ$-PS $R$, we define, by induction on $\textit{depth}(R)$ the set of \emph{experiments of $R$}: it is is the set of triples $(R, e_\mathcal{P}, e_\mathcal{B})$, where $e_\mathcal{P}$ is a function $\portsatzero{R} \to D_A \cup \finitemultisets{D_A}$ and $e_\mathcal{B}$ is a function which associates to every $o \in \boxesatzero{R}$ a finite multiset of experiments of $B_R(o)$ such that
\begin{itemize}
\item for any $\{ p, q \} \in \axioms{\groundof{R}}$, we have $e_\mathcal{P}(p) = \alpha$, $e_\mathcal{P}(q)  = \alpha^\perp$ for some $\alpha \in D_A$;
\item for any $p \in \portsatzerooftype{\tens}{R}$ (resp. $p \in \portsatzerooftype{\parr}{R}$), for any $w_1, w_2 \in \wiresatzero{R}$ such that $\target{\groundof{R}}(w_1) = p = \target{\groundof{R}}(w_2)$, $w_1 \in \leftwires{\groundof{R}}$ and $w_2 \notin \leftwires{\groundof{R}}$, we have $e_\mathcal{P}(p) = (+, e_\mathcal{P}(w_1), e_\mathcal{P}(w_2))$ (resp. $e_\mathcal{P}(p) = (-, e_\mathcal{P}(w_1), e_\mathcal{P}(w_2))$);
\item for any $p \in \portsatzerooftype{\one}{R}$ (resp. $p \in \portsatzerooftype{\bot}{R}$), we have $e_\mathcal{P}(p) = (+, \ast)$ (resp. $e_\mathcal{P}(p) = (-, \ast)$);
\item for any $p \in \exponentialportsatzero{R}$,
we have $e(p) = \left\lbrace \begin{array}{ll} 
a & \textit{if $p \in \conclusionscirc{R}$;}\\
(-, a) & \textit{if $p \in \portsatzerooftype{\contr}{R}$;}\\
(+, a) & \textit{if $p \in \portsatzerooftype{\cod}{R}$;}
\end{array} \right.$
where 
\begin{eqnarray*}
a & = & \sum_{\begin{array}{c} w \in \wiresatzero{R}\\ \target{\groundof{R}}(w) = p \end{array}} [e_\mathcal{P}(w)] +\\ 
&  & \sum_{\begin{array}{c} o \in \boxesatzero{R} \end{array}} \sum_{e' \in \supp{e_\mathcal{B}(o)}}  ( \sum_{\begin{array}{c} q \in \conclusionsnotcirc{B_R(o)}\\ b_R(o)(q) = p \end{array}} e_\mathcal{B}(o)(e') \cdot [{e'}_\mathcal{P}(q)] + \sum_{\begin{array}{c} q \in \conclusionscirc{B_R(o)}\\ b_R(o)(q) = p \end{array}} e_\mathcal{B}(o)(e') \cdot {e'}_\mathcal{P}(q) ) 
\end{eqnarray*}
\end{itemize}
For any experiment $e = (R, e_\mathcal{P}, e_\mathcal{B})$, we set $\ports{e} = e_\mathcal{P}$ and $\boxes{e} = e_\mathcal{B}$.

For any differential $\circ$-PS $R$, we set $\sm{R}_A = \{ \restriction{\ports{e}}{\conclusions{R}} ; e \textit{ is an experiment of } R \}$.
\end{definition}

\begin{definition}\label{defin: injective}
Let $r \in \finitemultisets{D_A}$. We say that $r$ is \emph{injective} if, for every $\gamma \in A$, there are at most two occurrences of $\gamma$ in $r$. 

For any set $\mathcal{P}$, for any function $x : \mathcal{P} \to D_A$, we say that $x$ is \emph{injective} if $\sum_{p \in \mathcal{P}} [x(p)]$ is injective.

An experiment $e$ of a differential $\circ$-PS $S$ is said to be \emph{injective} if $\restriction{\ports{e}}{\conclusions{R}}$ is injective.
\end{definition}

\begin{definition}
Let $\sigma : A \to D_A$. For any $\alpha \in D_A$, we define $\sigma \cdot \alpha \in D_A$ as follows:
\begin{itemize}
\item if $\alpha \in A \cup \{ \ast \}$, then $\sigma \cdot (+, \alpha) = \sigma(\alpha)$ and $\sigma \cdot (-, \alpha) = \sigma(\alpha)^\perp$;
\item if $\delta \in \{ +, - \}$ and $\alpha_1, \alpha_2 \in D_A$, then $\sigma \cdot (\delta, \alpha_1, \alpha_2) = (\delta, \sigma \cdot \alpha_1, \sigma \cdot \alpha_2)$;
\item if $\delta \in \{ +, - \}$ and $\alpha_1, \ldots, \alpha_m \in D_A$, then $\sigma \cdot (\delta, [\alpha_1, \ldots, \alpha_m]) = (\delta, [\sigma \cdot \alpha_1, \ldots, \sigma \cdot \alpha_m])$.
\end{itemize}
For any set $\mathcal{P}$, for any function $x : \mathcal{P} \to D_A$, we define a function $\sigma \cdot x : \mathcal{P} \to D_A$ by setting: $(\sigma \cdot x)(p) = \sigma \cdot x(p)$ for any $p \in \mathcal{P}$.
\end{definition}

\begin{rem}
For any functions $\sigma, \sigma' : A \to D_A$, for any function $x : \mathcal{P} \to D_A$, we have $\sigma \cdot (\sigma' \cdot x) = (\sigma \cdot \sigma') \cdot x$.
\end{rem}

\begin{definition}
Let $S$ be a differential $\circ$-PS. Let $e$ be an experiment of $S$. Let $\sigma : A \to D_A$. We define, by induction of $\depthof{S}$, an experiment $\sigma \cdot e$ of $S$ by setting 
\begin{itemize}
\item $\ports{\sigma \cdot e} = \sigma \cdot \ports{e}$
\item $\boxes{\sigma \cdot e}(o) = \sum_{e_1 \in \supp{\boxes{e}(o_1)}} \boxes{e}(o_1)(e_1) \cdot [\sigma \cdot e_1]$ for any $o_1 \in \boxesatzero{S}$.
\end{itemize}
\end{definition}

Since we deal with untyped proof-nets, we cannot assume that the proof-nets are $\eta$-expanded and that experiments label the axioms only by atoms. That is why we introduce the notion of \emph{atomic experiment}:

\begin{definition}
For any differential $\circ$-PS $R$, we define, by induction on $\depthof{R}$, the set of \emph{atomic experiments of $R$}: it is the set of experiments $e$ of $R$ such that
\begin{itemize}
\item for any $\{ p, q \} \in \axioms{\groundof{R}}$, we have $\ports{e}(p), \ports{e}(q) \in \{ +, - \} \times A$;
\item and, for any $o_1 \in \boxesatzero{R}$, the multiset $\boxes{e}(o_1)$ is a multiset of atomic experiments of $B_R(o_1)$.
\end{itemize}
\end{definition}

\begin{fact}\label{fact: an experiment of a flat-pS induces an experiment of differential flat-PS}
Let $R$ be a $\circ$-PS. Let $e$ be an experiment of $R$. If $e$ is atomic, then $\mathcal{T}(e)$ is atomic.
\end{fact}

\begin{proof}
By induction on $\depthof{R}$.
\end{proof}

\begin{definition}
Let $\mathcal{P}$ be a set. Let $\mathcal{D} \subseteq {(D_A)}^\mathcal{P}$. 
Let $x \in \mathcal{D}$, we say that $x$ is \emph{$\mathcal{D}$-atomic} if we have 
$$(\forall \sigma \in {(D_A)}^A) (\forall y \in \mathcal{D}) (\sigma \cdot y = x \Rightarrow (\forall \gamma \in \atoms{y}) \sigma(\gamma) \in A)$$
where $\atoms{y}$ is the set of atoms occurring in $\im{y}$.

We denote by $\mathcal{D}_\textit{At}$ the subset of $\mathcal{D}$ consisting of the $\mathcal{D}$-atomic elements of $\mathcal{D}$.
\end{definition}

For any PS $R$, any $\sm{R}_A$-atomic injective point is the result of some atomic experiment of $R$:

\begin{fact}\label{fact: from point to experiment}
Let $R$ be a $\circ$-PS. Let $x \in {(\sm{R}_A)}_{\textit{At}}$ injective. Then there exists an atomic experiment $e$ of $R$ such that $\restriction{e}{\conclusions{R}} = x$.
\end{fact}

\begin{proof}
We prove, by induction on $(\depthof{R}, \Card{\portsatzero{R}})$ lexicographically ordered, that, for any non-atomic injective experiment $e'$ of $R$, there exist an experiment $e$ of $R$, a function $\sigma : A \to D_A$ such that $\sigma \cdot e = e'$ and $\gamma \in \atoms{\restriction{e}{\conclusions{R}}}$ such that $\sigma(\gamma) \notin A$.
\end{proof}

The converse does not necessarily hold: for some PS's $R$, there are results of atomic injective experiments of $R$ that are not $\sm{R}_A$-atomic. Indeed, consider Figure~\ref{fig: atomic experiments}. 
\begin{figure}[!t]
\centering
\scalebox{\scalefactfour}{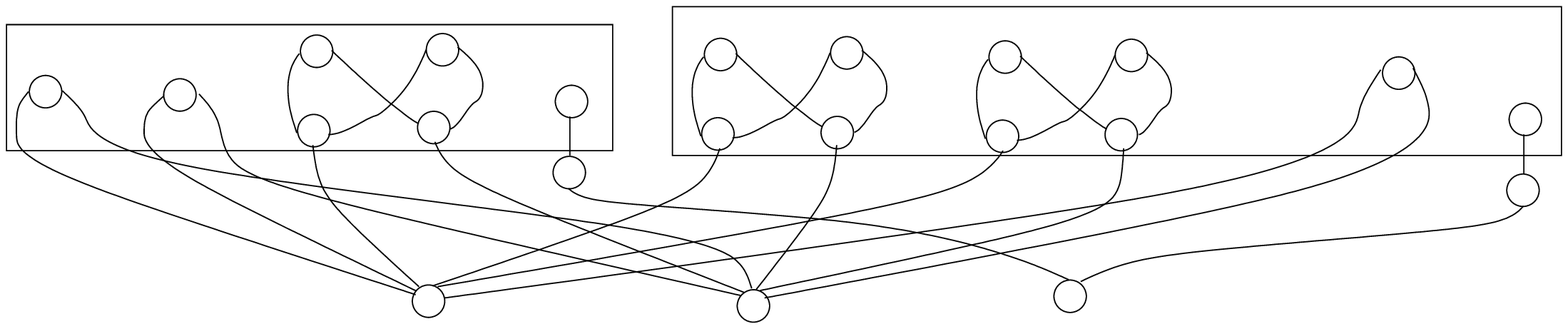}
\caption{PS $R''$}
\label{fig: atomic experiments}
\end{figure}
There exists an atomic injective experiment $e$ of $R''$ such that 
\begin{itemize}
\item $\ports{e}(p_1) = (-, [(+, \gamma_1), \ldots, (+, \gamma_7), (+, (+, \gamma_{8}), (+, \gamma_{9})), \ldots, (+, (+, \gamma_{22}), (+, \gamma_{23}))])$, 
\item $\ports{e}(p_2) = (-, [(-, \gamma_1), \ldots, (-, \gamma_7), (-, (-, \gamma_{8}), (-, \gamma_{9})), \ldots, (-, (-, \gamma_{22}), (-, \gamma_{23}))])$
\item and $\ports{e}(p_3) = (-, [(+, [(+, \ast), (+, \ast)]), (+, [(+, \ast), (+, \ast), (+, \ast)]) ])$,
\end{itemize}
where $\{ \gamma_1, \ldots, \gamma_{23} \} \subseteq A$. But $\restriction{e}{\{ p_1, p_2, p_3 \}}$ is not in ${(\sm{R''}_A)}_{\textit{At}}$: there exists an atomic injective experiment $e'$ of $R'$ such that 
\begin{itemize}
\item $\ports{e'}(p_1) = (-, [(+, \gamma_1), \ldots, (+, \gamma_8), (+, (+, \gamma_{10}), (+, \gamma_{11})), \ldots, (+, (+, \gamma_{22}), (+, \gamma_{23}))])$, 
\item $\ports{e'}(p_2) = (-, [(-, \gamma_1), \ldots, (-, \gamma_8), (-, (-, \gamma_{10}), (-, \gamma_{11})), \ldots, (-, (-, \gamma_{22}), (-, \gamma_{23}))])$ 
\item and $\ports{e'}(p_3) = (-, [(+, [(+, \ast), (+, \ast)]), (+, [(+, \ast), (+, \ast), (+, \ast)]) ])$;
\end{itemize}
we set $\sigma(\gamma) = \left\lbrace \begin{array}{ll} \gamma & \textit{if $\gamma \in A \setminus \{ \gamma_8 \}$;}\\ (+, (+, \gamma_8), (+, \gamma_9)) & \textit{if $\gamma = \gamma_8$;} \end{array} \right.$ - we have $\sigma \cdot \restriction{e'}{\{ p_1, p_2, p_3 \}} = \restriction{e}{\{ p_1, p_2, p_3 \}}$.

But it does not matter, because there are many enough atomic points:

\begin{fact}\label{fact: atomic are enough}
Let $R$ be a $\circ$-PS. For any $y \in \sm{R}_A$, there exist $x \in {(\sm{R}_A)}_\textit{At}$ and $\sigma : A \to D_A$ such that $\sigma \cdot x = y$. 
\end{fact}

\begin{proof}
By induction on $\size{\sum_{p \in \conclusions{R}} [y(p)]}$, where $\size{r} \in \Nat$ is defined for any $r \in \finitemultisets{D_A}$ as follows: $\size{r} = \sum_{\alpha \in \supp{r}} r(\alpha) \cdot \min \{ i \in \Nat ; \alpha \in D_{A, i} \}$: if $y \in {[\sm{R}_A)}_\textit{At}$, then we can set $x = y$ and $\sigma = id_A$; if $y \notin {(\sm{R}_A)}_{\textit{At}}$, then there exist a function $\sigma' : A \to D_A$, $y' \in \sm{R}_A$ such that $\sigma' \cdot y' = y$ and $\gamma \in \atoms{y'}$ such that $\sigma'(\gamma) \notin A$, hence $\size{\sum_{p \in \conclusions{R}} [y'(p)]} < \size{\sum_{p \in \conclusions{R}} [y(p)]}$. By induction hypothesis, there exist $x \in {(\sm{R}_A)}_\textit{At}$ and $\sigma'' : A \to D_A$ such that $\sigma'' \cdot x = y'$. We set $\sigma = \sigma' \cdot \sigma''$: we have $\sigma \cdot x = (\sigma' \cdot \sigma'') \cdot x = \sigma' \cdot (\sigma'' \cdot x) = \sigma' \cdot y' = y$.
\end{proof}

\end{document}